\PassOptionsToPackage{hypertexnames=false}{hyperref}
\documentclass{theoretics}

\usepackage{ogonek}

\usepackage{xspace}
\usepackage{amsthm}
\usepackage{thm-restate}
\usepackage[capitalise,noabbrev,nameinlink]{cleveref}

\newcommand{\R}{\ensuremath{\mathbb R}\xspace}
\newcommand{\N}{\ensuremath{\mathbb N}\xspace}
\newcommand{\ER}{\ensuremath{\exists \mathbb{R}}\xspace}
\newcommand{\NP}{\ensuremath{\textrm{NP}}\xspace}
\newcommand{\PSPACE}{\ensuremath{\textrm{PSPACE}}\xspace}

\newcommand{\etr}{{\sc{ETR}}\xspace}
\newcommand{\etrinv}{{\sc{ETR-Inv}}\xspace}
\newcommand{\rangeetrinv}{{\sc{Range-ETR-Inv}}\xspace}
\newcommand{\fgetr}{{\sc{Curve-ETR}}$[f,g]$\xspace}
\newcommand{\wiredinv}{{\sc{Wired-}}\fgetr}

\newcommand{\area}{\ensuremath{\textrm{area}}\xspace}
\newcommand{\diam}{\ensuremath{\textrm{diam}}\xspace}

\newcommand{\pack}[3]{{\sc{Pack}}$[ #1 \rightarrow #2 , #3 ]$\xspace}

\newcommand{\I}{\ensuremath{\mathcal{I}}\xspace}
\newcommand{\p}{\ensuremath{\mathbf{p}}\xspace}

\newcommand{\s}{\ensuremath{\mathbf{s}}\xspace}
\newcommand{\pl}[2]{{#1}^{#2}}
\newcommand{\m}{\ensuremath{\mathbf{m}}\xspace}
\newcommand{\disk}{\ensuremath{\mathrm{disk}}\xspace}
\newcommand{\slack}{\ensuremath{\mathrm{\mu}}\xspace}
\newcommand{\eps}{\ensuremath{\varepsilon}\xspace}
\newcommand{\mydef}{:=}
\newcommand{\inn}[1]{#1^{\text{in}}}
\newcommand{\out}[1]{#1^{\text{out}}}
\newcommand{\polQ}{Q} % {\mathcal Q}
\newcommand{\cont}{C}
\newcommand{\sqcont}{S}
\newcommand{\piecetype}{\ensuremath{\mathcal P}}
\newcommand{\conttype}{\ensuremath{\mathcal C}}
\newcommand{\motiontype}{\ensuremath{\mathcal M}}
\newcommand{\gadgets}{g}

\newcommand{\enc}[1]{\ensuremath{\left\langle \, #1 \,\right\rangle}\xspace}

\newcommand{\rotation}
{\includegraphics[page=1,scale =0.6]{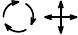}}
\newcommand{\translation}
{\includegraphics[page=2,scale =0.6]{figures/Symbols.pdf}}

\newcommand{\polygon}
{\includegraphics[page=3,scale =0.6]{figures/Symbols.pdf}\xspace}
\newcommand{\curved}
{\includegraphics[page=7,scale =0.6]{figures/Symbols.pdf}\xspace}
\newcommand{\convexcurved}
{\includegraphics[page=5,scale =0.6]{figures/Symbols.pdf}\xspace}
\newcommand{\convexpolygon}
{\includegraphics[page=4,scale =0.6]{figures/Symbols.pdf}\xspace}
\newcommand{\rectangle}
{\includegraphics[page=8,scale =0.6]{figures/Symbols.pdf}\xspace}
\newcommand{\disksymbol}
{\includegraphics[page=12,scale =0.6]{figures/Symbols.pdf}\xspace}

\title{Framework for \texorpdfstring{$\exists\mathbb R$}{ER}-Completeness of Two-Dimensional Packing Problems}

\ThCSshortnames{M.~Abrahamsen, T.~Miltzow, and N.~Seiferth}
 \ThCSshorttitle{Framework for \texorpdfstring{$\exists\mathbb R$}{ER}-Completeness of Two-Dimensional~Packing~Problems}
\ThCSauthor[1]{Mikkel Abrahamsen}{miab@di.ku.dk}[https://orcid.org/0000-0003-2734-4690]

\ThCSauthor[2]{Tillmann Miltzow}{t.miltzow@uu.nl}[https://orcid.org/0000-0003-4563-2864]

\ThCSauthor[3]{Nadja Seiferth}{nadja.seiferth@fu-berlin.de}[https://orcid.org/0000-0003-1462-006X]

\ThCSthanks{The paper was presented at the 61st Annual IEEE Symposium on Foundations of Computer Science (FOCS 2020) \cite{abrahamsen2020framework}. Mikkel Abrahamsen is supported by Starting Grant 1054-00032B from the Independent Research Fund Denmark under the Sapere Aude research career programme. BARC is supported by the VILLUM Foundation grant 16582. Tillmann Miltzow is supported by the Netherlands Organisation for Scientific Research (NWO) under project no.\ 016.Veni.192.25. Nadja Seiferth is supported by the German Science Foundation within the collaborative DACH project \emph{Arrangements and Drawings} as DFG Project MU3501/3-1.}

\ThCSaffil[1]{Basic Algorithms Research Copenhagen (BARC), University of Copenhagen.}
\ThCSaffil[2]{Utrecht University.}
\ThCSaffil[3]{Freie Universität Berlin.}
\ThCSkeywords{Packing, Polygons, \ER-completeness}

\ThCSyear{2024}
\ThCSarticlenum{11}
\ThCSdoicreatedtrue
\ThCSreceived{Apr 21, 2022}
\ThCSaccepted{Feb 15, 2024}
\ThCSpublished{Apr 26, 2024}
\begin{document}
%\pagenumbering{gobble}
%\date{February 28, 2024}
\maketitle
\begin{abstract}
The aim in packing problems is to decide if a given set of pieces can be placed inside a given container.
A packing problem is defined by the types of pieces and containers to be handled, and the motions that are allowed to move the pieces.
The pieces must be placed so that in the resulting placement, they are pairwise interior-disjoint.
We establish a framework which enables us to show that for many combinations of allowed pieces, containers and motions, the resulting problem is $\ER$-complete.
This means that the problem is equivalent (under polynomial time reductions) to deciding whether a given system of polynomial equations and inequalities with integer coefficients has a real solution.

We consider packing problems where only translations are allowed as the motions, and problems where arbitrary rigid motions are allowed, i.e., both translations and rotations.
When rotations are allowed, we show that it is an $\ER$-complete problem to decide if a set of convex polygons, each of which has at most $7$ corners, can be packed into a square.

Restricted to translations, we show that the following problems are $\ER$-complete:
(i) pieces bounded by segments and hyperbolic curves to be packed in a square, and
(ii) convex polygons to be packed in a container bounded by segments and hyperbolic curves.
\end{abstract}

\begin{figure}[h]
\centering
\includegraphics[page =1]{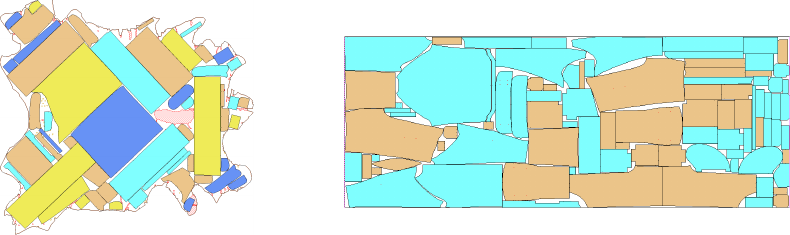}
\caption{Real examples of nesting on a leather hide (left) and a piece of fabric (right), kindly provided by MIRISYS and produced by their software for automatic nesting, \protect\url{https://www.mirisys.com/}.
}
\label{fig:DailyLife}
\end{figure}

%%%%%%%%%%%%%%%%%%%%%%%%%%%%%%
\section{Introduction}
\label{sec:Intro}
%%%%%%%%%%%%%%%%%%%%%%%%%%%%%%
Packing problems are everywhere in our daily lives.
To give a few examples, you solve packing problems when finding room for your tupperware in the kitchen, the manufacturer of your clothing arranges cutting patterns on a large piece of fabric in order to minimize waste, and at Christmas time you are trying to cut out as many cookies from a dough as you can.
In a large number of industries, it is crucial to solve complicated instances of packing problems efficiently.
In addition to clothing manufacturing, we mention leather, glass, wood and sheet metal cutting, selective laser sintering, shipping (packing goods in containers) and 3D printing (arranging the parts to be printed in the printing volume); see \Cref{fig:DailyLife}.

Packing problems can be easily and precisely defined in a mathematical manner, but many important questions are still completely elusive.
In this work, we uncover a fundamental aspect of many versions of geometric packing by settling their computational difficulty.

We denote \pack \piecetype \conttype \motiontype\ as the packing problem
with pieces of the type \piecetype, containers
of type $\conttype$ and motions of type $\motiontype$.
In an instance of \pack{\piecetype}{\conttype}{\motiontype}, we are
given pieces $p_1,\ldots,p_n$ of type $\piecetype$ and a container $\cont$ of type $\conttype$.
We want to decide if there is a motion of type~$\motiontype$ for each piece such that after moving the pieces by these motions, each piece is in $\cont$ and the pieces are pairwise interior-disjoint.
Such a placement of the pieces is called a \emph{valid} placement.

As the allowed motions, we consider \emph{translations} (\translation) and \emph{rigid motions} (\rotation), where a rigid motion is a combination of a translation and a rotation.
As containers and pieces, we consider squares ($\square$), convex polygons (\convexpolygon), simple polygons (\polygon) and curved polygons (\curved), where a \emph{curved polygon} is a region enclosed by a simple closed curve consisting of a finite number of line segments and arcs contained in hyperbolae (such as the graph of $y=1/x$).
Note that hyperbolae, like all conic sections, can be represented as rational quadratic B\'{e}zier curves~\cite{FANG2002297}, which are extensively used for computer-aided design and manufacturing.
It would therefore not be uncommon that curved pieces as the ones used here could appear in practical settings.

The problems with only translations allowed are relevant to some industries; for instance when arranging cutting patterns on a roll of fabric for clothing production, where the orientation of each piece has to follow the orientation of the weaving or a pattern printed on the fabric.
In other contexts such as leather, glass, or sheet metal cutting, there are usually no such restrictions, so rotations can be used to minimize waste.
As can be seen from \Cref{fig:DailyLife}, it is relevant to study packing problems where the pieces as well as the containers may be non-convex and have boundaries consisting of many types of curves (not just straight line segments).

We show that many of the above mentioned variants of packing are $\ER$-complete.
The complexity class $\ER$ will be defined below. 
We call the techniques we developed a \emph{framework}, since the same techniques turn out to be applicable to prove hardness for many versions of packing.
With adjustments or additions, they can likely be used for other versions or proofs of other types of hardness as well.

\subsubsection*{The Existential Theory of the Reals}
The term \emph{Existential Theory of the Reals} refers 
ambiguously to a formal language, 
a corresponding algorithmic problem (\etr) and a complexity class ($\ER$).
Let us start with the formal logic.
Let
\[\Sigma \mydef \{\forall, \exists, 0,1,x_1,\ldots, x_n, +, \cdot, =, \leq , < , \land
, \lor, \lnot\} \] be an alphabet for some $n\geq 1$.
A \emph{sentence} over $\Sigma$ is a well-formed formula with no free variables, i.e., so that every variable is bound to a quantifier.
The \emph{Existential Theory of the Reals} is the set of true sentences of the form
\[\exists x_1,\ldots,x_n\; \Phi(x_1,\ldots,x_n),\]
where $\Phi$ is a quantifier-free formula.
The algorithmic problem \etr is to decide whether a sentence of this form is true or not.
At last, this leads us to the complexity class \emph{Existential Theory of the Reals} (\ER),
which consists of all those languages that are
many-one reducible to \etr in polynomial time.
Given a quantifier-free formula $\Phi$, we define
the solution space of $\Phi$ as
$V(\Phi) \mydef \{\mathbf x\in\R^n : \Phi(\mathbf x)\}$.
Thus in other words, \etr is to decide 
if $V(\Phi)$ is empty or not.
It is currently known that
\begin{align}
\label{eq:inclusions}
\NP \subseteq \ER \subseteq \PSPACE .
\end{align}
To show the first inclusion is an easy exercise, whereas the second is non-trivial and was first proven by Canny~\cite{canny1988some}.
A problem $P$ is \emph{$\ER$-hard} if \etr is many-one reducible to $P$ in polynomial time, and $P$ is \emph{$\ER$-complete} if $P$ is $\ER$-hard and in $\ER$. 
None of the inclusions~\eqref{eq:inclusions} are known to be strict, but the first is widely believed to be~\cite{BIEDL2020105868}, implying that the $\ER$-hard problems are not in NP.
As examples of \ER-complete problems, we mention problems related to realization of order-types~\cite{richter1995realization, mnev1988universality, shor1991stretchability}, graph drawing~\cite{bienstock1991some, AreasKleist, LindaPHD, AnnaPreparation}, recognition of geometric graphs~\cite{cardinal2017intersection, cardinal2017recognition,kang2011sphere,mcdiarmid2013integer}, straightening of curves~\cite{erickson2019optimal},
the art gallery problem~\cite{ARTETR},
minimum convex covers~\cite{abrahamsen2021covering},
Nash-equilibria~\cite{berthelsen2019computational,garg2015etr},
linkages~\cite{abel, schaefer2013realizability, Schaefer-ETR},
matrix-decompositions~\cite{NestedPolytopesER,shitov2016universality,Shitov16a}, 
polytope theory~\cite{richter1995realization}, embedding of simplicial complexes~\cite{GeometricEmbedding} and training neural networks~\cite{abrahamsen2021training,TrainFull2NN}.
See also the surveys~\cite{CardinalSurvey, matousek2014intersection, Schaefer2010}.

\subsubsection*{\ER-membership}
Showing that the packing problems we are dealing with in this paper are contained in \ER is easy using the following recent result. 
\begin{theorem}[Erickson, Hoog, Miltzow~\cite{RobustComputation}]
For any decision problem $P$, there is a real verification algorithm for $P$ if
and only if $P \in \ER$.
\end{theorem}
A \emph{real verification algorithm} is like a verification algorithm for a problem in NP with the additional feature that it accepts real inputs for the witness and runs on the real RAM.
(We refer to~\cite{RobustComputation} for the full definition, as it is too long to include here.)

Thus in order to show that our packing problems lie in \ER, we have to specify a witness and a real verification algorithm. 
The witness is simply the motions that move the pieces to a valid placement.
The verification algorithm checks that the pieces are pairwise interior-disjoint and contained in the container.
Note that without the theorem above, we would need to describe an \etr-formula equivalent to a given packing instance in order to show \ER-membership.
Although this is not difficult for packing, it would still require some work.

\begin{table}
\centering
\includegraphics[page = 4,scale = 1.5]{figures/intro.pdf}
\caption{
This table displays $6$~variants of the 
packing problem with rotations and translations, and~$6$ with translations only.
$\ER$ means $\ER$-complete and NP means NP-complete.
We show that $10$ of the problems are \ER-complete, and the remaining $2$ are known to be $\NP$-complete.
The problems marked with * are the \emph{basic} problems.
The $\ER$-completeness of the remaining problems follows since there is a basic problem in the table which is more restricted.
}
\label{fig:Results}
\end{table}

\subsubsection*{Results}
We show that a wide range of two-dimensional packing problems are \ER-complete.
A compact overview of our results is displayed in \Cref{fig:Results}.
In the table, the second row (with problems \pack \piecetype\polygon\motiontype) is in some sense redundant, since the $\ER$-completeness results can be deduced from the more restricted third row (the problems \pack \piecetype\square\motiontype).
We anyway include the row since a majority of our reduction is to establish hardness of problems with polygonal containers, and only later we reduce these problems to the case where the container is a square.

A strength of our reductions is that in the resulting constructions, all corners can be described with rational coordinates that require a number of bits only logarithmic in the total number of bits used to represent the instance.
Therefore, we show that the problems are \emph{strongly} $\ER$-hard.
Another strength is that all the pieces 
have constant complexity, i.e., each piece can be described by its boundary as a union of $O(1)$ straight line segments and arcs contained in hyperbolae.

In the following we sketch an argument why \pack\polygon\polygon\translation\ is in $\NP$.
This may be folklore; the second author learned the proof from Günter Rote.
We show that a valid placement can be specified as the translations of the pieces represented by a number of bits polynomial in the input size. 
Consider a valid placement of the pieces.
For each pair of a segment $s$ and a corner $c$ (of a piece or the container), we consider the line $\ell(s)$ containing $s$ and note which of the closed half-planes bounded by $\ell(s)$ contains $c$.
Then we build a linear program (LP) using that information in the natural way.
Here, the translation of each piece is described by two variables and for each pair $(s,c)$, we have one constraint involving at most four variables, enforcing $c$ to be on the correct side of $\ell(s)$.
It is easy to verify that the constraint is linear.
The solution of the LP gives a valid placement of every piece and as the LP is polynomial in the input, so is the number of the bits of the solution to the LP. 
Note that if rotations are included, the corresponding constraints become non-linear.

Our results show that some packing problems with rotations or non-polygonal features allowed are \ER-hard and thus likely not in \NP.
This gives a confirmation to the operations research community that most likely, they cannot employ standard algorithm techniques (solvers for ILP and SAT, etc.)~that work well for many NP-complete problems like scheduling and TSP.
The main message for the theory community is that efficient algorithms dealing with rotations or non-polygonal shapes can probably only be found if we relax the problems considerably.

\subsubsection*{Basis problem}
In the first version of this paper~\cite{abrahamsen2020framework}, we could not show \ER-hardness of the problem \pack\convexpolygon\square\rotation.
We introduced the problem \rangeetrinv, which in turn was a restricted version of the problem \etrinv used to prove \ER-hardness of the art gallery problem~\cite{ARTETR}.
When proving \ER-hardness by reducing from these problems, one must create gadgets for simulating \emph{inversion} constraints of the form $xy=1$, or, equivalently, $xy-1 = 0$.
This is usually obtained by making gadgets for each of the inequalities $xy-1 \geq 0$ and $-xy+1 \geq 0$.
We managed to make a gadget for the constraint $xy-1 \geq 0$ using convex pieces only, but despite much effort, we were unable to realize $-xy+1 \geq 0$ unless we introduced non-convex pieces.
Note that the equation $xy-1 = 0$ defines a convex curve while $-xy+1=0$ defines a concave curve.
In the recent paper~\cite{miltzow2021classifying}, Miltzow and Schmiermann proved that it is not important to make a gadget realizing the inversion constraint exactly, but that it suffices to make gadgets realizing constraints of the forms $f(x,y)\geq 0$ and $g(x,y)\geq 0$ for any sufficiently well-behaved functions $f$ and $g$, where one is convex and the other is concave.
This is the basis for the reduction in this paper.
We now specify the details of the problem we are reducing from.

\begin{definition}[\fgetr formula]
Let $\delta>0$, $U\mydef[-\delta,\delta]$ and $f,g: U^2\longrightarrow\mathbb R$.
An \emph{\fgetr formula} $\Phi=\Phi(x_1,\ldots,x_n)$ is a conjunction
\[
\bigwedge_{i=1}^m C_i,
\]
where each constraint $C_i$ has one of the forms
\begin{align*}
x \geq 0, \quad x = \delta, \quad x+y=z,\quad f(x,y)\geq 0,\quad g(x,y) \geq 0,
\end{align*}
for $x,y,z \in \{x_1, \ldots, x_n\}$.
Each constraint of the form $f(x,y)\geq 0$ or $g(x,y) \geq 0$ is called a \emph{curved} constraint.
\end{definition}

\begin{definition}[Well-behaved and convexly/concavely curved function]\label{def:wellbehaved}
Let $\delta>0$ and $U\mydef [-\delta,\delta]$.
We say a function $f : U^2 \longrightarrow \R$ is \emph{well-behaved} if the following conditions are met.
\begin{itemize}
\item $f$ is a $C^3$-function, i.e., three times continuously differentiable,
\item $f(0, 0) = 0$, and all partial derivatives $f_x$, $f_y$, $f_{xx}$, $f_{xy}$ and $f_{yy}$ are rational in $(0,0)$, and
\item $f_x(0, 0) \neq 0$ or $f_y(0, 0) \neq 0$.
\end{itemize}

We write the curvature of a well-behaved function $f$ at $(0,0)$ as
\[\kappa = \kappa(f)  = \left(\frac{f_y^2f_{xx} - 2f_xf_yf_{xy} + f_x^2f_{yy}}{(f_x^2 + f_y^2)^{\frac{3}{2}}}\right)(0,0).\]
We say $f$ is \emph{convexly curved} if $\kappa(f) < 0$, and \emph{concavely curved} if $\kappa(f) >0$.
\end{definition}

\begin{theorem}[Miltzow and Schmiermann~\cite{miltzow2021classifying}]\label{thm:main}
Let $\delta>0$, $U\mydef [-\delta,\delta]$ and $f,g : U^2 \longrightarrow \mathbb R$ be two well-behaved functions, one being convexly curved, and the other being concavely curved.
Then deciding whether a \fgetr formula $\Phi$ is satisfiable is an $\ER$-hard problem, even when $\delta = n^{-c}$ for any constant $c>0$ and when we are promised that $V(\Phi)\subset U^n$.
\end{theorem}

The promise that we only need to look for solutions to a \fgetr formula in the tiny hyper-cube $U^n$ will be crucial in our reduction.
Note that the paper~\cite{miltzow2021classifying} is also focused on proving \ER-membership of problems like \fgetr.
We will only use the problem to prove \ER-hardness, which is why we present the theorem in a slightly simpler form.

\subsubsection*{Related work}
Milenkovich~\cite{DBLP:conf/stoc/Milenkovic96,DBLP:journals/algorithmica/Milenkovic97} described exact algorithms for \pack\polygon\polygon\translation\ with running times exponential in the number of pieces; see also the papers by Milenkovich and Daniels~\cite{doi:10.1111/j.1475-3995.1999.tb00171.x, DBLP:journals/algorithmica/DanielsM97}.
In another paper, Milenkovich~\cite{DBLP:journals/comgeo/Milenkovic99} gave an algorithm and described a robust floating point implementation for the problem \pack\polygon\polygon\rotation\ using a combination of computational geometry and mathematical programming.

Alt~\cite{DBLP:journals/eatcs/Alt16} provides a survey of the literature on packing problems from a theoretical point of view.
A lot of work has been done on bin packing, strip packing and knapsack, usually with rectangular pieces that can be translated or rotated by $90^\circ$.
We refer to the survey of Christensen, Khan, Pokutta and Tetali~\cite{CHRISTENSEN201763} for an overview.

Recently, Merino and Wiese studied a version of $2$-dimensional knapsack where the pieces are convex polygons and arbitrary rotations are allowed and presented a QPTAS for the problem~\cite{DBLP:conf/icalp/MerinoW20}.
Note that the problem \pack\convexpolygon\square\rotation\ is a special case of the knapsack problem, so it is also $\ER$-hard by our result.

Several packing variants are known to be NP-hard.
Here we mention the problem of packing squares into a square by translation~\cite{leung1990packing},
packing segments into a simple polygon by translation~\cite{PackingSegments},
packing disks into a square~\cite{demaine2010circle},
packing identical simple polygons into a simple polygon by translation~\cite{DBLP:journals/corr/abs-1209-5307},
and packing unit squares into a polygon with holes by translation~\cite{berman1982optimal,fowler1981optimal}.
Alt~\cite{DBLP:journals/eatcs/Alt16} proves by a simple reduction from the partition problem that packing rectangles into a rectangle is NP-hard, and this reduction works with and without rotations allowed (note that \emph{a priori}, it is not clear that rotations make the problems more difficult, and it is straightforward to define (artificial) problems that even get easier with rotations).
It is easy to modify the reduction to the problem of packing rectangles into a square, so this implies NP-hardness of all problems in \Cref{fig:Results}.

\begin{figure}
\centering
\includegraphics[page = 2]{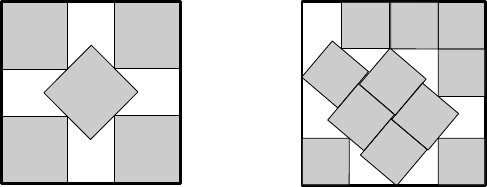}
\caption{Left: The optimal packing of five unit squares already requires rotations. Right: The currently best known packing of eleven unit squares into a larger square~\cite{Gensane2005}.}
\label{fig:UnitSquarePacking}
\end{figure}

A fundamental problem related to packing is to find the smallest square containing a given number of \emph{unit} squares, with rotations allowed. 
A long line of mathematical research has been devoted to this problem, initiated by Erd\H{o}s and Graham~\cite{erdos1975packing} in 1975, and it is still an active research area~\cite{chung2019efficient}.
Even for \emph{eleven} unit squares, the exact answer is unknown~\cite{Gensane2005}; see \Cref{fig:UnitSquarePacking}.
Other packing problems have much older roots, for instance Kepler's conjecture on the densest packings of spheres from 1611, famously proven by Hales in 2005~\cite{hales2005proof}.
The 2D analog, i.e., finding the densest packings of unit disks, was solved already in 1773 by Lagrange under the assumption that the disk configurations are lattices, and the general case was solved by Fejes T\'{o}th in 1940~\cite{fejes1942dichteste} (Thue already published a proof in 1910~\cite{thue1910} which is considered incomplete by some experts~\cite{10.2307/26395747}).

There is a staggering amount of papers in operations research on packing problems.
The research is mainly experimental and focuses on the development of heuristics to solve benchmark instances efficiently. 
We refer to some surveys for an overview~\cite{bennell2008geometry, bennell2009tutorial, dyckhoff1992cutting, hopper2001review,
leao2020irregular, sweeney1992cutting}.
In contrary to theoretical work, there is a lot of experimental work on packing pieces with irregular shapes and with arbitrary rotations allowed.

\subsubsection*{Open problems}

A natural continuation of our research is to study even more restricted packing variants.
The problem \pack{\rectangle}{\square}{\rotation} (packing arbitrary rectangles) is particularly interesting, as rectangles are very simple and widely studied. 
The problem \pack{\disksymbol}{\square}{\translation} (packing disks of arbitrary sizes, which curiously has some relevance to origami~\cite{demaine2010circle}) is interesting for similar reasons.
Our techniques seem not to extend to these special cases.
If both problems are \ER-complete, we expect that the proof techniques must be very different, since the non-linearity stems from rotations in the rectangle case, but from the shape of the pieces in the disk case.
A first step to show \ER-completeness of \pack{\disksymbol}{\square}{\translation} could be to show \ER-completeness of \pack{\convexcurved}{\polygon}{\translation}.
(Here, \convexcurved denotes convex curved polygons, or another reasonable generalization of convex polygons to pieces with curved boundaries.)
The access to both convex and concave constraints is key in the currently known techniques for proving \ER-hardness~\cite{miltzow2021classifying}.
If all the pieces are convex and only translations are allowed, it seems that we can only encode concave constraints.
It is therefore conceivable that \pack{\convexcurved}{\polygon}{\translation} and thus also \pack{\disksymbol}{\square}{\translation} are contained in \NP, by an argument similar as to why \pack\polygon\polygon\translation is in NP.

Another direction of research would be to consider the parameterized complexity of geometric packing problems.
A natural parameter is the number of pieces, $k$.
Then an instance of the problem \pack\polygon\polygon\rotation can be formulated as an \etr formula with $3k$ real variables, since the placement of each piece can be described by a two-dimensional translation and a rotation.
Using general purpose algorithms from real algebraic geometry~\cite{basu2006algorithms}, this implies that the instance can be decided in $L^{O(k)}$ time, where $L$ is the total length of the formula.
It is interesting to find out if this is best possible under widely believed hypotheses.
As a first step, one could investigate whether packing is W[1]-hard, or to show that $L^{O(k)}$ is the best possible running time assuming the exponential time hypothesis (ETH).
In a second step, it would be interesting to see if it is possible to solve the packing problem using an \etr formula with fewer real variables.
We note that W[1]-hardness or ETH-based lower bounds are not enough to give lower bounds on the number of real variables that are needed, as there could be a two-phase algorithm:
The first phase runs in $L^{O(k)}$ time, without using tools from real algebraic geometry.
The second phase solves an \etr formula with, say, only $O(\sqrt{k})$ real variables.
Neither W[1]-hardness nor an ETH lower bound of $L^{\Omega(k)}$ would exclude this scenario.

\subsubsection*{Acknowledgments}
We thank Reinier Schmiermann for useful discussions related to the use of tools from~\cite{miltzow2021classifying}.
We would also like to thank anonymous reviewers for comments on earlier versions of this article.

%%%%%%%%%%%%%%%%%%%%%%%%%%%%%%
\section{Reduction skeleton}
\label{sec:overview}
%%%%%%%%%%%%%%%%%%%%%%%%%%%%%%

In this section, we give an overview of the steps and concepts needed in our reductions.
The rest of the paper will then fill out the details.

\subsection{The problem \texorpdfstring{\wiredinv}{Wired-Curve-ETR[f,g]}}
We will reduce from an auxiliary problem called \emph{\wiredinv}.
An instance of this problem is a graphical representation of a \fgetr formula $\Phi$, i.e., a drawing of $\Phi$ of a specific form, which we call a \emph{wiring diagram}; see \Cref{fig:WiringDiagram}.
We denote by $n$ the number of variables of $\Phi$.

\begin{figure}[b]
\centering
\includegraphics[page=4]{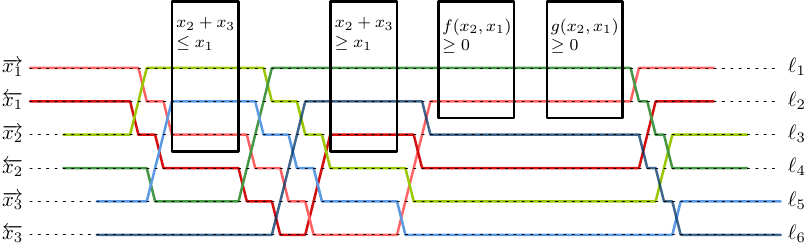}
\caption{A wiring diagram corresponding to the \fgetr formula $x_2+x_3=x_1\land f(x_2,x_1)\geq 0\land g(x_2,x_1)\geq 0$.}
\label{fig:WiringDiagram}
\end{figure}

The term ``wiring diagram'' is often used for drawings of a similar appearance used to represent electrical circuits
or pseudoline arrangements.
We define equidistant horizontal \emph{diagram lines} $\ell_1,\ldots,\ell_{2n}$ so that $\ell_1$ is the topmost one and $\ell_{2n}$ is bottommost. 
The distance between consecutive lines is~$10$.
In a wiring diagram, each variable $x_i$ in $\Phi$ is represented by two $x$-monotone polygonal curves $\overrightarrow{x_i}$ and $\overleftarrow{x_i}$, which we call \emph{wires}.
We think of $\overrightarrow{x_i}$ as oriented to the right and $\overleftarrow{x_i}$ as oriented to the left.
The wire $\overrightarrow{x_i}$ starts and ends on $\ell_{2i-1}$, and $\overleftarrow{x_i}$ starts and ends on $\ell_{2i}$.
Each wire consists of horizontal segments contained in the diagram lines and \emph{jump segments}, which are line segments connecting one diagram line $\ell_j$ to a neighbouring diagram line $\ell_{j\pm 1}$.
The wires are disjoint except that each jump segment must cross exactly one other jump segment.
Thus, the jump segments are used when two wires following neighbouring diagram lines swap lines.

The left and right endpoints of $\overrightarrow{x_i}$ and $\overleftarrow{x_i}$ are vertically aligned, and the wires appear and disappear in the order $(\overrightarrow{x_1},\overleftarrow{x_1}),\ldots,(\overrightarrow{x_n},\overleftarrow{x_n})$ from left to right in a staircase-like fashion.

In the wiring diagram, we represent each addition constraint of $\Phi$ as two inequalities, i.e., $x_i+x_j=x_k$ becomes $x_i+x_j\leq x_k$ and $x_i+x_j\geq x_k$.
Each addition inequality and each curved constraint ($f(x_i,x_j)\geq 0$ and $g(x_i,x_j)\geq 0$) is represented by an axis-parallel \emph{constraint box} intersecting the three or two topmost diagram lines; three for addition constraints and two for curved constraints.
These boxes are pairwise disjoint.
For a constraint $x_i+x_j\leq x_k$, the right-oriented wires $\overrightarrow{x_i},\overrightarrow{x_j},\overrightarrow{x_k}$ must inside the box occupy the lines $\ell_1,\ell_2,\ell_3$, respectively.
For $x_i+x_j\geq x_k$, we need the left-oriented wires $\overleftarrow{x_i},\overleftarrow{x_j},\overleftarrow{x_k}$ instead.
For a curved constraint $f(x_i,x_j)\geq 0$ or $f(x_i,x_j)\geq 0$, we need one of the wires of $x_i$ and one of the wires of $x_j$ to occupy $\ell_1$ and $\ell_2$.
Which combination and which order depends on the particular variant of packing that we are reducing to.

As we define our packing instance using a vertical line sweeping over the wiring diagram from left to right, we require that each vertical line is allowed to cross either zero or two jump segments and in the latter case, these two must cross each other.
A vertical line crossing a constraint box must not cross any jump segment.
This ensures that we only make one new feature in each step of the construction.

\begin{definition}
\label{def:wiredinv}
An instance $\mathcal I\mydef[\Phi,D]$ of the \emph{\wiredinv problem} consists of a \fgetr formula $\Phi$ together with a wiring diagram $D$ of $\Phi$.
\end{definition}

\begin{lemma}
\label{lem:Reduction-INV-WIRED}
Given a \fgetr formula $\Phi$ with variables $x_1,\ldots,x_n$, we can in $O(n^4)$ time construct a wiring diagram of $\Phi$.
\end{lemma}

\begin{proof}
We may assume that $\Phi$ has $O(n^3)$ constraints, since there will otherwise be duplicates of some constraints.
We construct a wiring diagram as follows; refer to \Cref{fig:WiringDiagram}.
We construct all the curves simultaneously from left to right.
We handle the constraints in order and define the curves as we go along.
For instance, for a constraint such as $x_i+x_j\leq x_k$, we route $\overrightarrow{x_i}$ to the line $\ell_1$ using jump segments.
This defines how all other curves should behave in the same range of $x$-coordinates where we have routed $\overrightarrow{x_i}$.
We then route $\overrightarrow{x_j}$ to the line $\ell_2$, and then route $\overrightarrow{x_k}$ to~$\ell_3$.
Each time we route a curve to a specific line, we introduce $O(n)$ crossings.
Therefore, we make $O(n^4)$ crossings in total.
\end{proof}

\subsection{Constructing a packing instance}
Let an instance $\I$ of \wiredinv be given with \fgetr formula $\Phi$ of variables $x_1,\ldots,x_n$ and $\delta\mydef n^{-300}$.
We are going to construct an instance of a packing problem with $N=O(n^4)$ pieces, since this will be the complexity of the size of the wiring diagram of $\I$.
The general idea is to build a packing instance on top of the wiring diagram.
We define a polygonal container $\cont\mydef \cont(\mathcal I)$ containing the wiring diagram in the interior, and a set of pieces $\p$ to be placed in $\cont$.
The container $\cont$ is bounded from below by a line segment, from left and right by $y$-monotone chains, and from above by an $x$-monotone chain.
See \Cref{fig:complete} for a sketch of a complete example.

In \Cref{sec:gadgets,sec:Curved}, we present reductions to the packing problems \pack \convexpolygon  \polygon \rotation, \pack \convexpolygon \curved \translation, and \pack \curved \polygon \translation, i.e., where the container is a polygon or a curved polygon.
In \Cref{sec:SquareContainer}, we show how for fixed polygon and motion type, packing into a polygonal container reduces to packing into a square container.
Together, these results imply hardness for \pack \convexpolygon \square \rotation\ and \pack \curved \square \translation.

\begin{figure}
\centering
\newlength{\imagewidth}
\settowidth{\imagewidth}{\includegraphics[page=2]{figures/FullExample2.pdf}}
\includegraphics[page=2, trim=0 0 0.8333333\imagewidth{} 0, clip, width = \textwidth]{figures/FullExample2.pdf} \\
\includegraphics[page=2, trim=0.166666\imagewidth{} 0 0.6666666\imagewidth{} 0, clip, width = \textwidth]{figures/FullExample2.pdf} \\
\includegraphics[page=2, trim=0.333333\imagewidth{} 0 0.5\imagewidth{} 0, clip, width = \textwidth]{figures/FullExample2.pdf} \\
\includegraphics[page=2, trim=0.5\imagewidth{} 0 0.3333333\imagewidth{} 0, clip, width = \textwidth]{figures/FullExample2.pdf} \\
\includegraphics[page=2, trim=0.666666\imagewidth{} 0 0.16666666\imagewidth{} 0, clip, width = \textwidth]{figures/FullExample2.pdf}
\includegraphics[page=2, trim=0.833333\imagewidth{} 0 0.0\imagewidth{} 0, clip, width = \textwidth]{figures/FullExample2.pdf}
\caption{A sketch of the instance of \protect\pack{\protect\polygon}{\protect\polygon}{\protect\rotation} we get from the wiring diagram in \Cref{fig:WiringDiagram}, broken over six lines.
The adders and curvers (swings) are marked with gray boxes.
The (light and dark) red, blue and green pieces are the variable pieces and the pieces of each nuance form a lane.}
\label{fig:complete}
\end{figure}

\subsubsection*{Defining the construction in steps from left to right}
We define the packing instance as we sweep over the wiring diagram of $\I$ with a vertical sweep line from left to right.
Each step corresponds to one of the following events:
\begin{itemize}
\item
the introduction of a pair of wires $(\overrightarrow{x_i},\overleftarrow{x_i})$,

\item
a crossing of two wires,

\item
an addition or curved constraint,

\item
the termination of a pair of wires $(\overrightarrow{x_i},\overleftarrow{x_i})$.
\end{itemize}
In each step, we add one or more \emph{gadgets}, each involving a constant number of pieces and possibly a constant number of edges to the boundary of the container $\cont$.
When the sweep line passes over the right endpoints of the last wires $(\overrightarrow{x_n},\overleftarrow{x_n})$, the construction of the container $\cont$ and all the pieces $\p$ is complete.

The overall goal of the construction is to prove the following theorem.

\begin{theorem}
\label{thm:Reduction}
Let $\I$ be an instance of \wiredinv.
For each of the problems \pack\convexpolygon\polygon\rotation, \pack\convexpolygon\curved\translation, and  \pack\curved\polygon\translation, we can in polynomial time construct an instance of the problem consisting of a container $\cont$ and a set of pieces \p such that $\I$ has a solution if and only if there is a valid placement of $\p$ in $\cont$.
\end{theorem}

From \Cref{thm:Reduction} and the \ER-hardness of \wiredinv, we now immediately get the claimed results of row one and two in \Cref{fig:Results}.

\begin{corollary}
\label{thm:packPolygon}
The problems \pack\convexpolygon\polygon\rotation, \pack\convexpolygon\curved\translation\ and  \pack\curved\polygon\translation\ are \ER-hard.
\end{corollary}

The third row in \Cref{fig:Results} follows from \Cref{sec:SquareContainer}, as described later in this section.
We now describe how values of variables can be encoded by the placement of certain pieces in our constructed packing instances.

\subsubsection*{Variable pieces}
Each variable $x$ of $\Phi$ will be represented by a number of \emph{variable pieces} in our construction, each of which is a convex polygon.
Each variable piece represents exactly one variable $x$, and we make a correspondence between certain placements of the piece and the values of $x$.
When adding a variable piece to our construction, we also specify the \emph{zero placement} of the piece, which is a specific placement where it encodes the value $0$ of $x$.
In the zero placement, the piece will have a pair of (long) horizontal edges which have distance $10$.
By sliding the piece to the left or to the right from the zero placement, we obtain placements of the piece that encode all real values of $x$, even values outside the range $[-\delta,\delta]$.
Each variable piece will be defined to be either right- or left-oriented.
By sliding a right-oriented (resp.~left-oriented) variable piece to the right by some amount $t\geq 0$, we obtain a placement that encodes the value $t$ (resp.~$-t$), while sliding it to the left by $t$ results in a placement encoding $-t$ (resp.~$t$).
If the piece is rotated differently or placed higher or lower than the zero placement, we do not define any value of $x$ to be encoded by the placement (we will in fact prove that in any solution to the constructed packing instance, no piece can have such an undesirable placement).
We define the \emph{canonical placements} of a variable piece to be the placements that encode values in the interval $[-\delta,\delta]$; see \Cref{fig:VariablePiece}.

On each wire $\overrightarrow{x_i}$ or $\overleftarrow{x_i}$ 
in the wiring diagram of $\I$, we will place several variable pieces representing $x_i$.
The pieces placed on one wire are called a \emph{lane}.
The variable pieces on $\overrightarrow{x_i}$ and $\overleftarrow{x_i}$ are oriented to the right and left, respectively. 
We also introduce some variable pieces which will be placed at other places than on the wires, namely above the topmost wire where they will be introduced in the steps of the construction corresponding to addition or curved constraints.

\begin{figure}
\centering
\includegraphics[page=1]{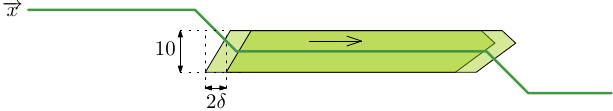}
\caption{A wire $\protect\overrightarrow x$ and a variable piece representing $x$ placed on top, showing the leftmost and rightmost canonical placements of the piece.
The large arrow in the piece indicates that the piece is right-oriented.}
\label{fig:VariablePiece}
\end{figure}

\begin{figure}
\centering
\includegraphics[page=5]{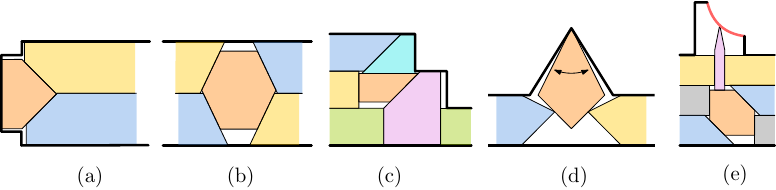}
\caption{
(a): anchor, (b): swap, (c): adder, (d): swing, (e): gramophone.}
\label{fig:reduction-gadgets}
\end{figure}

\subsubsection*{Gadgets}
Sketches of some of the gadgets can be seen in \Cref{fig:reduction-gadgets}.
In the wiring diagram, each variable~$x_i$ is represented by two wires $\overrightarrow{x_i}$ and $\overleftarrow{x_i}$ such that the left endpoints of $\overrightarrow{x_i}$ and $\overleftarrow{x_i}$ are vertically aligned at distance~$10$, as are the right endpoints.
In both ends of the wires, we build an \emph{anchor} (\Cref{sec:anchor}) which ensures that the pieces placed on $\overrightarrow{x_i}$ and those placed on $\overleftarrow{x_i}$ encode the value of $x_i$ consistently.
Furthermore, the anchor will ensure that the encoded value of $x_i$ is in the \emph{range} $I(x_i)$, which we define as
\[
I(x) =
\begin{cases}
\{\delta\} & \text{if $\Phi$ has a constraint } x = \delta,\\
[0, \delta] & \text{else, if $\Phi$ has a constraint } x \geq 0,\\
[-\delta, \delta] & \text{otherwise.}
\end{cases}
\]

Whenever two wires cross, we build a \emph{swap} (\Cref{sec:swap}).
The swap employs a central piece that can translate in all directions, so that when it is pushed by a variable piece, the push will propagate to the neighbouring variable piece on the other side of the crossing.
We describe \emph{adders} (\Cref{sec:addition}) to implement the addition constraints and \emph{curvers} (\Cref{sec:Curved}) for the curved constraints.
We describe two curvers (see \Cref{fig:reduction-gadgets} (d--e)), both of which exist in a convex and a concave variant.
Which version we use depends on the variant of packing we are reducing to.

Every time we add a gadget to the construction, we also introduce a constant number of new pieces.
Each variable piece is introduced in one gadget where the left end of the piece is defined.
The piece then extends outside the gadget to the right.
The right end of the piece will be defined in another gadget added later to the construction.
The piece is \emph{exiting} the former gadget and \emph{entering} the latter.
In between the left and right end of the piece, defined in these two gadgets, the piece is bounded by a pair of horizontal edges.
All pieces that are not variable pieces are contained within a single gadget.

\subsubsection*{Canonical placements}
Recall that we define canonical placements of each variable piece.
We do not define individual canonical placements of pieces that are not variable pieces, but instead we define canonical placements of all pieces of one or more gadgets:
A placement of a set of pieces is canonical if (1) the placement is valid (i.e., the pieces are in $\cont$ and are non-overlapping), (2) all variable pieces have a canonical placement, and (3) the pieces have certain relationships such as edge-edge contacts between each other.
Part (3) will be specified for each gadget individually.

\subsubsection*{Preservation of solutions}
The following lemma will be used to prove that for every solution to the \fgetr formula~$\Phi$, there is a canonical placement where the pieces encode that solution, meaning that for each variable $x$, all variable pieces representing $x$ encode the same value of $x$ as in the solution.
Define $\gadgets$ to be the total number of gadgets, and let $\p_i$ denote the set of all pieces introduced in the first $i$ gadgets, where $i\in\{0,\ldots,\gadgets\}$, so that $\p_0\mydef\emptyset$.
The proof will be given in the sections describing the individual types of gadgets.

\begin{lemma}[Solution preservation]
\label{lem:preservation}
Consider any $i\in\{1,\ldots,\gadgets\}$ and suppose that for every solution to $\Phi$, there is a canonical placement of the pieces $\p_{i-1}$ that encodes that solution.
Then the same holds for $\p_i$.
\end{lemma}

\subsubsection*{Soundness of the reduction}
As mentioned, each variable $x$ will be represented by many variable pieces in the complete construction.
A difficulty is that conceivably, such a piece may not be placed in a way that encodes a value of $x$.
Even if all the pieces happen to be placed such that they \emph{do} encode values of $x$, these values could be different and therefore not represent a solution to the formula $\Phi$.

For a small number $\Delta\geq \delta$, we are going to introduce a class of placements called \emph{aligned $\Delta$-placements}.
These are defined from the canonical placements by relaxing the requirements a bit.
In an aligned $\Delta$-placement, each variable piece must be placed so that it encodes a value, but it may slide $\Delta$ sideways from the placement encoding the value $0$ instead of at most $\delta$ as for the canonical placements.
The requirements to the other pieces are likewise relaxed and will be given later.
The values encoded by the variable pieces in an aligned $\Delta$-placement may therefore conceivably be outside the required range $[-\delta,\delta]$.
The following lemma tells us that this is not the case for $\Delta$ sufficiently small.
In fact, the existence of such a placement is enough to ensure that $\I$ has a solution.
The number $\slack$ is the \emph{slack} of the construction defined as the area of the container $\cont$ minus the total area of the pieces $\p$, and as will be explained later, $\mu=O(n^{-296})$ in our construction.
The number $g=O(n^4)$ is the number of gadgets.

\begin{lemma}[Soundness]
\label{lem:consistency}
Consider an aligned $g\slack$-placement and any variable $x$.
There is a specific non-empty subset of the pieces representing $x$ that encode the value of $x$ consistently (i.e., they all encode the same value of $x$) and the value is in the range $I(x)$.
Furthermore, these values of the variables satisfy the constraints of~$\Phi$.
\end{lemma}

In fact, it will follow from \Cref{lem:consistency} that every aligned $g\slack$-placement is canonical, but this is not important for our proof of \Cref{thm:Reduction}.
The remaining work in proving the theorem will be to prove the following lemma.

\begin{lemma}
\label{lem:validAligned}
Every valid placement is an aligned $g\slack$-placement.
\end{lemma}

The proof of \Cref{thm:Reduction} is now straight-forward:

\begin{proof}[Proof of Theorem \ref{thm:Reduction}]
If $\I$ has a solution, then it follows directly from repeated use of \Cref{lem:preservation} that there is a valid placement.
Suppose now that there is a valid placement.
By \Cref{lem:validAligned}, the placement is an aligned $g\slack$-placement.
By \Cref{lem:consistency}, some variable pieces encode a solution to the formula $\Phi$.
Therefore, $\I$ has a solution.
\end{proof}

\subsection{Basic tools: Slack, fingerprinting and unique angles}

In the following we will describe some tools needed to prove \Cref{lem:validAligned}.

\subsubsection*{The slack of the construction}
The \emph{slack} of an instance of a packing problem is the area of the container $\cont$ minus the total area of the pieces, and we denote the slack of our construction by $\slack$.
We need the slack to be very small in order to use the fingerprinting technique which will be described later.

We now give an upper bound on the slack of the complete construction.
Our construction will be described as depending on the number $\delta\mydef n^{-300}$.
We place each variable piece so that it encodes the value $0$, and we place the remaining pieces as shown in the sections that describe the individual gadgets.
We now define $\slack'$ to be the area of the container $\cont$ that is not covered by pieces in this placement and thus trivially have $\slack\leq\slack'$.
By checking each type of gadget, it is straightforward to verify that the placement can be realized as a canonical (and thus valid) placement of the pieces in all gadgets except for the anchors, where some pieces are not completely contained in $\cont$.
The uncovered area in each gadget will appear as a thin layer along some of the edges of the pieces in the gadget.
This layer has thickness $O(\delta)$, and the edges along which it appears have total length $O(1)$, so the area is $O(\delta)$ in each gadget.
There will be no empty space outside the gadgets because that space will be completely covered by variable pieces.
Since the final construction has $g=O(N)=O(n^4)$ gadgets, it follows that $\slack\leq\slack'=O(n^4\delta)=O(n^{-296})$.
It may seem a little odd to measure the slack in this indirect way, but we found it to be the easiest way to get the bound since we do not explicitly specify the area of the container or the pieces in our construction. 
The bound $\slack = O(n^{-296})$ ensures that $\slack$ is sufficiently small that we can use of the fingerprinting technique, which is described in the following.

\begin{figure}
\centering
\includegraphics[page = 2]{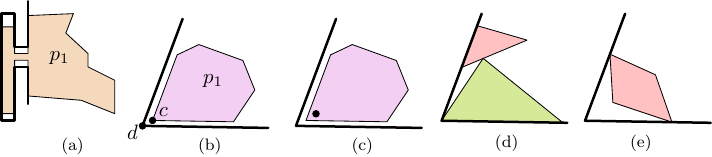}
\caption{(a): A pocket and an augmentation that fit perfectly together as in a jigsaw puzzle.
(b): A wedge of the empty space and a piece which fit together.
(c): The corner of the piece we are fingerprinting is marked with a dot.
(d) and (e): Two examples where space is wasted because a wedge is not occupied by a piece with a matching angle.
}
\label{fig:Jigsaw}
\end{figure}

\subsubsection*{Fingerprinting}
In order to prove \Cref{lem:validAligned}, we first show that every piece must be placed very close to a canonical placement using a technique we call \emph{fingerprinting}.
To grasp the idea of this technique, we first present another simpler technique that only works for non-convex pieces, and which we call the \emph{jigsaw puzzle} technique for obvious reasons; see \Cref{fig:Jigsaw} (a).
The idea behind this technique is to force each piece $p_1$ 
to be at a specific position by creating a pocket of the container and a corresponding
augmentation of the piece $p_1$ intended to be placed there.
This is done in a way that only the piece $p_1$ has an augmentation that fits into the pocket, just as the principle behind a jigsaw puzzle, and it can be done in a way that gives the piece freedom to slide back and forth or rotate by a slight amount, etc.
The pocket can also be created in another piece $p_2$ if $p_1$ is intended to be placed next to $p_2$.
Making enough of these pairs of pockets and extensions, we can therefore deduce where all the pieces are placed in all valid placements.

In fact, the jigsaw puzzle technique can be used to prove $\ER$-hardness of packing problems with non-convex pieces in a much simpler way than the proofs of this paper, but unfortunately, the technique is not directly realizable with convex pieces.
In \emph{fingerprinting}, instead of making complicated augmentations of the pieces, we only work with a piece $p_1$ with a convex corner~$c$ of a specific angle $\alpha_1$.
In the canonical placements, the empty space left by the other pieces forms a wedge with apex corner $d$ of angle $\alpha_1$ which can thus be covered very efficiently by~$p_1$ by placing the corner $c$ at or very close to $d$, as in \Cref{fig:Jigsaw} (b).
We make sure that every corner of every other piece has an angle $\alpha_2$ significantly different from $\alpha_1$, in the sense that $|\alpha_2-\alpha_1|=\Omega(N^{-2})$.
It should likewise hold that the total angle of any combination of corners of other pieces is different from $\alpha_1$ in that sense.
Furthermore, the slack $\slack$ is tiny, as described above.
As a result, we can show that if $c$ is \emph{not} placed very close to $d$, this will result in the empty space in a neighbourhood around $d$ with an area exceeding $\slack$, because no other piece (or combination of pieces) can cover that neighbourhood efficiently, illustrated in \Cref{fig:Jigsaw}~(d-e).
In our constructions, fingerprinted corners are marked with a dot; see \Cref{fig:Jigsaw}~(c).
For technical reasons, the fingerprinted corners must have angles in the range from~$5\pi/180$ to~$\pi/2$.

We add a few remarks about the use of fingerprinting.
First of all, the situation shown in \Cref{fig:Jigsaw}~(b) is simplified.
In our applications, the fat segments (bounding the empty space left by the other pieces) do not need to meet at the apex corner $d$, since a short portion (of length $O(\delta)$) close to the corner can be missing.
Furthermore, the angle between the two fat segments does not have to be exactly $\alpha_1$ before the technique can be used; just very close to $\alpha_1$ (this will be important when we are fingerprinting more than one piece in a row).

Second, the fingerprinting does not imply that the piece $p_1$ \emph{must} be placed with the corner~$c$ coincident with $d$, but only that the distance $\|cd\|$ has to be small.
This is used deliberately in our constructions, since it allows for the piece $p_1$ to move slightly.

Third, when we introduce a new gadget and its $k$ new pieces $p_i,\ldots,p_{i+k-1}$, 
we use fingerprinting iteratively to argue where the new pieces must be placed.
Here, the $j$'th piece $p_{i+j-1}$, $j\in\{1,\ldots,k\}$, can be fingerprinted in a wedge of the empty space formed by the preceding pieces $p_i,\ldots,p_{i+j-2}$.
However, the bound on the uncertainty of where $p_{i+j-1}$ is placed increases with $j$.
Slightly simplified, the bound grows as $O(n^{O(j)}\sqrt{\slack})$, and we need the bound to be at most some small constant to be of any use.
We prefer to create an instance where we need only a logarithmic number of bits to represent the coordinates of the container and the pieces, since this will prove that the packing problems are \emph{strongly} $\ER$-hard.
It is therefore important that we only apply fingerprinting iteratively a constant number of times, i.e., that $k=O(1)$, as we will otherwise need to choose the slack $\slack$ to smaller than $n^{-q}$ for every constant $q>0$, and then it will require a superlogarithmic number of bits to represent the coordinates of our instance.
In the construction, we will always have $k\leq 7$ and we can do with choosing $\delta\mydef n^{-300}$ so that $\slack=O(n^{-296})$.

In \Cref{sec:FingerPrinting}, we will develop the fingerprinting technique in detail.
We consider this part the technically most challenging of the paper.
The technique is versatile and can likely be used in other reductions to packing.
Let us for instance mention that fingerprinting also works for packing problems where we are allowed to reflect the pieces when placing them.
This gives two possibilities for each piece, and one must be excluded by other reasons, such as overlap with other pieces.
As the conditions for the fingerprinting technique are technical, we will list them in \Cref{sec:FingerPrinting} in detail.
We will show for each gadget that those conditions are met.

After using fingerprinting iteratively a constant number of times for the new pieces, we use other techniques, such as the \emph{alignment} (to be described in the sequel), to argue about their placement.

\subsubsection*{Choosing unique angles}
As described in the previous paragraph, whenever we apply the fingerprinting technique to argue that a corner with angle $\alpha_1$ of some piece is placed close to some specific point in the container, we need that every combination of corners of the other pieces have angles that sum to an angle $\alpha_2$ such that $|\alpha_2-\alpha_1|=\Omega(N^{-2})$.
This will be called the \emph{unique angle property}.
In order to obtain this property, the construction will be designed so that each piece has a special corner where the angle can be chosen freely (within some interval of angles of size $\Omega(1)$).
Likewise, the wedge (where the special corner is intended to be placed) formed by the boundary of the container or the other pieces is flexible, so that the angle of the wedge can match the chosen angle of the corner.
In our figures, the special fingerprinted corners are marked with a dot, as in \Cref{fig:Jigsaw} (c).
The following lemma is used to choose these free angles such that we get the unique angle property.

\begin{lemma}
\label{lem:UniqueAngles}
Let $S_k\mydef \{\frac ik+\frac 1{2k^2}\mid i\in\{1,\ldots,k\}\}$, consider a subset $R\subseteq S_k$, and let $x\mydef \sum_{r\in R} r$.
If $x\in S_k$, then $R=\{x\}$.
\end{lemma}

\begin{proof}
If $R$ consists of $m$ elements, then $x$ has the form $\frac jk+\frac m{2k^2}$ for some $j\in\N$.
A number of this form, for $m\leq k$, can only be in $S_k$ if $m=1$.
\end{proof}

The lemma provides a set of $k$ numbers in a range of size $O(1)$ and any number is $\Omega(k^{-2})$ away from the sum of any combination of other numbers.
We multiply the numbers in $S_k$ by~$\pi$ to get a set of rational angles and choose the free angles from such a set $\pi S_k$, for $k=O(N)$.
The free angles are restricted to various subintervals of $[0,\pi]$, so we choose $k$ so large that $\pi S_k$ contains enough angles from each of these subintervals.
However, as each subinterval has size $\Omega(1)$, we can do with $k=O(N)$.

\subsection{Proof structure of \texorpdfstring{\Cref{lem:validAligned}}{Lemma \ref{lem:validAligned}}}
In the proof of \Cref{lem:validAligned}, we use the fingerprinting technique to prove that in every valid placement, the pieces are placed almost as in a canonical placement.
To explain the structure of the argument in more detail, we need some notions of placements that are close to being canonical, which will be defined in the following paragraph.

\subsubsection*{Almost-canonical placements and aligned placements}
We say that a valid placement of the pieces of a gadget is \emph{almost-canonical} if there exists rigid motions that move the pieces to a canonical placement such that every point in each piece is moved a distance of at most $n^{-1}$ (in other words, the \emph{displacement} between the actual placement and the canonical placement of each piece is $n^{-1}$).

We say that a placement of the pieces of a gadget is an \emph{aligned $\Delta$-placement} for $\Delta\geq\delta$ if (i) the placement is almost-canonical, and (ii) for each variable $x$, each variable piece representing~$x$ encodes a value in the range $[-\Delta,\Delta]$.
Note that since the placement is almost-canonical, we can always assume $\Delta\leq n^{-1}+\delta=n^{-1}+n^{-300}$.

The following lemma says that the pieces of every new gadget can be assumed to be almost-canonical if the preceding pieces have an aligned $\Delta$-placement, for $\Delta$ sufficiently small.
Recall that $\p_i$ is the set of all pieces introduced in the first $i$ gadgets, where $i\in\{0,\ldots,\gadgets\}$, so that $\p_0\mydef\emptyset$.
The lemma will follow from the use of fingerprinting, and the proof will be given in the sections describing the individual types of gadgets.

\begin{lemma}[Almost-canonical Placement]
\label{lem:AlmostCanonicalPlacement}
For any $i\in\{1,\ldots,\gadgets\}$, consider a valid placement~$P$ (of all the pieces) for which the pieces $\p_{i-1}$
have an aligned $(i-1)\slack$-placement. 
It then holds for~$P$ that the pieces $\p_i$ have an almost-canonical placement. 
\end{lemma}

\begin{figure}
\centering
\includegraphics[page = 3]{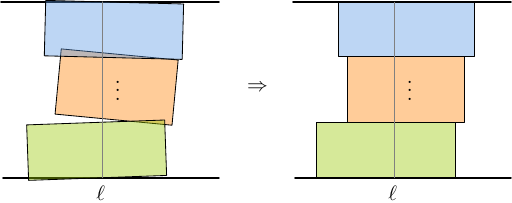}
\caption{The alignment segment $\ell$ makes us conclude that the pieces must be horizontally aligned: otherwise, they would overlap, cross the container boundary, or cover more of the alignment segment than what is available.
}
\label{fig:Align}
\end{figure}

\subsubsection*{Aligning pieces}
Once we know that the pieces $\p_i$ have an almost-canonical placement, provided by the previous lemma, we can use so-called \emph{alignment segments} to further restrict where the new pieces $\p_i\setminus\p_{i-1}$ introduced in gadget $i$ can be placed.
In particular, we will be able to fix the rotations of some pieces to be as in the canonical placements.
The idea is sketched in \Cref{fig:Align}.
From the rough placements we get from fingerprinting, we know that a set of the pieces each has a pair of parallel edges that are both cut through by a vertical alignment segment~$\ell$.
If we sum the distance between the two parallel edges over all the pieces, we get exactly the length of~$\ell$.
Since the portions of~$\ell$ covered by the pieces must be pairwise disjoint in a valid placement, we can conclude that the pieces have to be rotated so that the parallel edges are perpendicular to~$\ell$.
This technique will be used to prove the following lemma for each gadget individually.
Using the lemma repeatedly together with \Cref{lem:AlmostCanonicalPlacement}, we get that every valid placement is also an aligned $\gadgets\slack$-placement, proving \Cref{lem:validAligned}.

\begin{lemma}[Aligned placement]
\label{lem:alignedPl}
  For any $i\in\{1,\ldots,\gadgets\}$, consider a valid placement $P$ (of all the pieces) for which the pieces $\p_{i-1}$
  have an aligned $(i-1)\slack$-placement and the pieces $\p_i$ have an almost-canonical placement.
  It then holds for $P$ that the pieces $\p_i$ have an aligned $i\slack$-placement. 
\end{lemma}

It now remains to prove \Cref{lem:consistency}.

\subsection{Proof structure of \texorpdfstring{\cref{lem:consistency}}{Lemma \ref{lem:consistency}}}
The proof of \Cref{lem:consistency} goes along the following lines.
In an aligned $\gadgets\slack$-placement, each variable piece $p_x$ encodes a value for the variable $x$ it is representing, which we will denote by $\enc{p_x}$.
The problem is that different pieces representing the same variable $x$ may conceivably not encode the value consistently.
However, recall that we build lanes of pieces on top of the two wires $\overrightarrow x,\overleftarrow x$, and these meet at the left and right endpoints of the wires.
We prove that the values encoded by these pieces $p_1,\ldots,p_m$ make a cycle of inequalities: $\enc{p_1}\leq\enc{p_2}\leq\cdots\leq\enc{p_m}\leq\enc{p_1}$.
It thus follows that all these pieces encode a value of $x$ consistently.
Furthermore, the anchors, which are the gadgets that we place at the left and right endpoints of the wires $\overrightarrow x,\overleftarrow x$, will ensure that $\enc{p_1}\in I(x)$, so that the encoded values are in the correct range.

In our construction, we also make additional lanes of pieces going to the adders and curvers.
The functionality of the specific gadgets imply that the addition and curved constraints of $\Phi$ are all satisfied.
In order to describe the structure of the argument, we introduce a graph~$G_x$ for each variable $x$ as described in the next paragraph.

\begin{figure}
\centering
\includegraphics[page = 7]{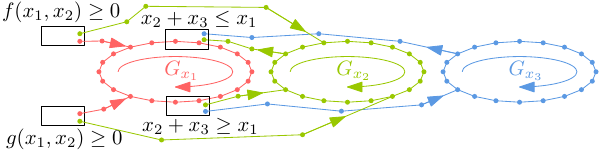}
\caption{An abstract drawing of the dependency graphs of the instance we get from the wiring diagram in \Cref{fig:WiringDiagram} and how the graphs connect to the gadgets for addition and curved constraints.
The number of vertices on the cycles and paths are neither important nor correct.
}
\label{fig:Graph}
\end{figure}

\subsubsection*{Dependency graph of variable pieces}
For each variable $x$, we introduce a directed \emph{dependency graph} $G_x$.
The vertices of $G_x$ are the variable pieces representing $x$.
Consider a gadget and two variable pieces $p_1,p_2$ appearing in the gadget and both representing $x$.
We add an edge from $p_1$ to $p_2$ in $G_x$ if $p_1$ is an entering right-oriented piece or an exiting left-oriented piece and $p_2$ is an exiting right-oriented piece or an entering left-oriented piece.
In crossings between the two wires $\overrightarrow x,\overleftarrow x$ representing $x$, there will be a swap where this rule introduces unintended edges, so we make one exception described in \Cref{sec:swap} where the swap is described in detail.

The following lemma is going to follow trivially from the way we make the lanes on top of the wires $\overrightarrow x,\overleftarrow x$ for each variable $x$, and the way we connect the gadgets representing addition and curved constraints to these lanes.
See \Cref{fig:Graph} for an illustration.

\begin{lemma}
\label{lem:graph}
For each variable $x$, the graph $G_x$ consists of a directed cycle $K_x$ with some directed paths attached to it (oriented towards or away from $K_x$).
The vertices of the cycle $K_x$ are the variable pieces appearing on the wire $\overrightarrow x$ from left to right and the wire $\overleftarrow x$ from right to left in this order.
For each path attached to $K_x$, the vertex farthest from $K_x$ is a piece entering or leaving a gadget representing an addition or curved constraint.
\end{lemma}

In the following, we consider a given aligned $\gadgets\slack$-placement of all the pieces.
Since all edges of $G_x$ are between pieces appearing in the same gadget, the following lemma will be proven for each gadget individually.
\begin{lemma}[Edge inequality]
\label{lem:graphIneq}
Consider a variable $x$ and an edge $(p_1,p_2)$ of $G_x$.
Then $\enc{p_1}\leq \enc{p_2}$.
\end{lemma}

From \Cref{lem:graph} and~\Cref{lem:graphIneq}, we now get the following (except that the part about the anchor gadget will be proven in \Cref{sec:anchor}).

\begin{lemma}
\label{lem:consistentCycle}
For each variable $x$, all the pieces of the cycle $K_x$ encode the value of $x$ consistently.
Furthermore, due to the design of the anchor gadget, the value is in $I(x)$.
\end{lemma}

By the above lemma, we may write $\enc{K_x}$ to denote the value represented by all pieces of~$K_x$.

\subsubsection*{Adders and curvers work}
We will show in \Cref{sec:addition} and \Cref{sec:Curved} that the adders and curvers actually enforce addition 
and curved constraints as they are supposed to.
This entails showing that the gadgets implement the addition constrainst or various convexly or concavely curved constraints in a geometric sense and also that the variable pieces of the gadgets are correctly connected to the cycles in the respective dependency graphs.
In particular, we will show the following two lemmas.

\begin{lemma}[Adders work]
\label{lem:addition}
For each constraint $x+y= z$ in the formula $\Phi$, we have $\enc{K_x}+\enc{K_y}=\enc{K_z}$.
\end{lemma}

\begin{lemma}[Curvers work]
\label{lem:inversion}
For each of the problems \pack\convexpolygon\polygon\rotation, \pack\convexpolygon\curved\translation, and  \pack\curved\polygon\translation, there exists well-behaved functions $f$ and $g$ that are convexly and concavely curved, respectively, such that for every constraint of the form $f(x,y)\geq 0$ in the \fgetr formula $\Phi$, we have $f(\enc{K_x}, \enc{K_y})\geq 0$, and for every constraint $g(x,y)\geq 0$, we have $g(\enc{K_x}, \enc{K_y})\geq 0$.
\end{lemma}

Combining \Cref{lem:consistentCycle,lem:addition,lem:inversion}, we then have a proof of \Cref{lem:consistency}.

\subsection{Square container}
In \Cref{sec:SquareContainer}, we describe a reduction from problems of type \pack{\piecetype}{\polygon}{\motiontype}  to \pack{\piecetype}{\square}{\motiontype}.
It will be crucial that the container $\cont$ is \emph{$4$-monotone}, as defined below.

\begin{restatable}{definition}{deffourmonotone}
\label{def:4monotone}
A simple closed curve $\gamma$ is \emph{4-monotone} if $\gamma$ can be partitioned into four parts $\gamma_1,\ldots,\gamma_4$ in counterclockwise order that move monotonically down, to the right, up, and to the left, respectively.
A polygon $\polQ$ is $4$-monotone if the boundary of $\polQ$ is a $4$-monotone curve.
\end{restatable}

\begin{lemma}
\label{lem:4monotone}
In the reductions resulting from using the gadgets described in \Cref{sec:gadgets,sec:Curved}, the resulting container is $4$-monotone.
\end{lemma}

\begin{proof}
The boundary of the resulting container has a left and a right staircase, $\gamma_1$ and $\gamma_3$, created by the left and right anchors, respectively, and these staircases are $y$-monotone, and their upper and lower endpoints are horizontally aligned.
The lower endpoints of the staircases $\gamma_1$ and $\gamma_3$ are connected by a single horizontal line segment $\gamma_2$ bounding the bottom lane from below.
The upper endpoints of the staircases are connected by a curve $\gamma_4$ which bound the topmost lane and the adders and curvers from above.
The curve $\gamma_4$ is $x$-monotone, as can easily be verified by inspecting the boundary added due to the adders and curvers.
Hence, the container is $4$-monotone.
\end{proof}

\begin{figure}
\centering
\includegraphics[page = 6]{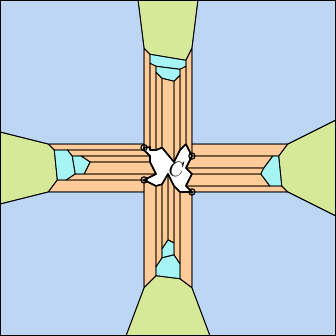}
\caption{Construction used in the reduction to packing problems with a square container.
The space left by the exterior pieces (blue, green, orange, and turquoise) is exactly the $4$-monotone container $\cont$ of the instance we are reducing from.}
\label{fig:squareSkeleton}
\end{figure}

We get from the lemma that the packing problems are even $\ER$-hard for $4$-monotone containers.
Let $\I_1$ be an instance of a packing problem where the container $\cont\mydef\cont(\I_1)$ is 4-monotone.
We place $\cont$ in the middle of a larger square and fill out the area around $\cont$ with pieces carefully; the details are given in \Cref{sec:SquareContainer} and \Cref{fig:squareSkeleton} shows an example of the construction.
We call these new pieces the \emph{exterior} pieces, whereas we call the pieces of $\I_1$ the \emph{inner} pieces.
Using fingerprinting and other arguments, we are able to prove that there is essentially only one way to fit the exterior pieces in the square, and the space left for the inner pieces is exactly the container $\cont$.
Therefore, there exists a valid placement of the pieces in the resulting instance if and only if there is one of the inner pieces in $\cont$.
We get the results in the third row of \Cref{fig:Results} as expressed by the following corollary.

\begin{restatable}{corollary}{thmSquarePack}
\label{thm:squarePack}
The problems \pack \convexpolygon \square \rotation\ and \pack \curved \square \translation\ are \ER-hard.
\end{restatable}

%%%%%%%%%%%%%%%%%%%%%%%%%%%%%%
\section{Fingerprinting}
\label{sec:FingerPrinting}
%%%%%%%%%%%%%%%%%%%%%%%%%%%%%%

In this section, we develop a technique to argue that pieces are roughly
at the position where we intend them to be. 
The high level idea is based on a few properties.
First, the slack $\slack$, i.e., the difference between the area of the container
and the total area of the pieces, is very small.
Second, every piece $p$ has a specific corner $v$ with a unique angle that
fits precisely at one position.
If a piece is placed at a different location than the intended one, the empty space would exceed $\slack$.
In order to make such arguments, we first need to carefully define a few concepts.

\subsubsection*{Motion and Placement}
We encode a rotation by a \emph{rotation matrix}, which is a matrix $M$ of 
the form 
\[M = \begin{pmatrix}
a & -b \\
b & a
\end{pmatrix},\]
with $\det M = 1$.
From a translation $t\in\R^2$ and a rotation $M$, we get a \emph{motion} $m \mydef (t,M)$.
If only translations are allowed, we require that $M$ is the identity.

Given a piece $p$ and a motion 
$m = (M,t)$,
then we denote by $p^m$ the piece $p$ after moving $p$
according to $m$, i.e.,
\[\pl p m \mydef \{Mx + t : x\in p\}.\]
The set $\pl p m$ is the \emph{placement} of $p$ by $m$.
Given a tuple of $\mathbf{p} = (p_1,\ldots,p_k)$ of pieces
and a tuple $\m = (m_1,\ldots,m_k)$ of motions,
then we denote by 
\[\p^\m = (p_1^{m_1},\ldots, p_k^{m_k})\] the \emph{placement} 
of $\p$ by $\m$.
We may write $\pl {p_i}\m$ instead of $\pl {p_i}{m_i}$.

Given a container $\cont$, pieces $\p=(p_1,\ldots,p_k)$ and a motion $\m$,
we say that $\m$ (resp.~$\p^\m$) is a \emph{valid} motion (resp.~placement), if (i) $\pl {p_i}\m\subset C$ for all $i$, and (ii) $\pl {p_i}\m$ and $\pl {p_j}\m$ are interior-disjoint for all $i\neq j$.

\subsubsection*{Other geometric definitions}
Let $ab$ and $cd$ be two (oriented) line segments.
The \emph{angle between} $ab$ and $cd$ is the minimum angle that $ab$ can be turned such that $ab$ and $cd$ become parallel and point in the same direction, i.e., after turning, we should have $(b-a)^\bot \cdot (d-c)=0$ and $(b-a)\cdot (d-c)>0$.

Consider two motions $m_1$ and $m_2$ of a piece $p$.
The \emph{displacement} between $m_1$ and $m_2$ is
$
\sup_{x\in p}\|\pl x{m_1}-\pl x{m_2}\|.
$
The \emph{displacement angle} is the absolute difference in how much $m_1$ and $m_2$ rotate $p$ in the interval $[0,\pi]$.

Given a compact set $S$ in the plane, we denote by $\area(S)$ the area of $S$, and we define the \emph{diameter} of $S$ as $\diam(S)\mydef \max_{a,b\in S}\|ab\|$.

Given a container $\cont$ and pieces \p, we define the \emph{slack} as $\slack=\area(\cont) - \sum_{p\in\p}\area(p)$.

Let $A\subset\R^2$.
Then $A^c=\{x\in\R^2\mid x\notin A\}$.
Let $A,B\subset\R^2$.
Then we denote the Minkowski sum as $A\oplus B\mydef \{a+b\mid a\in A,b\in B\}$ and the Minkowski difference as $A\ominus B\mydef (A^c\oplus B)^c$.
For $\lambda>0$, define $\disk(\lambda)\mydef \{(x,y)\in\R^2\mid x^2+y^2\leq\lambda^2\}$.

\subsection{Fingerprinting a single piece}
\label{sec:fingerprintSingle}

Now let us go one level deeper into the details of the fingerprinting technique; see also \Cref{fig:finger-placingNew}.
We consider a case where we already know the position of some pieces (possibly with some uncertainty), and we consider the empty space $E$ where the remaining pieces must be placed.
Ideally, we could identify a corner $w$ of the empty space and a corner $v$ of a remaining piece $p$ which has exactly the same angle as $w$ and deduce that $p$ must be placed with $v$ at $w$.
Unfortunately, this is not the case, for two reasons.
First, we want to give most pieces some tiny but non-zero amount of wiggle room. 
This is important as pieces are meant to represent variables.
Second, we do not know the precise position of the other pieces as previous fingerprinting steps could only infer approximate and not exact positions of those pieces.
Thus, we will identify for each piece $p$ a triangle $T$, which has a corner $y$ with the same angle as $v$.
The edges adjacent to~$y$ will be very close to but not exactly on the boundary of the empty space $E$. 
The triangle $T$ will be our main protagonist in the forthcoming proofs and formal definitions. 
It may partially overlap existing pieces or have some distance to already placed pieces.
Another key player is the uncertainty value $\lambda\geq 0$, which is a measure of how much $T$ is off from the ideal.
We are now ready to go into the full details of the fingerprinting.

\begin{figure}
\centering
\includegraphics[page=2]{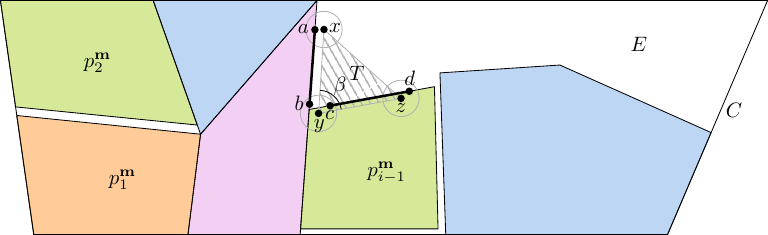}
\caption{We are considering the placement of the pieces $p_1,\ldots,p_{i-1}$ according to a valid motion $\m$.
The white area is the empty space $E$ available for the remaining pieces $p_i,\ldots,p_N$.
The radius of the gray circles centered at $x,y,z$ is the uncertainty value $\lambda$; the first circle must contain $a$, the second $b$ and $c$, and the third $d$, where $ab$ and $cd$ are segments on the boundary of $E$.}
\label{fig:finger-placingNew}
\end{figure}

\subsubsection*{Setup}
We are given a container $\cont$ and pieces $\p=(p_1,\ldots,p_N)$.
Each piece $p\in\p$ is a simple polygon with the following properties.
\begin{itemize}

\item
Each segment of $p$ has length at least $1$.

\item
The diameter of $p$ is at most some number $d_{\max}$. 

\item
The polygon $p$ is \emph{fat} in the following sense.
For any two points $v,w$ on different and non-neighbouring segments of $p$, we have $\|vw\|\geq \tau\mydef 1/100$.
\end{itemize}

In \Cref{sec:curvedPolygons}, we will show that the results developed in the following for polygonal pieces also hold when the pieces are allowed to be curved polygons (provided that the curvature is sufficiently small and the segments do not curve within distance $1$ from the corners).

\subsubsection*{The empty space $E$}
We consider an arbitrary valid motion $\m$ and analyze how we can infer something about the placement of the pieces $p_i,\ldots,p_N$ from the placement of the first $i-1$ pieces $p_1,\ldots,p_{i-1}$.
We can think of this situation as if we have already decided where to place the first $i-1$ pieces $p_1,\ldots,p_{i-1}$ in $\cont$ so that they are interior-disjoint, and we are now reasoning about where to place the next piece.

Let
\[
E\mydef \overline{\cont\setminus\bigcup_{j=1}^{i-1} \pl {p_j}\m}
\]
be (the closure of) the uncovered space available for the remaining pieces $p_i,\ldots,p_N$, see \Cref{fig:finger-placingNew}.
Then $E$ is a subset of $\cont$ bounded by a finite number of line segments, and each of these segments is contained in edges of the pieces $\pl {p_1}\m,\ldots,\pl {p_{i-1}}\m$ or $\cont$. 

\subsubsection*{Covering a wedge of $E$}
Assume that there is a special triangle $T\subset \cont$ with corners $x,y,z$ and with the following properties.
We have $\|xy\|=\|yz\|=1$.
Let $\lambda\geq 0$ be a (small) number that will be defined whenever we are going to apply the fingerprinting.
The value $\lambda$ can be thought of as a measure of uncertainty of the already placed pieces and the distance from the boundary of $T$ to the boundary of $E$.
Let $\inn{T}=\inn{x}\inn{y}\inn{z}$ denote the triangle $T\ominus\disk(\lambda)$ (where $\ominus$ is the Minkowski difference and $\disk(\lambda)$ the disk of radius $\lambda$, as defined in the beginning of this section), such that $\inn{x},\inn{y},\inn{z}$ are on the angular bisectors from $x,y,z$, respectively.

\begin{definition}\label{def:ubounding}
We say that $E$ is \emph{$\lambda$-bounding} $T$ at $y$ 
if the following two conditions hold:
\begin{itemize}
\item There are segments $ab$ and $cd$ on the boundary of $E$ such that each of the distances $\|ax\|$, $\|by\|$, $\|cy\|$, $\|dz\|$ is at most $\lambda$.

\item The interior of $\inn{T}$ is a subset of $E$.

\item The angle $\beta$ of $T$ at $y$ satisfies $\beta \in[\alpha_{\min},\alpha_{\max}]$, where $\alpha_{\min}\mydef 5\pi/180$ and $\alpha_{\max}\mydef \pi/2$.
\end{itemize}
\end{definition}

The second requirement means that no piece among $p_1,\ldots,p_{i-1}$ covers anything of $\inn T$ when placed according to $\m$.
The triangle $\inn{T}$ will have an area much larger than $\slack$, which implies that almost all of $\inn{T}$ must be covered by the pieces $\pl {p_i}\m,\ldots,\pl {p_N}\m$.

\subsubsection*{Unique angle property}
We assume that the pieces $p_i,\ldots,p_N$ have the following property, which we denote as the \emph{unique angle property} with respect to the angle $\beta$ of $y$ and a (small) number $\sigma>0$:

Consider any set $S\mydef \{v_1,\ldots,v_m\}$ such that it contains at most one corner from each piece $p_i,\ldots,p_N$.
If the sum of angles of corners in $S$ is in the interval $[\beta-\sigma,\beta+\sigma]$, then $S$ consists of only one corner $v$, i.e., $S=\{v\}$.

In most applications of the fingerprinting technique, there will be only one such set $S=\{v\}$.
In other words, the angle range $[\beta-\sigma,\beta+\sigma]$ uniquely identifies a specific piece and a specific corner $v$ of the piece.
However, in \Cref{sec:SquareContainer}, we are going to consider a special case where the container is a square where there will be more such sets.

We will argue that almost all of $\inn{T}$ must be covered by a piece with a corner $v$ with an angle in the range $[\beta-\sigma,\beta+\sigma]$, and the corner $v$ must be placed close to $y$.
Informally, since $\slack$ is much smaller than the area of $\inn T$, almost all of $\inn T$ must be covered by the pieces $p_i,\ldots,p_N$.
Because of the unique angle property, it is only possible to cover a sufficient amount of $\inn T$ by placing a piece with such a corner $v$ close to $y$ and with the adjacent edges close to parallel to $yx$ and $yz$, since the edges $ab$ and $cd$ of $\partial E$ are preventing $\inn T$ from being covered in another way.

\subsubsection*{Main lemma}
To sum up, we have made these assumptions:
\begin{itemize}
\item
We consider a valid placement $\m$ of the pieces $\p$.

\item
The empty space $E$ is $\lambda$-bounding the triangle $T$ at the corner $y$ of $T$.

\item
The pieces $p_i,\ldots,p_N$ have the unique angle property with respect to the angle $\beta$ of the corner $y$ and the small number $\sigma$.
\end{itemize}
In \Cref{sec:proof:lemma:boundDiff0}, we are going to prove the following lemma in the setting described above.

\begin{lemma}[Single fingerprint]
\label{lemma:boundDiff0}
There is a piece $p \in\{p_i,\ldots,p_N\}$ with a corner $v$ such that the angle of $v$ is in $[\beta-\sigma,\beta+\sigma]$ and 
$\|y\pl {v}\m\|= O\left(\lambda/\sigma+\sqrt{\slack/\sigma}\right)$.

Furthermore, let $u,w$ be the corners preceding and succeeding $v$ in counterclockwise direction, respectively.
Then the angle between $\pl {v}\m\pl {u}\m$ and $yx$ is
 $O\left(\lambda/\sigma+\sqrt{\slack/\sigma}\right)$, 
 as is the angle between $\pl {v}\m\pl {w}\m$ and $yz$.
\end{lemma}

\subsection{Fingerprinting more pieces at once}
\label{sec:fingerprintMorePieces}

In this section, we consider the iterated use of the fingerprinting technique (in particular \Cref{lemma:boundDiff0}) for some number of times.
This describes the situation whenever we have introduced the pieces of a new gadget to the construction.
More precisely, we consider the situation where we know how the pieces $p_1,\ldots,p_{i-1}$ must be placed, and we want to deduce how the following $k$ pieces $p_i,\ldots,p_{i+k-1}$, for some $k\geq 1$, must then be placed.
To this end, consider an arbitrary valid motion $\m$.
Consider a set $\s$ of \emph{intended} motions $s_i,\ldots,s_{i+k-1}$ of the pieces $p_i,\ldots,p_{i+k-1}$.
We are going to define what it means for the intended motions $\s$ to be \emph{sound}, and then we prove that if they are sound, then the valid motion $\m$ must place the pieces $p_i,\ldots,p_{i+k-1}$ in a way similar to the intended motions $\s$.

To define soundness of the intended motions, we first define the empty space $E^\s_j$, for $j\in\{i,\ldots,i+k-1\}$, as
\[
E^\s_j\mydef \overline{\cont\setminus\left(\bigcup_{l=1}^{i-1} \pl {p_l}{m_l}\cup \bigcup_{l=i}^{j-1} \pl {p_l}{s_l}\right)}.
\]
Thus, $E^\s_j$ is the free space where the piece $p_j$ can be placed if the pieces $p_1,\ldots,p_{i-1}$ are placed according to $\m$ while the pieces $p_i,\ldots,p_{j-1}$ are placed according to the intended motions $\s$.

\begin{definition}
\label{def:sound}
We say that the intended motion $s_j$, $j\in\{i,\ldots,i+k-1\}$, is $\lambda$-\emph{sound}, for a value $\lambda\geq 0$, if there exists a triangle $T_j=x_jy_jz_j$ and a corner $v_j$ of $p_j$ such that the following holds,
\begin{itemize}
\item the angle $\beta_j$ of $T_j$ at $y_j$ is in the range $[\alpha_{\min},\alpha_{\max}]$,
\item $\|x_jy_j\|= 1$ and $\|y_jz_j\|= 1$,
\item $E^\s_j$ is $\lambda$-bounding $T_j$ at $y_j$ 
(recall \Cref{def:ubounding}),
\item
if a set $S$ of at most one corner from each of the pieces $p_j,\ldots,p_N$ has a sum of angles in the range $[\beta_j-\sigma,\beta_j+\sigma]$, then $S=\{v_j\}$ (note that this is a stronger version of the unique angle property since the original definition just requires $S$ to be a singleton, while here $S$ must contain a specific corner $v_j$),
\item $T_j\subset \pl {p_j}{s_j}$, and
\item $\pl {v_j}{s_j}=y_j$ and $x_jy_j,y_jz_j\subset\partial \pl {p_j}{s_j}$.
\end{itemize}
We likewise define the placement $\pl {p_j}{s_j}$ to be $\lambda$-\emph{sound} if the motion $s_j$ is $\lambda$-sound.
\end{definition}

\begin{lemma}
\label{lem:fingerprinting-summary}
There exists an absolute constant $c>0$ such that the following holds.
Define
\begin{align*}
\Lambda_i & \mydef 0,\quad\text{and} \\ 
\Lambda_j & \mydef cd_{\max}/\sigma\cdot \Lambda_{j-1}+cd_{\max}(\lambda/\sigma+\sqrt{\slack/\sigma}),
\end{align*}
for $j>i$.
If the motions $s_i,\ldots,s_{i+k-1}$ are $\lambda$-sound, then for each $j\in\{i,\ldots,i+k-1\}$, the displacement between the motions $m_j$ and $s_j$ of the piece $p_j$ is at most $\Lambda_{j+1}$.
It holds that
\[
\Lambda_{k+1}\leq (k+1) (cd_{\max}/\sigma)^{k+1}(\lambda/\sigma+\sqrt{\slack/\sigma}),
\]
which is a bound on all the mentioned displacements.
\end{lemma}

\begin{proof}
We proceed by induction on $j$.
For $j=i$, we apply \Cref{lemma:boundDiff0}.
We get that
$\|y_i\pl {v_i}\m\|\leq O\left(\lambda/\sigma+\sqrt{\slack/\sigma}\right)$.
Furthermore, the second half of \Cref{lemma:boundDiff0} implies that the displacement angle between the motions $m_i$ and $s_i$ is likewise at most $O\left(\lambda/\sigma+\sqrt{\slack/\sigma}\right)$.
We therefore get that the displacement between $m_i$ and $s_i$ is
\[
cd_{\max}(\lambda/\sigma+\sqrt{\slack/\sigma})=\Lambda_{i+1}
\]
for some constant $c$.

Suppose now that the statement holds for indices $i,i+1,\ldots,j-1$.
Define
\[
E^\m_j\mydef \overline{\cont\setminus\bigcup_{l=1}^{j-1} \pl {p_l}{m_l}}.
\]
Since $E^\s_j$ is $\lambda$-bounding $T_j$ at $y_j$ and the displacement between $m_{j-1}$ and $s_{j-1}$ is at most $\Lambda_j$, we get that $E^\m_j$ is $(\Lambda_j+\lambda)$-bounding $T_j$ at $y_j$.
Therefore, \Cref{lemma:boundDiff0} gives that the displacement between $m_j$ and $s_j$ is at most
\[
cd_{\max}((\Lambda_j+\lambda)/\sigma+\sqrt{\slack/\sigma})=cd_{\max}/\sigma\cdot \Lambda_j+cd_{\max}(\lambda/\sigma+\sqrt{\slack/\sigma})=\Lambda_{j+1},
\]
for the constant $c$ introduced above.
Unfolding the expression, we get
\begin{align*}
\Lambda_{k+1} & =\sum_{j=0}^k (cd_{\max}/\sigma)^j\cdot cd_{\max}(\lambda/\sigma+\sqrt{\slack/\sigma}) \\
& \leq \sum_{j=0}^k (cd_{\max}/\sigma)^{j+1}\cdot (\lambda/\sigma+\sqrt{\slack/\sigma}) \\
& \leq (k+1) (cd_{\max}/\sigma)^{k+1}(\lambda/\sigma+\sqrt{\slack/\sigma}). 
\end{align*}
\end{proof}

The following lemma will be used to fingerprint the pieces in each gadget individually.

\begin{lemma}[Multiple fingerprints]
\label{lem:multipleFingerprints}
Consider a gadget together with its pieces $p_i,\ldots,p_{i+k-1}$, for $k\leq 7$, which are introduced in some step of the construction.
Suppose that there is a valid motion $\m$ of the complete construction.
Furthermore, suppose that there exists intended motions $s_i,\ldots,s_{i+k-1}$ which are $\gadgets\slack$-sound. 
Then the displacement between the intended motion $s_j$ and the actual motion $m_j$ is of the order $O(n^{-48})$.
\end{lemma}

\begin{proof}
In our construction, we have $d_{\max}=O(n^4)$, $\slack=O(n^{-296})$, and $g=O(n^4)$.
We use the method described in the proof of \Cref{lem:UniqueAngles} to choose unique angles.
As our reduction results in a packing instance of $N=O(g)=O(n^4)$ pieces, we get the unique angle condition satisfied for a value of $\sigma$ of the order $\Omega(N^{-2})=\Omega(n^{-8})$.
We now get from \Cref{lem:fingerprinting-summary} with $\lambda=g\slack=O(n^{-292})$ that the displacement is at most
\begin{align*}
8 (cd_{\max}/\sigma)^8(\lambda/\sigma+\sqrt{\slack/\sigma})
=\,& O((n^4n^8)^8(n^{-292}n^8+\sqrt{n^{-296}n^8})) \\
=\,& O(n^{96}n^{-144})=O(n^{-48}). 
\end{align*}
\end{proof}

\subsection{Proof of Single fingerprint \texorpdfstring{(\Cref{lemma:boundDiff0})}{(Lemma \ref{lemma:boundDiff0})}}\label{sec:proof:lemma:boundDiff0}

\begin{figure}
\centering
\includegraphics[page=3,width=\textwidth]{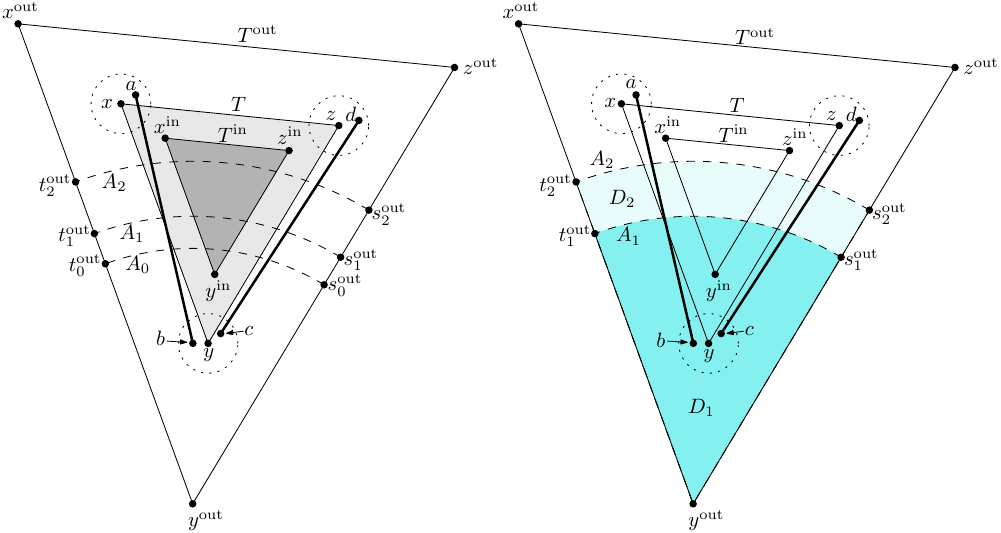}
\caption{Setup in the proof of \Cref{lemma:boundDiff0}.
The circles centered at $x,y,z$ bound the disks of radius $\lambda$ in which the points $a,b,c,d$ are known to be.
The figure is not to scale.
In practice, $\lambda$ is very small so that the triangles $\out{T},T,\inn{T}$ are almost equally large.
Furthermore, the radius of $A_2$ is much larger than the radius of $A_1$, which in turn is much larger than the radius of $A_0$.
The bottom figure shows the regions $D_1$ and $D_2$, that together make up~$D$.
The segments $ab$ and $cd$ are drawn with thick lines to indicate that these act as a restriction to where the next piece can be placed.}
\label{fig:actualPlacementT}
\end{figure}

\subsubsection*{Proof setup}
We prove auxiliary 
\Cref{lemma:zeta,lemma:wedges,lemma:stickingOut,lemma:areaBounds,lemma:covering1,lemma:covering2}
and then show \Cref{lemma:boundDiff0} (restated as \Cref{lemma:boundDiff01}).
See \Cref{fig:actualPlacementT}.
Let $\zeta(\theta)\mydef \frac{1}{\sin (\theta/2)}$.
The function $\zeta$ is important when computing the distances between corresponding corners of offset versions of the same triangle, as the following lemma makes clear.

\begin{lemma}\label{lemma:zeta}
(1) 
Consider a triangle $U=efg$ and define for some $s>0$ the triangle $\inn{U}\mydef U\ominus \disk(s)=\inn{e}\inn{f}\inn{g}$, so that $\inn{e},\inn{f},\inn{g}$ are on the angular bisectors of $e,f,g$, respectively.
Then $\|e\inn{e}\| = s \zeta(\theta)$, where $\theta$ is the angle of $U$ at $e$.

(2)
Let $e,f,g$ be points such that the distances from a point $\inn{e}$ to each of the segments $ef$ and $eg$ is at most $s$.
Then $\|e\inn{e}\|\leq s\zeta(\theta)$, where $\theta\in[0,\pi)$ is the angle between $ef$ and $eg$.
\end{lemma}

\begin{proof}
Proof of (1): Let $p$ be the projection of $\inn{e}$ on $ef$.
Then $\|e\inn e\|=\|p\inn e\|/\sin(\theta/2)=s\zeta(\theta)$.

Proof of (2): For a given angle $\theta$, the distance $\|e\inn{e}\|$ is maximum if the distances from $\inn{e}$ to each segment $ef$ and $eg$ are both $s$, so that, in particular, $\inn{e}$ is on the angular bisector between the segments $ef$ and $eg$.
We then proceed as in the proof for (1).
\end{proof}

\Cref{lemma:zeta} gives that $\|y\inn{y}\|=\lambda \zeta(\beta)$.
Let $\out{T}=\out{x}\out{y}\out{z}$ be the triangle we get by offsetting the edges of $T$ outwards in a parallel fashion by distance $\lambda\cdot \zeta(\alpha_{\min})$, i.e., $\out{T}$ is the triangle such that $\out{T}\ominus \disk(\lambda\cdot \zeta(\alpha_{\min}))=T$.
The corners $\out{x},\out{y},\out{z}$ are on the angular bisectors of $x,y,z$, respectively.

Define
$\psi  \mydef \zeta(\alpha_{\min})+\zeta(\alpha_{\min})^2=O(1)$.
By \Cref{lemma:zeta}, we have $\|y\out{y}\|=\lambda \zeta(\alpha_{\min})\zeta(\beta) \leq \lambda \zeta(\alpha_{\min})^2$, and $\|\out{y}\inn{y}\|=\lambda(\zeta(\beta)+ \zeta(\alpha_{\min})\zeta(\beta))\leq \lambda\psi$.

\subsubsection*{Subdividing $\out{T}$ by arcs}
For some small constants $c_1,c_2>0$, we define three radii as
\begin{align*}
r_0 & \mydef \lambda\psi, \\
r_1 & \mydef c_1\left(\lambda/\sigma+\sqrt{\slack/\sigma}\right), \\
r_2 & \mydef c_2\left(r_1/\sigma\right).
\end{align*}
We require that $r_2$ is much smaller than $1$, say $r_2< 1/10$ (as it turns out, by choosing $\delta$ small enough, we can make $\slack$ and $\lambda$ so small that $r_2$ is below any desired constant).
For the ease of presentation, we will not explicitly specify the constants $c_1,c_2$, but it will follow from the analysis that constants exist that will make the arguments work.
Note that in our application of \Cref{lem:multipleFingerprints}, we will have $\sigma=\Theta(n^{-8})$.
Furthermore, one should think of $r_1$ as much larger than $r_0$ and $r_2$ as much larger than $r_1$.

Let $A_i$ be the arc with center $\out{y}$ and radius $r_i$ from the point $\out{s_i}$ on segment $\out{y}\out{z}$ counterclockwise to the point $\out{t_i}$ on segment $\out{x}\out{y}$.
Let $D$ be the region bounded by segments $\out{t_2}\out{y}$ and $\out{y}\out{s_2}$ and the arc $A_2$.
The arc $A_1$ separates $D$ into two regions $D_1$ and $D_2$, where~$A_2$ appears on the boundary of $D_2$.

\begin{figure}
\centering
\includegraphics[page=7]{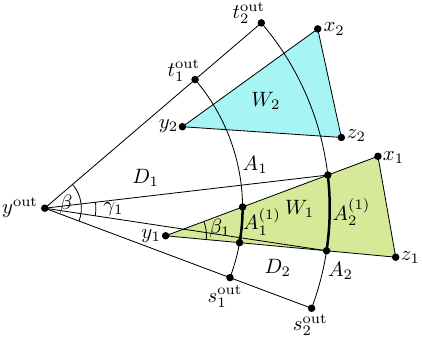}
\caption{The setup of \Cref{lemma:wedges}.
}
\label{fig:lemma}
\end{figure}

\subsubsection*{Geometric core lemma}
The following lemma is the geometric core of our argument and the setup is shown in \Cref{fig:lemma}.
We use this lemma to conclude that if a set $Q$ of pieces cover most of $D_1$, then they have corners whose angles sum to a number close to $\beta$.
It then follows from the unique angle property that $Q$ consists of just one piece.

\begin{lemma}\label{lemma:wedges}
Let $\{W_i=x_iy_iz_i\mid i=1,\ldots,m\}$ be a collection of triangles where each $y_i$ is a corner in $D_1$ and the edge $x_iz_i$ is disjoint from $D$.
Let $\beta_i$ be the angle of $W_i$ at $y_i$ and suppose that $\beta_i\in [\alpha_{\min},\alpha_{\max}]$.
Suppose that $W_1,\ldots,W_m$ are pairwise interior disjoint and that the interior of each $W_i$ is disjoint from $\out{t_2}\out{y}$ and $\out{y}\out{s_2}$.
Let $\mathcal W\mydef\bigcup_{i=1}^m W_i$, and let $\rho\in[0,1]$ be the fraction of $A_2$ covered by $\mathcal W$.
Then
\begin{align}
& \frac{\area(D_1\cap \mathcal W)}{\area(D_1)}\leq \rho+O(r_1/r_2),\text{ and}\label{eq:arearho} \\
& \sum_{i=1}^m \beta_i\in [\beta\rho-O(r_1/r_2),\beta\rho+O(r_1/r_2)]. \label{eq:anglesum}
\end{align}
\end{lemma}

\begin{proof}
We first analyze just a single triangle $W_i$ and then generalize to all of $W_1,\ldots, W_m$.
For $j\in\{1,2\}$, let $A_j^{(i)}\mydef A_j\cap W_i$ be the arc on $A_j$ contained in $W_i$, and let $\gamma_i\in (0,\beta]$ be the angle spanned by $A_2^{(i)}$.
We claim that
\begin{align}
&\beta_i\in [\gamma_i-O(r_1/r_2),\gamma_i+O(r_1/r_2)],\text{ and}\label{eq:betagamma}\\
&\area(D_1\cap W_i)\leq r_1^2\gamma_i/2+O(r_1^3/r_2). \label{eq:areaWi}
\end{align}

Before proving~\eqref{eq:betagamma} and~\eqref{eq:areaWi}, we show how~\eqref{eq:arearho} and~\eqref{eq:anglesum} follow.
We get from~\eqref{eq:areaWi} that
\[
\frac{\area(D_1\cap W_i)}{\area(D_1)}\leq \frac{r_1^2\gamma_i/2+O(r_1^3/r_2)}{r_1^2\beta/2}=\gamma_i/\beta+O(r_1/r_2).
\]

Note that by~\eqref{eq:betagamma} and since $\beta_i\in[\alpha_{\min},\alpha_{\max}]$ for all $i$, we know that the number of triangles is $m=O(1)$.
We now have that
\[
\frac{\area(D_1\cap \mathcal W)}{\area(D_1)}=\sum_{i=1}^m \frac{\area(D_1\cap W_i)}{\area(D_1)}\leq \sum_{i=1}^m \left(\gamma_i/\beta+O(r_1/r_2)\right)=\rho+O(r_1/r_2),
\]
which proves~\eqref{eq:arearho}.
Likewise,~\eqref{eq:anglesum} follows from~\eqref{eq:betagamma} using that $m=O(1)$ as
\[
\sum_{i=1}^m \beta_i\in \left[\sum_{i=1}^m \left(\gamma_i-O(r_1/r_2)\right),\sum_{i=1}^m \left(\gamma_i+O(r_1/r_2)\right)\right]=
[\rho\beta-O(r_1/r_2),\rho\beta+O(r_1/r_2)].
\]

\begin{figure}
\centering
\includegraphics[page=9]{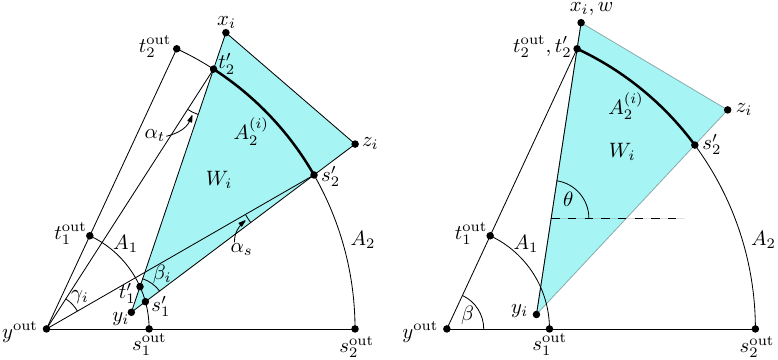}
\caption{Left: In the shown example, $\beta_i=\gamma_i-\alpha_s+\alpha_t$.
Since $\alpha_s=O(r_1/r_2)$ and $\alpha_t=O(r_1/r_2)$, we always have $\beta_i=[\gamma_i-O(r_1/r_2),\gamma_i+O(r_1/r_2)]$.
Right: The angle $\theta$ is an argument of the segment $y_iw$.
The argument is bounded by $\beta+O(r_1/r_2)$.
}
\label{fig:lemma2}
\end{figure}

For the following proof of~\eqref{eq:betagamma}, we refer to \Cref{fig:lemma2} (left).
Let $s'_1\in y_iz_i$ and $t'_1\in y_ix_i$ be the endpoints of $A_1^{(i)}$ and define $s'_2$ and $t'_2$ similarly as the endpoints of $A_2^{(i)}$.
Note first that if $y_i=\out{y}$, we have $\beta_i=\gamma_i$.
However, in general $y_i$ is just a point within distance $r_1$ from $\out{y}$.
Therefore, the angle $\alpha_s$ between the segments $\out{y}s'_2$ and $y_is'_2$ is at most $\arcsin(r_1/r_2)$, which is obtained when $\out{y}y_is'_2$ is a triangle with $\|\out{y}y_i\|=r_1$ and a right angle at $y_i$.
Since $r_2$ is much larger than $r_1$, we have $\arcsin(r_1/r_2)=O(r_1/r_2)$.
Similarly, the angle $\alpha_t$ between the segments $\out{y}t'_2$ and $y_it'_2$ is at most $\arcsin(r_1/r_2)=O(r_1/r_2)$.
Note that $\beta_i=\gamma_i\pm \alpha_s\pm \alpha_t$, so we get $\beta_i\in [\gamma_i-O(r_1/r_2),\gamma_i+O(r_1/r_2)]$.

By the \emph{argument} of a line $\ell$, we mean the counterclockwise angle from the $x$-axis to $\ell$.
The argument of a line segment $s$ is the argument of the line containing $s$.
Assume without loss of generality that $\out{y}\out{s_2}$ is horizontal with $\out{s_2}$ to the right of $\out{y}$, so that the argument of any line through $\out{y}$ and a point on $A_2$ is in the range $[0,\beta]$.
We claim that then the argument of every segment $y_iw$, where $w\in W_i$, is in the range $[-O(r_1/r_2),\beta+O(r_1/r_2)]$.
To verify the upper bound, note that the argument of $y_iw$ is maximum if $w=x_i$ and $t'_2=\out{t_2}$, see \Cref{fig:lemma2} (right).
By an argument as the one used in the previous paragraph, we get that the argument can be at most $O(r_1/r_2)$ larger than $\beta$.
The lower bound follows similarly.

\begin{figure}
\centering
\includegraphics[page=8]{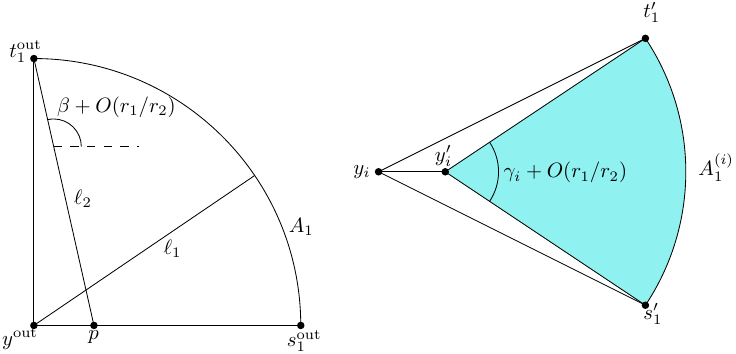}
\caption{Left: Longest segments in $D_1$.
The segment $\ell_1$ shows the case that the argument is in $[0,\beta]$, and then the longest segment has length $r_1$.
The segment $\ell_2$ shows the case where the argument is in $(\beta,\beta+O(r_1/r_2)]$, and then the segment has length $r_1+O(r_1^2/r_2)$.
Right: Figure to show that $\area(D_1\cap W_i)\leq r_1^2\gamma_i/2+O(r_1^3/r_2)$.
The segments $y_is'_1$ and $y_it'_1$ are assumed to have the same length $r_1+O(r_1^2/r_2)$, and the angle $\beta_i$ at $y_i$ is $\gamma_i+O(r_1/r_2)$.
We consider the point $y'_i$ such that $\|y'_is'_1\|=\|y'_it'_1\|=r_1$.
Then the angle at $y'_i$ is likewise $\gamma_i+O(r_1/r_2)$, so the area of the blue triangle is $r_1^2\gamma_i/2+O(r_1^3/r_2)$.
The white triangles have area at most $O(r_1^3/r_2)$, so the total area of $D_1\cap W_i$ is at most $r_1^2\gamma_i/2+O(r_1^3/r_2)$.
}
\label{fig:lemma2NEW}
\end{figure}

We now observe that each of the segments $y_is'_1$ and $y_it'_1$ has length at most $r_1+O(r_1^2/r_2)$, as follows.
See \Cref{fig:lemma2NEW} (left).
Since $\beta\leq\pi/2$, the longest segment $\ell_1$ in $D_1$ with an argument in $[0,\beta]$ connects $\out{y}$ to a point on $A_1$ and has length $r_1$.
The longest segment $\ell_2$ in $D_1$ with an argument in $(\beta,\beta+O(r_1/r_2)]$ 
connects $\out{t_1}$ to a point $p$ on $\out{y}\out{s_1}$.
We then get
\begin{align*}
\|\ell_2\| & = \|\out{t_1}p\| \leq \|\out{t_1}\out{y}\|+\|\out{y}p\| \\
&\leq r_1+r_1\tan(O(r_1/r_2)) = r_1+r_1\frac{\sin(O(r_1/r_2))}{\cos(O(r_1/r_2))} \\
& \leq r_1+r_1\frac{O(r_1/r_2)}{1-O(r_1/r_2)}=r_1+O\left(\frac{r_1^2}{r_2-r_1}\right)=r_1+O(r_1^2/r_2),
\end{align*}
where the last equality follows since $r_2$ is much larger than $r_1$.
Similarly, the longest segment in~$D_1$ with an argument in $[-r_1/r_2,0)$ connects $\out{s_1}$ to a point on $\out{y}\out{t_1}$ and has length $r_1+O(r_1^2/r_2)$.
Hence, $r_1+O(r_1^2/r_2)$ is also an upper bound on the length of $y_is'_1$ and $y_it'_1$.

We get an upper bound on $\area(D_1\cap W_i)$ in the case that $\beta_i$ and the edges $y_is'_1$ and $y_it'_1$ all reach the upper bounds.
This might not be realizable, but still provides an upper bound.
See \Cref{fig:lemma2NEW} (right).
If the edges have length $\|y_is'_1\|=\|y_it'_1\|=r_1$, the area is $r_1^2(\gamma_i+O(r_1/r_2))/2=r_1^2\gamma_i/2+O(r_1^3/r_2)$.
Extending the edges to $\|y_is'_1\|=\|y_it'_1\|=r_1+O(r_1^2/r_2)$, we are adding two triangles each of which has area at most $r_1\cdot O(r_1^2/r_2)=O(r_1^3/r_2)$, and the desired bound~\eqref{eq:areaWi} follows.
\end{proof}

\begin{figure}
\centering
\includegraphics[page=4]{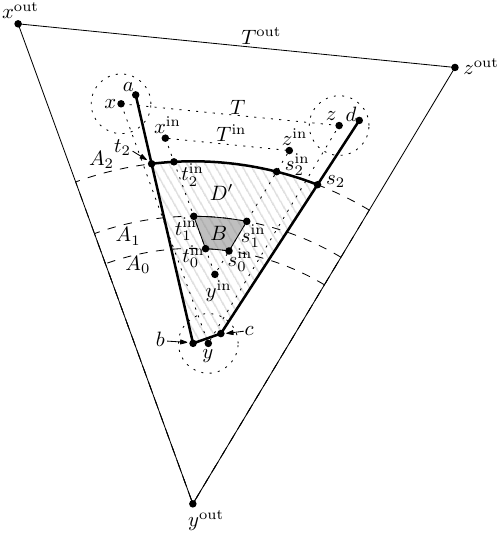}
\caption{Regions $B$ and $D'$.
The region $D'$ is drawn with a pattern of falling gray lines.
The fat part of~$A_2$ is the arc $A'_2$ that bounds $D'$.
A piece covering part of $B$ has edges crossing one or both of $A'_2$ and~$bc$.
The figure is not to scale---in practice $B$ will cover almost all of the region $D_1$ below the arc $A_1$.
}
\label{fig:regionDBNEW}
\end{figure}

We want to apply \Cref{lemma:wedges} to a set of pieces covering parts of $D_1$; see \Cref{fig:regionDBNEW}.
Let $\inn{s_i}$ and $\inn{t_i}$ be the intersection points of $A_i$ with $\inn{y}\inn{z}$ and $\inn{x}\inn{y}$, respectively.
Let $B\subset D_1\cap\inn{T}$ be the region bounded by segments $\inn{s_0}\inn{s_1}$, $\inn{t_0}\inn{t_1}$, and the arcs $A_0\cap\inn{T}$ and $A_1\cap\inn{T}$.
We consider only pieces that cover a part of $B$.
The reason we do not consider all pieces covering a part of~$D_1$ is that a piece covering a part of $D_1$ but not $B$ might violate the assumptions of \Cref{lemma:wedges}.
In particular, such a piece might not have an interior disjoint from $\out{t_2}\out{y}$ and $\out{y}\out{s_2}$, as is seen in case~(ii.a) in \Cref{fig:stickingFig2}.
Cases~(i.b) and~(ii.b) in \Cref{fig:stickingFig,fig:stickingFig2}, 
respectively, show the possibilities of a piece covering a part of $B$, and here the piece fits the assumptions of \Cref{lemma:wedges}.
The following lemma makes this intuition precise.
Conceivably, there may be some pieces that fulfill the conditions of \Cref{lemma:wedges} but do not cover a part of $B$.
However, even considering only pieces covering a part of $B$, we will be able to arrive at our desired conclusion.
Recall that the segments $ab$ and $cd$ are bounding some pieces (or the container $\cont$) in the valid motion $\m$, so these segments act as obstacles that restrict the placement of a piece $\pl{p_j}\m$ covering a part of $B$.

\begin{lemma}
\label{lemma:stickingOut}
A piece $\pl {p_j}\m$ covering a part of the interior of $B$ has the following properties:
\begin{itemize}
\item There is a corner $\pl v\m$ of $\pl {p_j}\m$ contained in $D_1$.

\item The edges of $\pl {p_j}\m$ adjacent to $\pl v\m$ cross $A_2$.
\end{itemize}
\end{lemma}

\begin{figure}
\centering
\includegraphics[page=5,width=\textwidth]{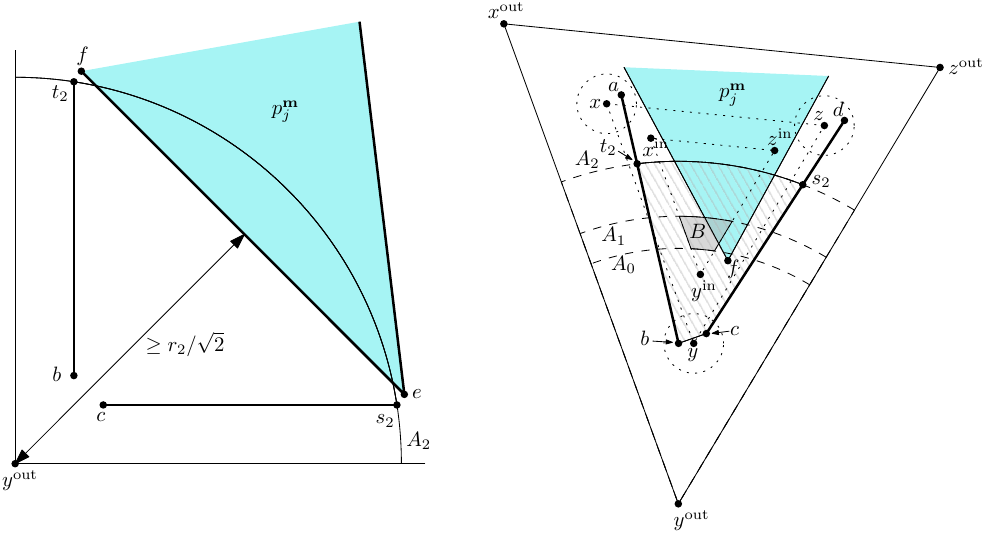}
\caption{
Cases from the proof of \Cref{lemma:stickingOut}.
Left: Case (i.a). In this case, the piece $\pl{p_j}\m$ is too far from $\out{y}$ to cover any of $B$.
Right: Case (i.b). This case agrees with the statement of the lemma.}
\label{fig:stickingFig}
\end{figure}

\begin{figure}
\centering
\includegraphics[page=6,width=\textwidth]{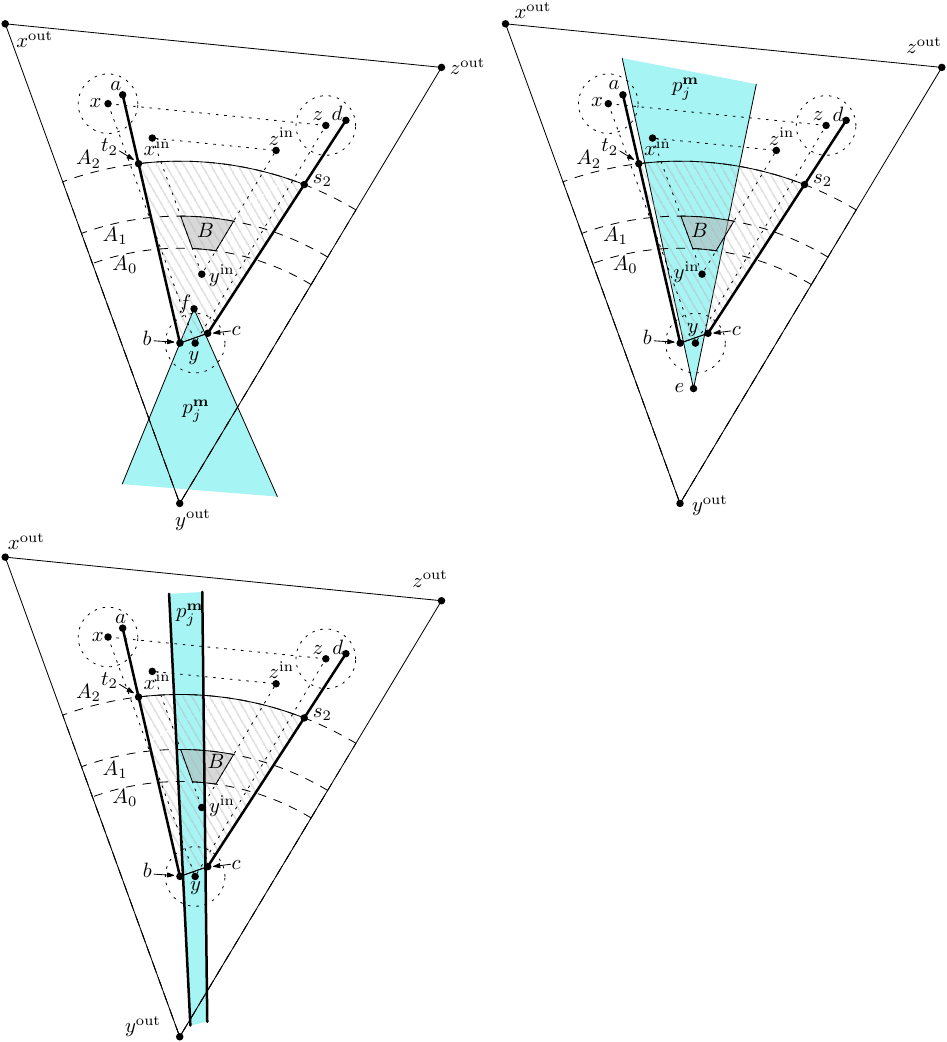}
\caption{
Cases from the proof of \Cref{lemma:stickingOut}.
Top left: Case (ii.a). In this case, $\pl{p_j}\m$ cannot cover any of $B$.
Top right: Case (ii.b). This case agrees with the statement of the lemma.
Bottom left: Case (ii.c). This case violates the fatness assumption of the pieces.}
\label{fig:stickingFig2}
\end{figure}

\begin{proof}
Let $s_2$ and $t_2$ be the intersection points of $A_2$ with segment $cd$ and $ab$, respectively, as shown in \Cref{fig:regionDBNEW}.
Let $A'_2$ be the part of $A_2$ from $s_2$ counterclockwise to $t_2$.
Let $D'$ be the region bounded by segments $s_2 c$, $bc$, $t_2b$, and the arc $A'_2$.
Then $B\subset D'\subset D$.
Since $r_2<1/2$, the diameter of $D$ is less than $1$.
Since $\pl {p_j}\m$ covers a part of the interior of $B$, there is one or more edges of $\pl {p_j}\m$ that cross the boundary of $D'$.
An edge of $\pl {p_j}\m$ can only cross the boundary of $D'$ at a point on the segment $bc$ or the arc $A'_2$, since the segments $cs_2$ and $bt_2$ are bounding some other pieces.
We divide into the following cases, which are also shown in 
\Cref{fig:stickingFig,fig:stickingFig2}, 
 respectively:
\begin{itemize}
\item Case (i):
No edge of $\pl{p_j}\m$ crosses $bc$.
Then there is an edge $ef$ crossing $A'_2$.
We have the following two cases:
\begin{itemize}
\item Case (i.a):
The edge $ef$ crosses $A'_2$ twice.
Since $\beta\leq\pi/2$, we get that the distance from $\out{y}$ to $ef$ is at least $r_2/\sqrt 2$.
Since $r_1$ is much smaller, this edge cannot contribute to covering a part of $B$, so there must be some edges of $\pl{p_j}\m$ crossing the boundary of $D'$ that do not belong to this case.

\item Case (i.b):
One of the endpoints $e$ and $f$ is inside $D'$ while the other is outside.
Assume without loss of generality that $f$ is inside.
It follows that the succeeding edge $fg$ likewise intersects $A'_2$ due to the minimum length of the edges, and the claim holds.
\end{itemize}

\item Case (ii):
An edge $ef$ of $\pl{p_j}\m$ crosses $bc$.
Suppose that as we follow $ef$ from $e$ to $f$, we enter~$D'$ as we cross $bc$.
In particular $e\notin D'$.
There must likewise be another edge $gh$ of~$\pl {p_j}\m$ crossing $bc$, since otherwise, the interior of $\pl{p_i}\m$ would intersect $ab$ or $cd$.
We have the following cases depending on whether an endpoint of $ef$ coincides with one of $gh$:
\begin{itemize}
\item Case (ii.a):
$f$ coincides with an endpoint of $gh$.
Assume without loss of generality that $f=g$.
By \Cref{lemma:zeta} part (1), we have $\|\out{y} y\| \le \lambda \zeta(\alpha_{\min})^2$, and by part (2), we have $\|y f\|\leq \lambda \zeta(\alpha_{\min})$.
Therefore $\|\out{y} f\|\leq \|\out{y}y\|+\|y f\|\leq r_0$.
But then the edges $ef$ and $gh$ do not get far enough into $D'$ so that the wedge they form can cover a part of $B$, as every point in $B$ has distance at least $r_0$ to $\out{y}$.
Therefore, there must be some edges of $\pl{p_j}\m$ crossing the boundary of $D'$ that do not belong to this case.

\item Case (ii.b):
$e$ coincides with an endpoint of $gh$.
Assume without loss of generality that $e=g$.
Because the angle at $e$ is at least $\alpha_{\min}$, it follows from \Cref{lemma:zeta} part (2) that $\|y e\|\leq \lambda \zeta(\alpha_{\min})$.
Therefore, $e$ is contained in $D$.
It then follows that $ef$ and $gh$ both cross $A_2$, since they must exit $D$, and the claim of the lemma thus holds.

\item Case (ii.c):
No endpoint of $ef$ coincides with one of $gh$.
Since $\|bc\|\leq 2\lambda<\tau$ and both segments $ef$ and $gh$ cross $\|bc\|$, we conclude that the fatness condition is violated in this case.
\end{itemize}
\end{itemize}
\end{proof}

Here, we give an informal description of the following three lemmas.
\Cref{lemma:areaBounds} states that the area of $B$ grows quadratically in $r_1$ while $D_1\setminus B$ grows only linearly.
\Cref{lemma:covering1} says that almost all of $B$ must be covered by pieces, as the uncovered area will otherwise be larger than $\slack$.
We are then able to conclude in \Cref{lemma:covering2} that almost all of $D_1$ must be covered by the pieces covering $B$, as $B$ has asymptotically the same area as $D_1$.
This eventually makes it possible to apply \Cref{lemma:wedges} in the proof of \Cref{lemma:boundDiff0}.

\begin{lemma}
\label{lemma:areaBounds}
We have
\begin{itemize}
\item
$\area(D_1\setminus B)=O(r_0^2+\lambda r_1)$.
\item
$\area( B)=\Omega(r_1^2-r_0^2)$. 
\end{itemize}
\end{lemma}

\begin{proof}
The points in $D_1\setminus B$ are either within distance $r_0$ from $\out{y}$ or within distance $\lambda (1+\zeta(\alpha_{\min}))=O(\lambda)$ from one of the line segments $\out{y}\out{s_1}$ or $\out{y}\out{t_1}$, each of length $r_1$.
It then follows that $\area(D_1\setminus B)=O(r_0^2+\lambda r_1)$.

We thus have $\area( B)=\area(D_1)-\area(D_1\setminus B)=\Omega(r_1^2)-O(r_0^2+\lambda r_1)=\Omega(r_1^2-r_0^2)$.
\end{proof}

Let $Q\mydef \{p_j | j\in\{i,\ldots,n\}\text{ and }\pl{p_j}\m\cap  B\neq\emptyset\}$, and let $\mathcal Q\mydef \bigcup_{p_j\in Q}\pl{p_j}\m$.

\begin{lemma}
\label{lemma:covering1}
By choosing $r_1\mydef \Omega\left(\lambda /\sigma+\sqrt{\slack/\sigma}\right)$, where $\Omega$ hides a sufficiently large constant, we get
$\area( B\cap\mathcal Q)\geq (1-\sigma/4)\area( B)$.
\end{lemma}

\begin{proof}
Recall that the pieces $\pl{p_1}\m,\ldots,\pl{p_{i-1}}\m$ are interior disjoint from $\inn{T}$, as $\inn{T}\subset \pl{E_i}\m$.
Since
\[
r_1=\Omega\left(\lambda /\sigma+\sqrt{\slack/\sigma}\right)=\Omega\left(r_0+\sqrt{\frac{\slack}{\sigma}}\right)=\Omega\left(\sqrt{r_0^2+\frac{\slack}{\sigma}}\right),
\]
we get from \Cref{lemma:areaBounds} that
\[
\area( B)= \Omega(r_1^2-r_0^2)=\Omega(\slack/\sigma).
\]
Now, if the constant hidden in the $\Omega$-notation is large enough, we have $\area( B) \geq \frac{\slack}{\sigma/4}$, or equivalently, $\sigma/4\cdot\area( B)\geq\slack$.
This means that the area of $ B$ covered by the pieces in $Q$ is at least $(1-\sigma/4)\area( B)$, as otherwise the uncovered part would be larger than $\slack$.
\end{proof}

\begin{lemma}
\label{lemma:covering2}
By choosing $r_1\mydef \Omega\left(\lambda /\sigma+\sqrt{\slack/\sigma}\right)$, where $\Omega$ hides a sufficiently large constant, we get $\area(D_1\cap \mathcal Q)\geq (1-\sigma/2)\area(D_1)$.
\end{lemma}

\begin{proof}
Note that $r_1=\Omega(\lambda /\sigma)=\Omega(r_0/\sqrt\sigma+\lambda/\sigma)$.
We get $r_1^2=\Omega\left(\frac{r_0^2+\lambda r_1}{\sigma}\right)$, where we can choose the constant hidden in the $\Omega$-notation as big as needed.
Since $r_1=\Omega(r_0)$, we then get $r_1^2-r_0^2=\Omega\left(\frac{r_0^2+\lambda r_1}{\sigma}\right)$.
Now, if we choose the constant big enough, we get from \Cref{lemma:areaBounds} that $\sigma/4\cdot \area( B)\geq\area(D_1\setminus B)$ and thus
\[
\frac{\area (D_1)}{\area ( B)}=\frac{\area ( B)+\area(D_1\setminus B)}{\area ( B)}\leq 1+\sigma/4.
\]
It follows that
\[
\frac{\area ( B)}{\area (D_1)}=\frac{\area ( B)}{\area ( B)+\area(D_1\setminus B)}\geq \frac{1}{1+\sigma/4}\geq \frac{1-\sigma/2}{1-\sigma/4}.
\]
Hence we have that
\[
(1-\sigma/4)\area( B)\geq (1-\sigma/2)\area(D_1).
\]
The claim now follows from \Cref{lemma:covering1} as
\begin{align*}
\area(D_1\cap \mathcal Q)\geq
\area( B\cap \mathcal Q)\geq
(1-\sigma/4)\area( B)\geq
(1-\sigma/2)\area(D_1).
\end{align*}
\end{proof}

We are now ready to prove \Cref{lemma:boundDiff0}.
We rephrase the lemma as follows.

\begin{lemma}
\label{lemma:boundDiff01}
By choosing $r_2\mydef \Omega\left(r_1/\sigma\right)$, where $\Omega$ hides a sufficiently large constant, we get that~$Q$ consists of just one piece $p_j$, and $p_j$ has a corner $y_j$ such that the angle of $y_j$ is in $[\beta-\sigma,\beta+\sigma]$ and $\pl {y_j}\m\in D_1$.
In particular, $\|y\pl {y_j}\m\|\leq 2r_1=O\left(\lambda /\sigma+\sqrt{\slack/\sigma}\right)$.

Furthermore, let $x_j,z_j$ be the corners preceding and succeeding $y_j$, respectively.
Then the angle between $\pl {y_j}\m\pl {x_j}\m$ and $yx$ is $O\left(\lambda /\sigma+\sqrt{\slack/\sigma}\right)$, as is the angle between $\pl {y_j}\m\pl {z_j}\m$ and $yz$.
\end{lemma}

\begin{proof}
By \Cref{lemma:stickingOut}, each piece $p_j\in Q$ has a corner $y_j$ such that $\pl{y_j}\m$ is contained in $D_1$, and the two adjacent edges cross $A_2$.
For each piece $p_j\in Q$, we consider the triangle $W_j\mydef x_jy_jz_j$, such that $x_jy_j$ and $y_jz_j$ are the edges adjacent to $y_j$.
The triangles $W_j$ now fit in the setup of \Cref{lemma:wedges}.
Let $\rho$ be the fraction of $A_2$ covered by $\mathcal Q$.
We then get by \Cref{lemma:covering2} and \Cref{lemma:wedges} that
\[
1-\sigma/2\leq \frac{\area(D_1\cap\mathcal Q)}{\area(D_1)}\leq\rho+O(r_1/r_2).
\]

\Cref{lemma:wedges} furthermore yields that the sum $S$ of angles of the corners $y_j$ is in the range $[\beta\rho-O(r_1/r_2),\beta\rho+O(r_1/r_2)]$.
Since $1-\sigma/2-O(r_1/r_2)\leq \rho\leq 1$, we get
\begin{align}
\beta\rho-O(r_1/r_2) & \geq (1-\sigma/2-O(r_1/r_2))\beta -O(r_1/r_2)\geq \beta-\sigma\pi/4-O(r_1/r_2),\text{ and}\\
\beta\rho+O(r_1/r_2) & \leq \beta +O(r_1/r_2).
\end{align}

If we now choose $r_2\mydef \Omega\left(r_1/\sigma\right)$, where $\Omega$ hides a sufficiently large constant, we get $S\in[\beta-\sigma,\beta+\sigma]$.
We then get from the unique angle property that $Q$ consists of just one piece~$p_j$.

Since $\|y\pl{y_j}\m\|=O(r_1)$ and $\|xa\|\leq \lambda=O(r_1)$ and $r_1$ is much smaller than $\|xy\|=1$, we get that the angle between $\pl {y_j}\m\pl {x_j}\m$ and $yx$ is $O(r_1/\|yx\|)=O(r_1)$, and likewise for the angle between $\pl {y_j}\m\pl {z_j}\m$ and $yz$.
\end{proof}

\subsection{Generalization to curved polygons}
\label{sec:curvedPolygons}

Let $\gamma\colon [0,L] \longrightarrow \R^2$ be a simple curve parameterized by arc-length and of length $L\geq 1$.
We say that $\gamma$ is a \emph{curved segment} if
\begin{itemize}
\item the prefix $\gamma([0,1])$ and suffix $\gamma([L-1,L])$ are line segments (each of length $1$),
\item $\gamma$ is differentiable, 
\item the \emph{mean curvature} of $\gamma$ at most some small constant $\kappa =O(1)$, i.e., for all $s_1,s_2\in (0,L)$, we have
\[
\|\gamma'(s_1)-\gamma'(s_2)\|\leq \kappa |s_1-s_2|.
\]
\end{itemize}

As an example, consider a simple curve $\gamma$ satisfying the first condition and which is the concatenation of line segments and circular arcs.
Then $\gamma$ is a curved segment as long as each circular arc has radius at least $\kappa$ and each transition from one line segment or circular arc to the next is tangential.

We claim that the above results on fingerprinting also hold when the pieces are curved polygons, each of which has a boundary which is a finite union of curved segments.
The first requirement for a curved segment (i.e., that it has a straight prefix and suffix of length $1$) ensures that near every corner of a curved polygon, the curved polygon behaves as a normal polygon.
Because of this, the setup described in
\Cref{sec:fingerprintSingle,sec:fingerprintMorePieces}
also makes sense for curved polygons.

\begin{figure}
\centering
\includegraphics[page=1]{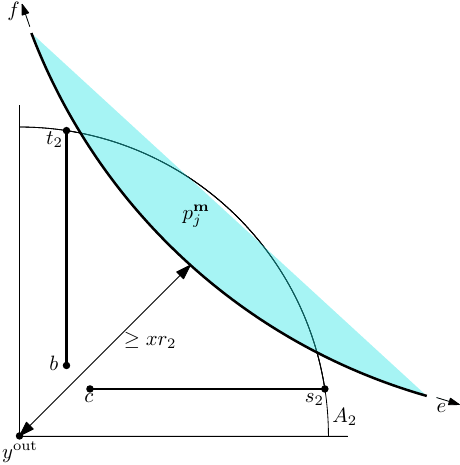}
\caption{
Case (i.a) from the proof of \Cref{lemma:stickingOut} when the pieces can be curved.
Since the curvature is bounded by some small constant $\kappa$, the distance from $\out{y}$ to $ef$ is at least $x\cdot r_2$ for some constant $x$.}
\label{fig:stickingFig2curved}
\end{figure}

We first check that \Cref{lemma:stickingOut} still holds.
Because we require a curved segment to have a straight prefix and suffix of length $1$, the Case~(ii.a) is excluded for the same reason when using curved polygons as when using normal polygons.
Likewise, Case~(ii.c) is excluded because we still have the fatness assumption.
Hence, Case~(i.a) is the only case that should be excluded where the curved segments make a difference, namely the case where a curved segment $ef$ crosses $A'_2$ twice.
Here, the piece can cover a bit more of $D'$ because the segment can curve; see \Cref{fig:stickingFig2curved}.
However, the mean curvature is at most $\kappa=O(1)$ and we can choose $r_2$ to be an arbitrarily small constant.
We therefore still get that the distance from $\out{y}$ to $ef$ is at least $x r_2$ for some constant $x<1/\sqrt 2$ (where, in the version for normal polygons, we had $x=1/\sqrt 2$).
But we have $r_1=O(\sigma\cdot r^2)=O(n^{-8}r^2)$, since $\sigma=\Theta(n^{-8})$, so it follows that it is still impossible that the piece can cover anything of $B$ in this case when $n$ is large enough.

\Cref{lemma:areaBounds,lemma:covering2}
do not use any assumption on the geometry of the pieces, and can thus still be used.
Finally, in the proof of \Cref{lemma:boundDiff01}, we can now apply \Cref{lemma:stickingOut} as before.
Since the prefix and suffix of each curved segment are line segments of length $1$, we can again define triangles $W_j$ and apply \Cref{lemma:wedges} (the geometric core lemma), so the proof goes through unaltered.

%%%%%%%%%%%%%%%%%%%%%%%%%%%%%%
\section{Linear gadgets}
\label{sec:gadgets}
%%%%%%%%%%%%%%%%%%%%%%%%%%%%%%

In this section we are describing four types of gadgets called \emph{anchor}, \emph{swap}, \emph{split}, and \emph{adder}.
They all work with convex polygonal pieces, a polygonal container, and translations.
They also work when rotations are allowed and can thus be used for all packing variants studied in the paper.

For each gadget, we will define canonical placements and verify the four required lemmas of \Cref{sec:overview}.
Here we repeat the properties we need to verify for each gadget.
\begin{itemize}
\item
For every solution to $\Phi$, if the previously added pieces can be placed so that they encode the solution, then the same holds when the pieces of this gadget are added (\Cref{lem:preservation}).

\item
In a valid placement of all the pieces, if the earlier introduced pieces have an aligned placement, then the pieces of this gadget must have an almost-canonical placement (\Cref{lem:AlmostCanonicalPlacement}).

\item
In a valid placement of all the pieces, if the earlier introduced pieces have an aligned placement and the pieces of this gadget have an almost-canonical placement, then the pieces of this gadget must also have an aligned placement (\Cref{lem:alignedPl}).

\item
For each edge $(p_1,p_2)$ of the dependency graph $G_x$, where $p_1$ and $p_2$ are pieces of this gadget, we have $\enc{p_1}\leq \enc{p_2}$, i.e., the value encoded by $p_1$ is at most that encoded by $p_2$ (\Cref{lem:graphIneq}).
\end{itemize}

Some of the steps will be very similar for all the gadgets.
In order to avoid unnecessary repetition, we will handle the first two gadgets, the anchor and the swap, in greater detail than the subsequent gadgets.

%%%%%%%%%%%%%%%%%%%%%%%%%%%%%%
\subsection{Anchor}
\label{sec:anchor}
%%%%%%%%%%%%%%%%%%%%%%%%%%%%%%
Recall that each variable $x$ is represented by two wires $\overrightarrow{x}$ and $\overleftarrow{x}$ in the wiring diagram of the instance $\I$ of \wiredinv which we reduce to a packing instance.
Furthermore, the left endpoints of the wires are vertically aligned and occupy neighbouring diagram lines 
$\ell$ and~$\ell'$, as do the right endpoints.
In our packing instance, we cover each wire with variable pieces that can slide back and forth and thus encode the value of $x$, and the pieces covering one wire are called a \emph{lane}.
In order to make the value represented by the lane on $\overrightarrow{x}$ consistent with that of the lane on $\overleftarrow{x}$, we make an \emph{anchor} at both ends, which will propagate a push from one lane to the other.
Most of this section will be about the anchors at the left ends of the wires.
The anchors at the right ends will be handled in the end of the section.

When we use an anchor in our construction, we also define part of the boundary of the container.
Two of the three introduced pieces are variable pieces that will extend out through the right side of the gadget, and the remaining part of those will be defined as part of another gadget farther to the right, which will be described in other parts of the paper.
It is a general convention in our figures of gadgets that if a part of the boundary of the gadget is drawn with thick full segments, it will be part of the container boundary.
If part of the boundary is drawn with thick dashed segments, it means that the segments can be either part of the container boundary or part of the boundary of other pieces that have been introduced to the construction in earlier steps.

\begin{figure}
\centering
\includegraphics[page = 1]{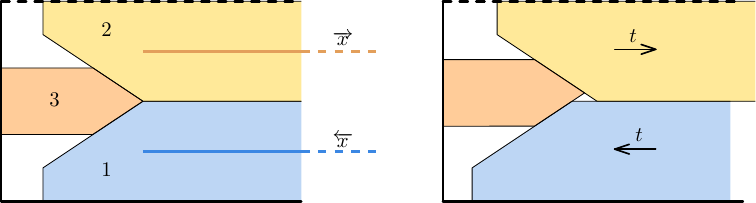}
\caption{Left: A simplified illustration of the anchor and how it is placed on top of the two wires representing a variable $x$.
Right: If the blue piece is pushed to the left, the yellow
piece must move by an equal amount to the right, and vice versa.
Color codes: $1$ blue, $2$ yellow, $3$ orange.
}
\label{fig:anchor-idea}
\end{figure}

\subsubsection*{Simplified anchor}
The anchor is meant to be a connection between the two lanes that represent a variable $x$; see \Cref{fig:anchor-idea} for an illustration.
The gadget consists of part of the boundary of the container and three pieces: orange, yellow, and blue.
The yellow and blue pieces are the two leftmost pieces on the lanes of $\overrightarrow{x}$ and $\overleftarrow{x}$, respectively.
The orange piece functions as a connection between the two lanes.
The idea is that if we move the blue piece to the left by $t$, then we have to move the yellow piece to the right by at least $t$ as well, and vice versa.

The segment bounding the gadget from below is part of the container boundary.
The segment bounding the gadget from above is part of the boundary of a piece introduced in an earlier anchor, except for the very first anchor, which will be bounded from above by the container boundary.

\begin{figure}
\centering
\includegraphics[page = 7]{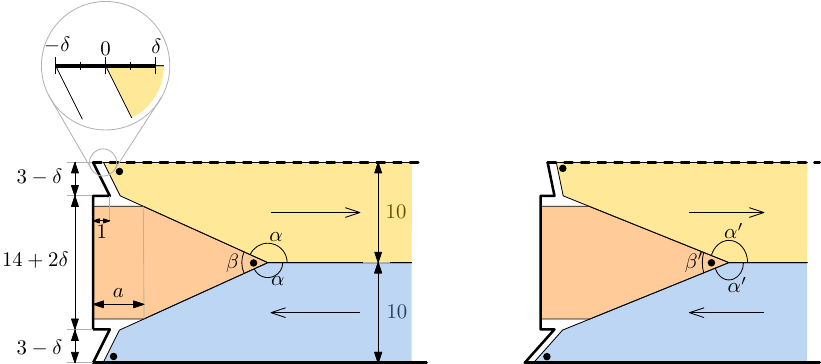}
\caption{Left: An illustration of the anchor gadget.
The placement of the pieces corresponds to the value $x=0$.
The arrows show the orientation of the pieces, and the dots show the fingerprinted corners.
The magnifying glass shows a scale of the correspondence between placements of the yellow piece and the encoded value of the variable $x$.
The length $a$ is $1+O(\delta)$ and depends on $\alpha$.
Right: Another instance of the gadget with other angles chosen.}
\label{fig:anchor-precise}
\end{figure}

\subsubsection*{The actual anchor}
See \Cref{fig:anchor-precise} for an illustration of the following description.
Recall that we need the slack added by each gadget to be only $O(\delta)$, where $\delta\mydef n^{-300}$.
We therefore design the boundary of the anchor to follow the pieces closely.
The yellow and blue pieces are fingerprinted on the boundary, as indicated by the dots.
The orange piece is fingerprinted in the wedge created by the yellow and blue pieces.

Lines that appear axis-parallel must be axis-parallel; those are important for the alignment.
The height of the orange piece is $14$.
The angle $\alpha$ is in the range $[3\pi/4,7\pi/8]$.
The lower bound ensures that a range of size at most $2\delta$ is needed for the orange piece, while the upper bound ensures that the length $a$ is only $1+O(\delta)$.
The angle $\alpha$ and the angles where the yellow and blue pieces are fingerprinted are not fixed, and we therefore have flexibility to choose the angles of the fingerprinted corners freely to obtain the unique angle property; two different choices of angles are shown in \Cref{fig:anchor-precise}.

\subsubsection*{Canonical placements and solution preservation}
As the next step, we define the set of canonical positions for the three pieces.
The yellow and blue pieces are variable pieces, and in the placement shown in \Cref{fig:anchor-precise}, both pieces encode the value $x=0$.
By definition, the canonical placements of each are all placements obtained by sliding the piece to the left or right by distance at most $\delta$ from the placement shown.
Recall that a placement of all pieces of a gadget is canonical if it is (1) valid, (2) canonical for each variable piece, and (3) have certain extra properties defined for each gadget individually.
For the anchor, we specify point (3) as having edge-edge contacts between the pieces and the container boundary as shown in the figure.

We are now ready to prove \Cref{lem:preservation} for the anchor, namely that the reduction preserves solutions.

\begin{proof}[Proof of Lemma \ref{lem:preservation} for the anchor]
Suppose that for a given solution $\mathbf x$ to the \etrinv formula $\Phi$, there exists a canonical placement of the previously introduced pieces $\p_{i-1}$ that encodes that solution.
To extend the placement to the pieces $\p_i$, i.e., with the yellow, blue, and orange piece of this anchor included, we place the yellow and blue pieces so that they represent the value of $x$ in $\mathbf x$.
This leaves room for the orange piece to be placed with the required edge-edge contacts to the blue and yellow pieces.
\end{proof}

\begin{figure}
\centering
\includegraphics[page = 3]{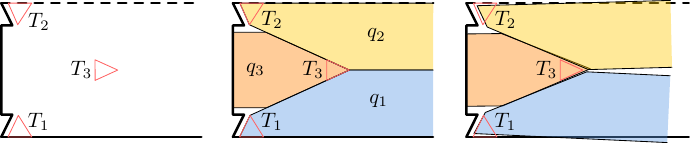}
\caption{The three triangles $T_1,T_2,T_3$ (shown a bit larger than they should be for clarity) are a witness that the intended placements (middle) are $O(\delta)$-sound.
We can therefore conclude that the pieces have an almost-canonical placement, as shown to the right (ignore here that the shown placement is not a valid placement).
}
\label{fig:anchor-Fingertriangle}
\end{figure}

\subsubsection*{Fingerprinting and almost-canonical placement}
We now prove \Cref{lem:AlmostCanonicalPlacement}.
Recall that in this lemma, we assume that the previous pieces $\p_{i-1}$ have an 
aligned $(i-1)\slack$-placement and we want to conclude that 
the new pieces $\p_i \setminus\p_{i-1}$, i.e., the three pieces of this anchor, must have an almost-canonical placement.

\begin{proof}[Proof of Lemma \ref{lem:AlmostCanonicalPlacement} for the anchor]
Since the pieces $\p_{i-1}$ have an 
aligned $(i-1)\slack$-placement, the gadget is bounded from above by an earlier introduced piece (unless $i=1$, in which case it is bounded by the boundary of the container).
We want to use \Cref{lem:multipleFingerprints} (Multiple Fingerprints) to prove that the three pieces have an almost-canonical placement.
For this we need to point out intended placements that are $\gadgets\slack$-sound; recall \Cref{def:sound}.
We choose the placement shown in \Cref{fig:anchor-precise}.
To certify that the intended placements are $\gadgets\slack$-sound, we need to point out three triangles $T_1,T_2,T_3$, such that $\pl {E_j}\s$ is $\gadgets\slack$-bounding $T_j$ for $j\in\{1,2,3\}$.
Here $\pl {E_j}\s$ is the empty space left by the pieces $\p_{i-1}$ and the intended placements of the pieces $q_1,\ldots,q_{j-1}$, where $q_1,q_2,q_3$ are the blue, yellow, and orange piece, respectively.
We choose the triangles as the tips of fingerprinted corners, as shown in \Cref{fig:anchor-Fingertriangle}.
The placements of the blue and yellow pieces are $\delta$-sound because the boundary of the container or the pieces of $\p_{i-1}$ follow the relevant edges of the respective triangles $T_1$ and $T_2$ within distance $\delta$.
The placement of the orange piece is $0$-sound since the boundaries of the blue and yellow pieces contain the edges of $T_3$.
As the intended placements are $\delta$-sound, we conclude using \Cref{lem:multipleFingerprints} (Multiple Fingerprints) that the displacement is $O(n^{-48})$, so the three pieces must have an almost-canonical placement.
\end{proof}

\subsubsection*{Aligned placement}
We show now \Cref{lem:alignedPl}, which assumes that the pieces in the anchor are almost-canonical and all previous pieces $\p_{i-1}$ have an aligned $(i-1)\slack$-placement.
The goal is to conclude that the pieces of this gadget have an aligned $i\slack$-placement, which requires the variable pieces to be correctly aligned and encode values in the range $[-i\slack,i\slack]$.

\begin{proof}[Proof of Lemma \ref{lem:alignedPl} for the anchor]
Consider a valid placement where the pieces~$\p_{i-1}$ have an aligned $(i-1)\slack$-placement and the pieces $\p_i\setminus\p_{i-1}$, i.e., the yellow, blue, and orange pieces of this anchor, have almost-canonical placements; see \Cref{fig:Anchor-aligning}.
We consider the alignment segment $\ell$ which has length $20$.
Since the placements are almost-canonical, it follows that $\ell$ crosses both of the parallel edges of the yellow and the blue pieces.
Since the distance between the segments in each pair is $10$, it follows that the segments must be horizontal.

\begin{figure}
\centering
\includegraphics[page = 5]{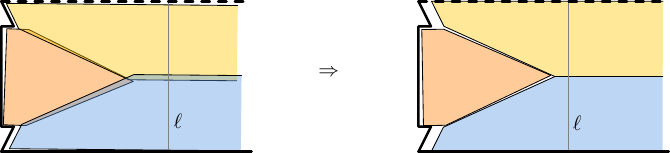}
\caption{The parallel edges of the yellow and blue pieces 
must intersect the alignment segment $\ell$.
Therefore, the segments must be horizontal.}
\label{fig:Anchor-aligning}
\end{figure}

In order to show that the pieces $\p_i$ have an aligned $i\slack$-placement, 
it remains to bound the horizontal displacement of the yellow and blue pieces as compared to the placements encoding the value $0$.
Consider the yellow piece $p$, the argument for the blue piece is similar.
We need to prove $\enc{p}\in[-i\slack,i\slack]$.
We will actually show the stronger statement that 
$\enc{p}\in[-\delta,\slack] \subset [-i\slack,i\slack]$.

Recall that $p$ is right-oriented.
The constraint that $p$ must be inside $\cont$ ensures that $\enc{p} \geq -\delta$.
Since the pieces have an almost-canonical placement by assumption, we know that the displacement as compared to the situation in \Cref{fig:anchor-precise} is at most $n^{-1}$.
It therefore follows that no other pieces than the orange piece can fit in the region to the left of the yellow and blue pieces.
We consider a canonical placement and analyze how much $p$ can slide to the right before too much empty space has been made in the region.
Observe that sliding $p$ to the right by $t$ creates empty space of area $10t$, since the height of the piece is $10$.
It therefore follows that we must have $t\leq\slack/10$.
This translates to 
$\enc{p}\leq \delta +\slack/10 \leq \slack$, and we are done.
\end{proof}

\begin{figure}
\centering
\includegraphics[page = 6]{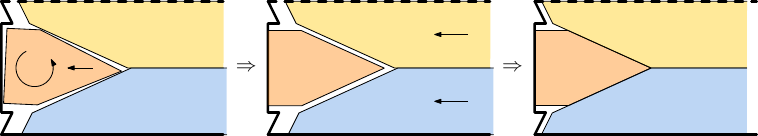}
\caption{If the pieces do not have a canonical placement, there is room for the orange piece to be oriented in the canonical way and get a vertical edge-edge contact with the container.
Then the blue and yellow pieces can be pushed to the left to obtain a canonical placement, which will decrease the value encoded by the yellow piece and increase the value encoded by the blue piece.}
\label{fig:anchor-Pushing}
\end{figure}

\subsubsection*{Edge inequalities}
Recall that \Cref{lem:graphIneq} states that for any edge $(p_1,p_2)$ in the dependency graph $G_x$, we have the inequality $\enc{p_1}\leq \enc{p_2}$, where $\enc{\cdot}$ denotes the value encoded by a piece.
We show this lemma now for the anchor.
\begin{proof}[Proof of Lemma \ref{lem:graphIneq} for the anchor]
Denote by $p_1$ and $p_2$ the blue and yellow piece, respectively.
The pieces induce the edge $(p_1,p_2)$ in $G_x$ and we have to show $\enc{p_1}\leq \enc{p_2}$.
We have that $\enc{p_1}= \enc{p_2}$ when the pieces have an edge-edge contact with the orange piece, as shown in \Cref{fig:anchor-precise}.
We have to show that it is not possible that $\enc{p_1} > \enc{p_2}$. 
This could only potentially happen if the pieces do not have a canonical placement.
However, some straightforward rotation arguments show that in this case we even have $\enc{p_1}< \enc{p_2}$; see \Cref{fig:anchor-Pushing}.
It follows that $\enc{p_1}\leq \enc{p_2}$.
\end{proof}

\begin{figure}
\centering
\includegraphics[page = 2]{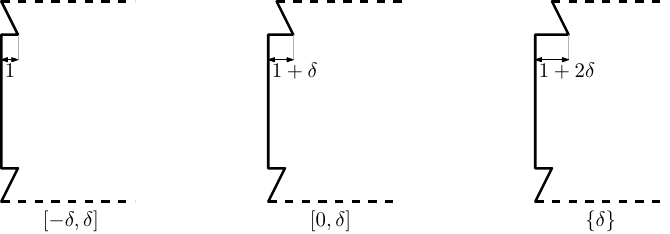}
\caption{The figure shows how we adjust the left wall of an anchor to ensure that we respect the range constraint.}
\label{fig:anchor-range}
\end{figure}

\subsubsection*{Range Constraints}
Recall that together with the variable $x$ is also given an interval $I(x)$ of one of the forms  $[-\delta,\delta],[-0,\delta],\{\delta\}$.
It is claimed in \Cref{lem:consistentCycle} that the pieces representing the variable $x$ encode a consistent value $\enc{K_x}$, which is in the range $I(x)$ due to the design of the anchor.
This is ensured by adjusting the left wall of the anchor as showed in \Cref{fig:anchor-range}, so that the value encoded by the yellow piece is bounded from below according to $I(x)$.
Since the blue piece is restricted to encode a value of at most $\delta$, it then follows that $\enc{K_x}\in I(x)$.

\subsubsection*{Staircases of anchors}
As the next step, we describe how we organize all the anchors of the construction.
Recall that the wires of the wiring diagram appear and disappear in the order $(\overrightarrow{x_1},\overleftarrow{x_1}),\ldots,(\overrightarrow{x_n},\overleftarrow{x_n})$ from left to right in a staircase-like fashion.
Therefore, we also stack the anchors onto one another as displayed in \Cref{fig:Anchor-Stitching}, so that the boundary of each stack appears similar to a staircase.
This ensures that the container will be $4$-monotone, which is used when proving hardness for packing into a square container in \Cref{sec:SquareContainer}.
The rest of the construction will be in between these two staircases.
Anchors in the right staircase are treated in the following paragraph.

\begin{figure}
\centering
\includegraphics[page = 4]{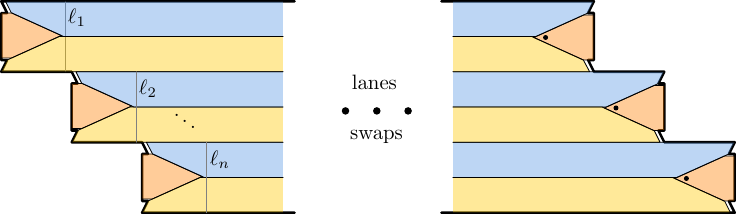}
\caption{Anchors are placed one on top of the other to form two staircases.
The alignment segments $\ell_1,\ldots,\ell_n$ are used to align the blue and yellow pieces of the anchors in this order.
The dots in the orange pieces in the right staircase mark the fingerprinted corners (all other pieces are fingerprinted in other gadgets).
}
\label{fig:Anchor-Stitching}
\end{figure}

\subsubsection*{Anchors in the right staircase}
For an anchor on the right side, the entering yellow and blue variable pieces have been started in other gadgets farther to the left, and we only add the orange piece.
We define the canonical placements in an analogous way as for the left anchors.
The proof of \Cref{lem:preservation} (solution preservation) is trivial since no new variable pieces are introduced.
The orange piece in a right anchor is fingerprinted as in the left anchors, and \Cref{lem:AlmostCanonicalPlacement} follows.
\Cref{lem:alignedPl} is trivial since no new variable pieces are introduced.
The edge inequality of \Cref{lem:graphIneq} is proven as for the left anchors.

%%%%%%%%%%%%%%%%%%%%%%%%%%%%%%
\subsection{Swap}
\label{sec:swap}
%%%%%%%%%%%%%%%%%%%%%%%%%%%%%%

\begin{figure}
\centering
\includegraphics[page = 1]{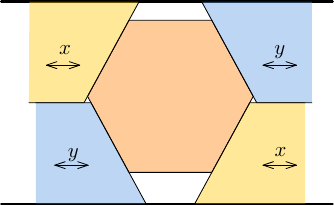}
\caption{The simplified swap. The yellow pieces representing $x$ are either both right- or both left-oriented.
Likewise for the blue pieces representing $y$.
For color codes, see \Cref{fig:swap-precise}.}
\label{fig:swap-idea}
\end{figure}

\subsubsection*{Idea}
Recall that in the wiring diagram, the wires may cross each other (see \Cref{fig:WiringDiagram}).
On top of such a crossing, we build a swap.
The purpose of the swap is thus to make a crossing of two neighbouring lanes of pieces.
To get intuition about how the gadget works, consider \Cref{fig:swap-idea}.
The yellow pieces encode a variable $x$ and the blue pieces encode a variable $y$.
It is possible that $x=y$, which will happen only when the two wires $\overrightarrow x,\overleftarrow x$ cross each other.
Therefore, the yellow pieces will have the opposite orientation of the blue pieces in this special case.

We want to show that when the pieces have edge-edge contacts to the orange piece, the variables are encoded consistently, so that the lanes have been swapped.
The key observation is that if the left blue piece pushes
to the right and the yellow pieces are fixed, then the orange piece
will slide along the yellow pieces
and push the right blue piece by an equal amount. 
Similarly, if the left yellow piece pushes
to the right, then the orange piece
will slide along the blue pieces,
and will push the right yellow piece by an equal amount.
For this to work the two opposite edges of the
orange piece in contact with the blue pieces
must be parallel and similarly 
the two other edges in contact with the yellow pieces 
must be parallel as well.
The conclusion is that for all placements of the orange piece where it has edge-edge contacts to all yellow and blue pieces, the horizontal distance between the two yellow pieces is the same, as is the distance between the two blue pieces.

\begin{figure}
\centering
\includegraphics[page = 2]{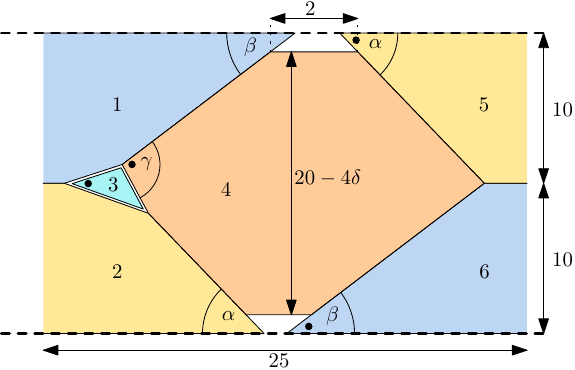}
\caption{The actual swap. Color codes: $1$ blue, $2$ yellow, $3$ turquoise, $4$ orange, $5$ yellow, $6$ blue.}
\label{fig:swap-precise}
\end{figure}

\subsubsection*{The actual swap}
See \Cref{fig:swap-precise} for the 
following description.
The swap consists of six pieces. 
Those are the left and right yellow piece, the left and 
right blue piece, the orange piece and the turquoise
piece.

The left yellow and blue pieces extend outside the gadget to the left, where they have been introduced in other gadgets added earlier to the construction.
Likewise, the right yellow and blue pieces extend outside the gadget to the right, where their ends will be defined in other gadgets added later to the construction.
The gadget is bounded by horizontal segments from above and below, and these are either part of the container boundary or the boundary of other pieces that have been added earlier.
As is seen from the figure, the yellow pieces have corners with the same angle $\alpha$, and the blue pieces have corners with the same angle $\beta$.
The precise value of those angles is chosen freely for the fingerprinting.
Similarly, the orange and turquoise piece have a corner with a flexible angle that can be chosen freely for fingerprinting. 
The orange piece has a horizontal top and bottom edge of length~$2$.

\begin{figure}
\centering
\includegraphics[page=3]{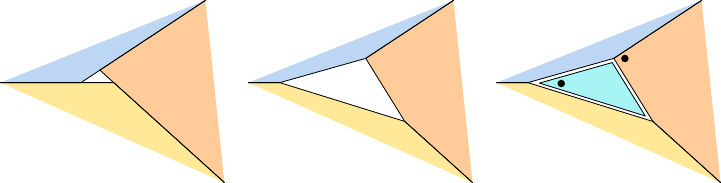}
\caption{Left: the construction without the turquoise piece.
Middle: We remove part of the pieces, leaving a triangular empty space with edges of length more than $1$.
Right: The turquoise piece is designed to fit in the triangle of empty space, but it is surrounded by empty space of thickness $\delta$ in every direction. This leaves enough wiggle room for the 
other pieces to encode all solutions to $\Phi$.}
\label{fig:swap-turquoise}
\end{figure}

For the way we construct the turquoise piece, we refer
the reader to \Cref{fig:swap-turquoise}.
The role of the turquoise piece is solely
to be able to fingerprint the orange piece.
It has no direct use in the functionality of
the swap.
If we avoid using the turquoise piece and instead use pieces as in the simplified \Cref{fig:swap-idea}, the left corner of the orange piece, that we want to fingerprint in the wedge between the left yellow and blue pieces, has angle $\alpha+\beta$, and thus the unique angle property is violated; indeed, the right yellow and blue pieces can cover the wedge equally well.
The left endpoints of the horizontal segments of the orange piece cannot be used for fingerprinting, because the angles are more than $\pi/2$.
It may be tempting to believe that one could avoid the turquoise piece, but we could not find such a way using only convex pieces.
Note that we have only two degrees of freedom without the turquoise  piece, as certain edges must be parallel.
This is not enough to choose three angles freely for fingerprinting.

\subsubsection*{Canonical placements and solution preservation}
Recall that the yellow and blue pieces represent the variables $x$ and $y$, respectively.
In \Cref{fig:swap-precise}, all variable pieces encode the value $0$.
For a placement of the six pieces to be canonical, we require that the yellow and blue pieces have edge-edge contacts to the orange piece as shown and that the turquoise piece is enclosed by the left yellow and blue pieces and the orange piece as shown.

We prove now \Cref{lem:preservation}, about solution preservation, for the swap.
\begin{proof}[Proof of Lemma \ref{lem:preservation} for the swap]
Suppose that for a given solution~$\mathbf x$ to the \etrinv formula $\Phi$, there exists a canonical placement of the previously introduced pieces $\p_{i-1}$ that encodes that solution.
To extend the placement to the pieces $\p_i$, i.e., with the yellow, blue, turquoise and orange pieces of this swap included, we place the yellow and blue pieces so that they represent the value of $x$ and $y$ in $\mathbf x$.
This leaves room for the orange piece to be placed with the required edge-edge contacts to the blue and yellow pieces.
The turquoise piece has enough wiggle room to be placed correctly as well. 
\end{proof}

\subsubsection*{Fingerprinting and almost-canonical placement}

Here we are less detailed in the application of \Cref{lem:multipleFingerprints} (Multiple Fingerprints) than in the section about the anchor in order to avoid unnecessary repetition.

Recall that in \Cref{lem:AlmostCanonicalPlacement}, we assume that the previous pieces $\p_{i-1}$ have an 
aligned $(i-1)\slack$-placement and we want to conclude that the new pieces $\p_i$ have an almost canonical placement.

\begin{proof}[Proof of Lemma \ref{lem:AlmostCanonicalPlacement} for the swap]
Recall that $\p_{i-1}$ consists of all pieces introduced previously, including the left blue and yellow pieces, whereas $\p_i$ additionally includes the orange, turquoise and the right blue and yellow pieces.
We now fingerprint the turquoise, orange, right blue, and right yellow pieces in this order.
We use the intended placements of those pieces shown in \Cref{fig:swap-precise} (corresponding to the case where the variables encode the value $0$) and fingerprint the corners marked with dots.
These intended placements are $O(\delta+(i-1)\slack)$-sound for any aligned $(i-1)\slack$-placement of the pieces $\p_{i-1}$, since the empty space has thickness $O(\delta)$ and the left blue and yellow pieces have displacement of at most $(i-1)\slack$ compared to the shown placement, by the assumption of \Cref{lem:AlmostCanonicalPlacement}.
We conclude using \Cref{lem:multipleFingerprints} (Multiple Fingerprints) that the displacement is $O(n^{-48})$, so the pieces must have an almost-canonical placement.
\end{proof}

\subsubsection*{Aligned placement}
\begin{proof}[Proof of Lemma \ref{lem:alignedPl} for the swap]
Consider a valid placement where the pieces~$\p_{i-1}$ have an aligned $(i-1)\slack$-placement and the pieces of this swap have an almost-canonical placement; see \Cref{fig:swap-aligning} (left).
Similarly as in the proof for the anchor, we consider the alignment segment $\ell_1$ and get that the right yellow and blue edges are in horizontal position.

\begin{figure}
\centering
\includegraphics[page = 4]{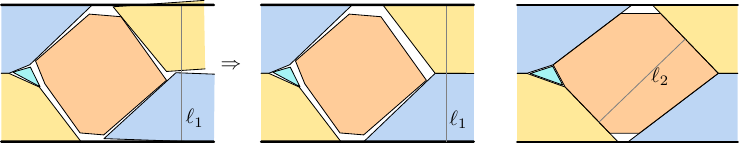}
\caption{Left: Due to the alignment segment $\ell_1$,
the right yellow and blue pieces must be axis-parallel.
Right: Due to the alignment segment $\ell_2$, the orange piece must also be in a canonical placement.}
\label{fig:swap-aligning}
\end{figure}

In order to show that the pieces $\p_i$ have an aligned $i\slack$-placement, it remains to bound the horizontal displacement of the right yellow and blue pieces as compared to the placements encoding the value $0$.
Consider the yellow piece, the argument for the blue is similar.
Let $p_1$ be the left yellow piece and $p_2$ the right one.
By assumption, we have $\enc{p_1}\in [-(i-1)\slack,(i-1)\slack]$, and we need to prove $\enc{p_2}\in[-i\slack,i\slack]$.
Suppose that the yellow pieces are right-oriented; the other case is similar.
From the proof of \Cref{lem:graphIneq} for the swap (given below), we have that $\enc{p_1}\leq\enc{p_2}$, so we just need to show $\enc{p_2}\leq i\slack$.
It is therefore sufficient to show $\enc{p_2}\leq\enc{p_1}+\slack$.
This follows as in the proof for the anchor by considering how much the piece~$p_2$ can be slid to the right before the empty space thus created gets larger than $\slack$.
\end{proof}

\subsubsection*{Edge inequalities}
\begin{proof}[Proof of Lemma \ref{lem:graphIneq} for the swap]
In the swap, the yellow pieces induce an edge in the graph $G_x$ and the blue induce an edge in $G_y$.
In the special case that $x=y$, the blue and yellow pieces have opposite orientations, so according to the rule about when to add edges to the dependency graph $G_x$, there will also be an edge between the left pieces and one between the right pieces.
We make an exception to the rule and do not add these edges to the dependency graph.

We now prove the edge inequality for the edge between the yellow pieces; the argument for the blue pieces is analogous.
Suppose that the yellow pieces are right-oriented and let $p_1$ be the left and $p_2$ be the right, so that they induce the edge $(p_1,p_2)$ of $G_x$.
The argument is analogous if they are left-oriented.
Recall that $\enc{p_1}= \enc{p_2}$ exactly when the pieces have the horizontal distance shown in \Cref{fig:swap-precise}, where the orange piece has edge-edge contacts with both yellow pieces.
Since the pieces have an almost-canonical placement by assumption, the displacement of the orange piece is at most $n^{-1}$.
It therefore follows that the same pair of parallel edges of the orange piece will prevent the yellow pieces from being closer than in the figure, so we have $\enc{p_1}\leq \enc{p_2}$.
\end{proof}

%%%%%%%%%%%%%%%%%%%%%%%%%%%%%%
\subsection{Split}
\label{sec:split}
%%%%%%%%%%%%%%%%%%%%%%%%%%%%%%
The purpose of the split is to make an extra lane representing a variable $x$.
This will be needed in order to lead lanes into the gadgets for the adders and curvers.

\begin{figure}
\centering
\includegraphics[page = 2]{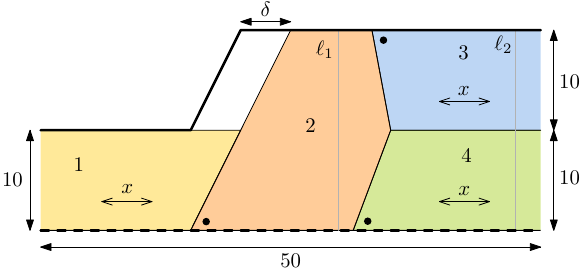}
\caption{The split. 
The yellow, blue, and green pieces represent the variable $x$ and are either all right-oriented or all left-oriented.
Color codes: $1$ yellow, $2$ orange, $3$ blue, $4$ green.}
\label{fig:split-precise}
\end{figure}

\subsubsection*{Description}
See \Cref{fig:split-precise} for an illustration of the split.
We always split the topmost lane in the construction, so that there is room to expand above with one more lane.
Therefore, the split is bounded from above by the boundary of the container and from below by pieces in the second-highest lane, which have been added to the construction earlier.

The yellow piece extends outside the gadget to the left, where it has been introduced to the construction in an earlier step.
The yellow piece is in contact to the right with an orange piece with height $20$, i.e., twice the height of a lane.
The orange piece is in contact to the right with the blue and green pieces, which extend outside the gadget to the right, where they will enter other gadgets defined later in the construction.
Each of the orange, blue, and green piece has a corner with an angle that can be freely chosen for fingerprinting.

\subsubsection*{Canonical placements and solution preservation}
The yellow, blue, and green pieces are all variable pieces encoding a variable $x$.
The position as indicated in \Cref{fig:split-precise} show the situation where they all encode the value $0$.
The placement of all four pieces is canonical if they have the shown edge-edge contacts.
\Cref{lem:preservation} (solution preservation) follows trivially.

\subsubsection*{Fingerprinting and almost-canonical placement}
The proof of \Cref{lem:AlmostCanonicalPlacement} for the split is completely analogous to the one for the swap.

\subsubsection*{Aligned placement}
All pieces must be aligned because of the container boundary and the edge of a piece from $\p_{i-1}$ bounding the gadget from below.
As for the swap, to get \Cref{lem:alignedPl} for the split, we need to verify that the blue and green piece encodes a value that is at most $\slack$ larger than that encoded by the yellow piece.
The proof is analogous as the one for the swap.

\subsubsection*{Edge inequalities}
In the split, we get two edges in the dependency graph.
If the pieces are right-oriented, we have edges from the yellow to green and to the blue pieces.
Otherwise, we have edges from green to yellow and blue to yellow.
That the edges satisfy the edge inequality (\Cref{lem:graphIneq}) follows by construction, since the orange piece restricts how close the yellow piece can get to the green and blue pieces.

\subsubsection*{Left-split}
We explained above how to split an entering variable piece into two exiting pieces.
We will also need the opposite, i.e., splitting an exiting piece into two entering pieces (or, in other words, merging two entering pieces into one exiting piece).
It is straightforward to construct such a gadget by similar principles, as shown in~\Cref{fig:split-left}.
In order to avoid disambiguation, we will occasionally denote these splits as \emph{right-splits} and \emph{left-splits}, respectively.

\begin{figure}
\centering
\includegraphics[page = 1]{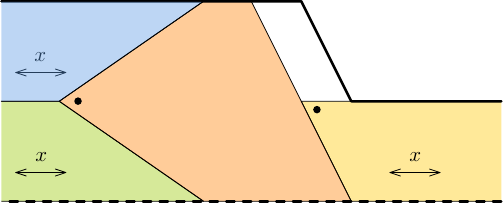}
\caption{The left-split.
}
\label{fig:split-left}
\end{figure}

%%%%%%%%%%%%%%%%%%%%%%%%%%%%%%
\subsection{Adder}
\label{sec:addition}
%%%%%%%%%%%%%%%%%%%%%%%%%%%%%%

\begin{figure}
\centering
\includegraphics[page = 1]{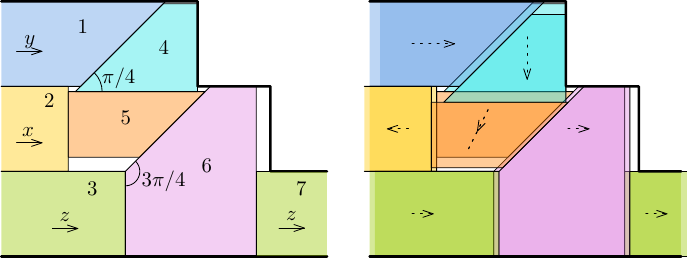}
\caption{Left: The simplified adder. 
Right: The variable $x$ is increased by two units. The variable $y$ is decreased by one unit. 
Thus the variable $z$ is increased by one unit.
Color codes: $1$ blue, $2$ yellow, $3$ green, $4$ turquoise, $5$ orange, $6$ pink, $7$ green.}
\label{fig:addition-idea}
\end{figure}

\subsubsection*{Idea}
For the following description, see~\Cref{fig:addition-idea}.
Here we explain the principle behind the adder for $x+y\leq z$.
The adder for $x+y\geq z$ is identical, but has the entering variable pieces for $x,y,z$ oriented to the left instead of to the right.
The adder has three entering variable pieces (yellow, blue, left green), representing three variable ($x,y,z$).
There is also one exiting green variable piece representing $z$.
In addition to this, there are three pieces that are not variable pieces (turquoise, orange, pink).
The role of the turquoise piece is to transform horizontal motion to the right of the blue piece to downwards vertical motion.
Motion of the orange piece downwards or to the right both make the pink piece move to the right by an equal amount.
Therefore, when the blue and yellow pieces push to the right, the pink piece will be pushed to the right by the sum of the two motions.

\begin{figure}
\centering
\includegraphics[page = 2]{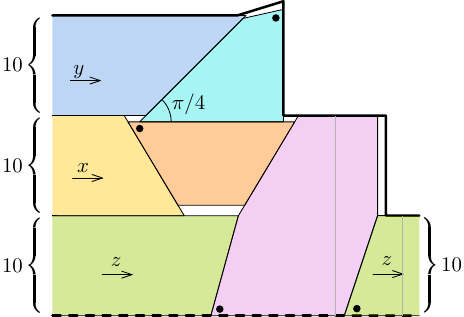}
\caption{The actual adder.}
\label{fig:addition-precise}
\end{figure}

\subsubsection*{Actual description}
The actual adder is shown in \Cref{fig:addition-precise}.
The figure shows the situation where all the variable pieces encode the value $0$.
The actual adder varies in several points from the simplified version in \Cref{fig:addition-idea}, which is needed in order to use fingerprinting.
The orange, turquoise, pink, and right green pieces must be fingerprinted, so they cannot have only nice angles as in the simplified gadget.
The turquoise piece can easily be fingerprinted using the top right corner.
The orange piece is fingerprinted in the upper left corner.
The pink and right green pieces are fingerprinted at their lower left corners.
In each case, the angle can be chosen freely by changing the slope of the edge-edge contact with the piece to the left.

\begin{figure}
\centering
\includegraphics[page = 7]{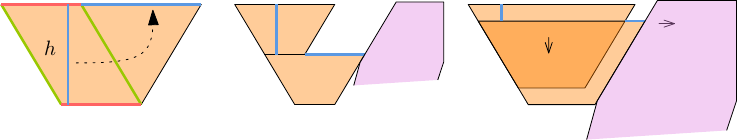}
\caption{In each of these three pictures, thick segments drawn with the same color are equally long.
Left: The construction of the orange piece.
Middle: Pushing the orange piece down by its height makes it push the pink piece by the same amount.
Right: This also holds when pushing less.}
\label{fig:addition-orange}
\end{figure}

As a consequence of changing the angle  on the top left of the orange piece to something else than $\pi/4$, we also have to change the angle of the top right corner.
See \Cref{fig:addition-orange} for an illustration of the  following.
First we describe how to construct the orange piece and then we explain why it actually works.
The orange piece is a trapezoid with horizontal bottom and top edges and height $10 - 10\delta$.
The left edge of the orange piece is parallel to the right edge of the yellow piece.
The length of the top edge should be at least as long as the bottom edge of the turquoise piece.
The length of the bottom piece is the length of the top edge minus the height $10 - 10\delta$.
The right edge is determined by the description of all the other edges.

We need to explain why pushing the blue piece to the right by some amount $t>0$ will push the pink piece to the right by $t$ as well; see \Cref{fig:addition-orange}.
It is helpful to consider the case where $t$ equals the height $h$ of the orange piece (even though there is not room for pushing the pieces so much).
This push of the blue piece will push the orange piece down by $h$.
Since the length of the top edge of the orange piece equals the length of the bottom edge plus $h$, the pink piece will be pushed to the right by $h$ as well.
All the pieces move linearly, so it will also be the case for smaller values of $t$.

\begin{figure}
\centering
\includegraphics[page = 5]{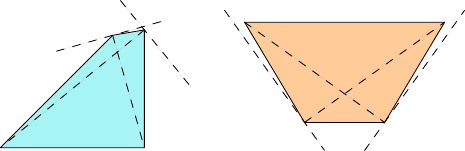}
\caption{The lines through corners and perpendicular to diagonals are disjoint from the interiors.}
\label{fig:diagonal-prop}
\end{figure}

We furthermore want the property that for each of the top corners of the turquoise piece, the line through the corner and perpendicular to the diagonal of the corner is a tangent to the piece (i.e., the line intersects the piece only at the corner).
The same must hold for the bottom corners of the orange piece; see \Cref{fig:diagonal-prop}.
It is easy to choose the fingerprinted angles so that the pieces have this property.
We will use the property in the proof of \Cref{lem:addition} given later to conclude that if the turquoise or the orange piece is not aligned as in the canonical placements, they will take up too much space.

\subsubsection*{Canonical placements and solution preservation}
The canonical placements are defined as the placements with the edge-edge contacts as shown in \Cref{fig:addition-precise}.
That there is a canonical placement encoding any given solution to $\Phi$ (\Cref{lem:preservation}) follows by construction.

\subsubsection*{Incorporating the gadget}
In order to incorporate the adder into the construction, we use two splits and eight swaps in order to organize the in-going and out-going lanes to the gadget; see \Cref{fig:addition-stitching} for a schematic illustration.
\begin{figure}
\centering
\includegraphics[page = 4]{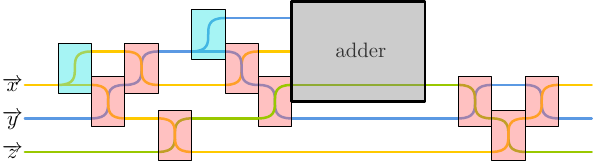}
\caption{Incorporation of the adder. Splits are marked in turquoise, swaps are marked red.}
\label{fig:addition-stitching}
\end{figure}

\subsubsection*{Fingerprinting and almost-canonical placement}
The proof of \Cref{lem:AlmostCanonicalPlacement} for the adder is completely analogous to the one for the swap.

\subsubsection*{Aligned placement}
The right green piece must be aligned because of the container boundary and the edge of a piece from $\p_{i-1}$ bounding the gadget from below.
As for the swap, to get \Cref{lem:alignedPl} for the adder, we need to verify that the right green piece encodes a value that is at most $\slack$ larger than that encoded by the left green piece.
The proof is analogous as the one for the swap.

\subsubsection*{Edge inequalities}
In the adder, the left and right green pieces induce an edge in the dependency graph.
That the edge satisfies the edge inequality (\Cref{lem:graphIneq}) is proven as for the swap.

\subsubsection*{The adder works}
In this paragraph we prove \Cref{lem:addition}.
Here we are considering an aligned $\gadgets\slack$-placement, and we need to prove that for every addition constraint $x+y= z$ of $\Phi$, we have $\enc{K_x}+\enc{K_y}=\enc{K_z}$.

\begin{proof}[Proof of Lemma \ref{lem:addition}]
We prove the inequality $\enc{K_x}+\enc{K_y}\leq \enc{K_z}$.
The other inequality follows from analogous arguments about the gadget for $x+y\geq z$.
Let $p_x,p_y,p_{z1}$ be the yellow, blue, and left green pieces, and $p_{z2}$ be the right green piece.

Note first that due to the way we incorporate the adder, there are paths $P_x$ and $P_y$ in $G_x$ and $G_y$ attached to and directed away from the cycles $K_x$ and $K_y$, and the vertices farthest away from the cycles are the pieces $p_x$ and $p_y$.
Furthermore, the pieces $p_{z1}$ and $p_{z2}$ are two consecutive vertices on the cycle $K_z$, so by \Cref{lem:consistentCycle}, we have $\enc{p_{z1}}=\enc{p_{z2}}=\enc{K_z}$.

\begin{figure}
\centering
\includegraphics[page = 6]{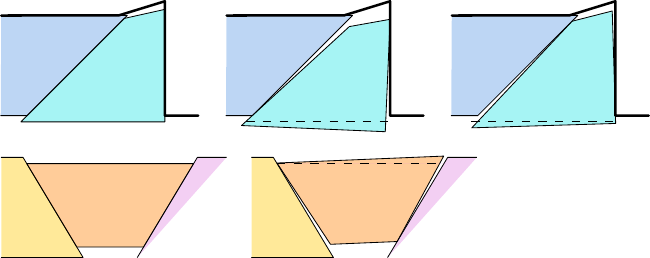}
\caption{Top: If the turquoise piece does not have edge-edge contacts to the blue piece and the boundary, then all points of the lower edge have strictly smaller $y$-coordinates than otherwise.
Bottom: If the orange piece does not have edge-edge contacts to the yellow and pink pieces, then the top edge is strictly higher than otherwise.
The dashed segments are the edges in the canonical situations shown to the left.}
\label{fig:rotate-prop}
\end{figure}

We argue that if $\enc{p_x}+\enc{p_y}= \enc{K_z}$, the only way to place the pieces is the canonical way.
This excludes the situation $\enc{p_x}+\enc{p_y}> \enc{K_z}$, since there the turquoise, orange, and pink pieces have strictly less space.
Consider first the turquoise piece.
It is straightforward to check that if it does not have edge-edge contacts to the blue piece and the container boundary, then the lower edge will be strictly below the placement of the edge where these contacts were present, in the sense that every point on the segment will have a placement with a smaller $y$-coordinate.
This can be seen in \Cref{fig:rotate-prop} and is due to the property that the lines through corners perpendicular to diagonals are tangents, as shown in \Cref{fig:diagonal-prop}.
Similarly for the orange piece, if it does not have edge-edge contacts with the yellow and pink pieces, the top edge is strictly above the edge in the placement where it does have these contacts.
In the canonical placements, the bottom edge of the turquoise piece is contained in the top edge of the orange piece.
Therefore, in any other placement of the turquoise and orange pieces, there will be points on the turquoise edge below the orange edge, which makes the placement invalid.
We can now conclude that in general, we must have $\enc{p_x}+\enc{p_y}\leq \enc{K_z}$.

To finish the proof, note that since the paths $P_x$ and $P_y$ are directed away from the cycles $K_x$ and $K_y$, we get from the edge inequalities (\Cref{lem:graphIneq}) that
\[
\enc{K_x}+\enc{K_y}\leq \enc{p_x}+\enc{p_y}\leq\enc{K_z}.
\qedhere
\]
\end{proof}

%%%%%%%%%%%%%%%%%%%%%%%%%%%%%%
\section{Curvers}
\label{sec:Curved}
%%%%%%%%%%%%%%%%%%%%%%%%%%%%%%

We are going to construct two different types of curvers.
Both curvers will be described in a convex and a concave version.
To prove $\ER$-hardness of the problem \pack\convexpolygon\polygon\rotation, we will use the convex and the concave \emph{swing}.
For the problems \pack\convexpolygon\curved\translation\ and \pack\curved\polygon\translation, we will instead use variants of the convex and the concave \emph{gramophone}.

%%%%%%%%%%%%%%%%%%%%%%%%%%%%%%
\subsection{Swing}
\label{sec:swing}
%%%%%%%%%%%%%%%%%%%%%%%%%%%%%%

The swing is reminiscent of the swings found on playgrounds, especially the twin version where two kids face each other and must collaborate to get the swing going; see \Cref{fig:swing-principle}.
In the preliminary version of this paper~\cite{abrahamsen2020framework}, we also used other playground inspired gadgets, namely the \emph{teeter-totter} and the \emph{seesaw}, which could encode a convex and a concave constraint, respectively.
However, the seesaw needed non-convex pieces, and since our new swing can also encode a convex constraint, we need neither of them anymore.
The orange piece in the swing plays the role of the swing itself and can rotate around the top corner $o$, while the blue and yellow pieces correspond to the children, pushing the swing from each side.
Surprisingly, by adjusting the height of the contact points $c_x$ and $c_y$ of the orange piece with the other pieces, we can make versions of the swing that enforce convex and concave constraints, as we like.

\begin{figure}[b]
\centering
\includegraphics[page = 1]{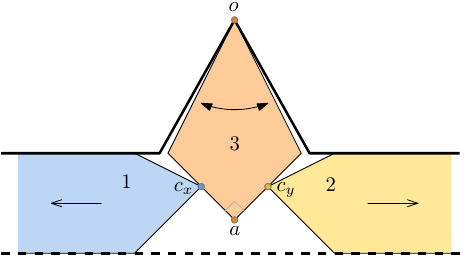}
\caption{A simplified drawing of the swing.
The corner $a$ is a right angle.
Color codes: $1$ blue, $2$ yellow, $3$ orange.
}
\label{fig:swing-principle}
\end{figure}

\subsubsection*{Principle}

For the following argument, see \Cref{fig:swing-geometry}.
We assume for simplicity that $\vert oa\vert =1$ and consider the situation where the orange piece is rotated so that the segment $oa$ makes an angle of $\varphi\in(-\pi/4,\pi/4)$ with the vertical axis.
The corners of the blue and yellow pieces that are in contact with the orange piece lie on the horizontal line $\ell_z$ which is $z>0$ below the horizontal line $\ell_o$ through $o$.
We consider the extensions of the edges of the orange piece incident at $a$ and construct the intersections $c_x$ and $c_y$ with $\ell_z$.
We want to find formulas for horizontal distances $x_1$ and $y_1$ from $o$ to $c_x$ and $c_y$, respectively, since the rightmost corner of the blue piece will be at $c_x$ or to the left, and the leftmost corner of the yellow piece will be at $c_y$ or to the right, so $x_1$ and $y_1$ will lead to lower bounds for the values encoded by the pieces, as the pieces are oriented away from $c_x$ and $c_y$, respectively.
(In the figure, $c_x$ is in the exterior of the orange piece and would thus not be a bound on the placement of the blue piece, but when using the gadget, the point will always be on the boundary of the orange piece, since it will only be slightly rotated.)
The distances $x_1$ and $y_1$ are signed distances, so that $x_1$ is positive if and only if $c_x$ is to the left of $o_z$ and $y_1$ is positive if and only if $c_y$ is to the right of $o_z$.

\begin{figure}
\centering
\includegraphics[page = 2]{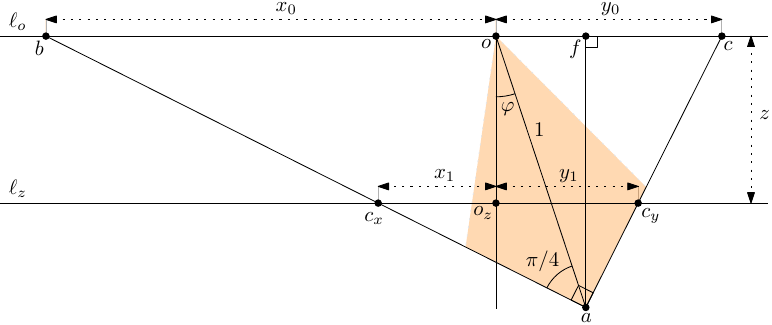}
\caption{The geometry of the swing.
}
\label{fig:swing-geometry}
\end{figure}

Note that $\vert af\vert = \cos \varphi$.
Considering the right triangles $afb$ and $afc$, we then obtain
\begin{align*}
\vert ab\vert & = \frac{\cos\varphi}{\cos(\pi/4+\varphi)}; \\
\vert ac\vert & = \frac{\cos\varphi}{\cos(\pi/4-\varphi)}.
\end{align*}

We then get
\begin{align*}
x_0 & = \vert ob \vert = \vert fb\vert -\vert of\vert = \sin(\pi/4+\varphi) \vert ab\vert -\vert of\vert = \tan(\pi/4+\varphi)\cos \varphi-\sin \varphi; \\
y_0 & = \vert oc \vert = \vert fc\vert +\vert of\vert = \sin(\pi/4-\varphi) \vert ac\vert +\vert of\vert = \tan(\pi/4-\varphi)\cos \varphi+\sin \varphi.
\end{align*}

We observe that if $\ell_z$ coincides with $\ell_o$, then $x_1=x_0$, while if $\ell_z$ passes through $a$, then $x_1=-\sin\varphi$.
Since the vertical distance from $\ell_o$ down to $a$ is $\cos \varphi$, we define $v\mydef z/\cos \varphi$ and get $x_1$ by interpolating between $x_0$ and $-\sin\varphi$ using the factor $v$, and we can express $y_1$ in an analogous way.
By standard trigonometric identities, we obtain the expressions
\begin{align*}
x_1 & = x_0 \cdot (1-v)-\sin\varphi \cdot v=\frac{1 + z(-\cos\varphi - \sin\varphi)}{\cos\varphi - \sin\varphi}; \\
y_1 & = y_0 \cdot (1-v)+\sin\varphi \cdot v=\frac{1 + z(-\cos\varphi + \sin\varphi)}{\cos\varphi + \sin\varphi}.
\end{align*}

\begin{figure}
\centering
\includegraphics[page = 1, width=0.7\textwidth]{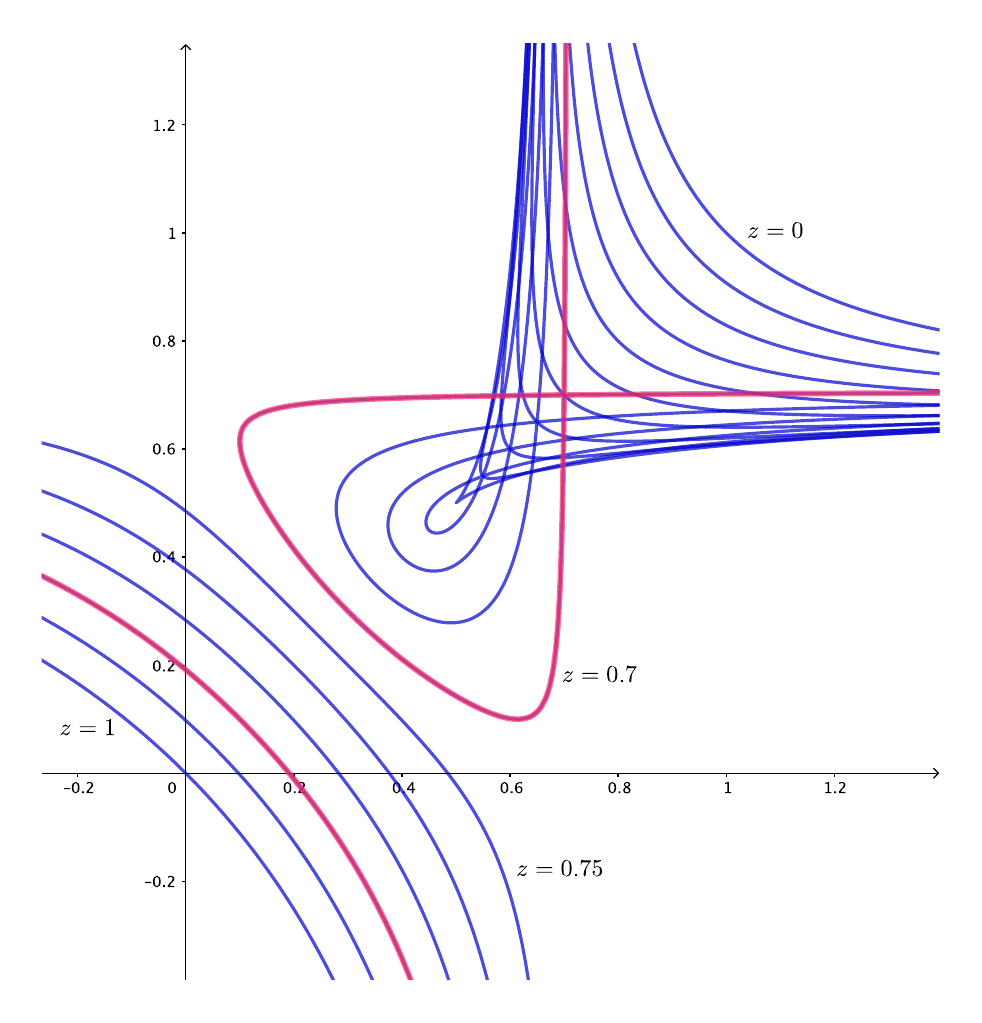}
\caption{The curve $\gamma_z$ for each of the values $z\in\{i/20\mid i=0,\ldots,20\}$.
The curves for $z=0.7$ and $z=0.9$, that we will use in our construction, have been highlighted.
}
\label{fig:swing-curves}
\end{figure}

We consider $x_1$ and $y_1$ as functions of $\varphi$.
Note that $x_1(0)=y_1(0)=1-z$.
The curve $\gamma_z: \varphi\mapsto (x_1(\varphi), y_1(\varphi))$ undergoes a remarkable metamorphosis as we increase $z$ from $0$ to $1$; see \Cref{fig:swing-curves}.
Note the sudden change from the curve being convex to concave at the diagonal point $(1-z,1-z)$ as $z$ changes from $z=0.7$ to $z=0.8$.
This is what we exploit to make convex and concave variants of the swing.

To obtain well-behaved functions $f$ and $g$, as required in \Cref{thm:main}, we will use an implicit function $F_z: U^2\longrightarrow \R$, for an interval $U\subset\mathbb R$ around $0$, such that $F_z(x,y)=0$ if and only if $(x,y)=\gamma_z(\varphi)+(z-1,z-1)$ for some $\varphi\in(-\pi/4,\pi/4)$.
We add $(z-1,z-1)$ to obtain a curve that passes through $(0,0)$ for $\varphi=0$.
In order to show that there exists such a function which is well-behaved, we use the following lemma  from~\cite[Sec.~1.5]{miltzow2021classifying}, which has been slightly rephrased.

\begin{lemma}
[Miltzow and Schmiermann~\cite{miltzow2021classifying}]
Let $\gamma = (\gamma_x, \gamma_y): (-\pi/4,\pi/4) \longrightarrow\mathbb R^2$ be a $C^3$-function such that $\gamma (t) = (0, 0)$ if and only if $t = 0$.
Assume that the derivatives $\gamma_x',\gamma_x'',\gamma_y',\gamma_y''$ are all rational in $0$ and that $\gamma_x'(t)\neq 0$ for all $t\in (-\pi/4,\pi/4)$.
Then there exists an interval around the origin $J\subset \mathbb R$ and a well-behaved function $F$ such that
\[
\{\gamma(t) \in \mathbb R^2\mid t\in (-\pi/4,\pi/4)\} =
\{(x, y) \in J \times \mathbb R \mid F(x, y) = 0\}.
\]
\end{lemma}

In our case, it follows trivially from the definitions that the first and second derivatives of $x_1$ and $y_1$ in $\varphi=0$ are rational and that $x_1'(0)\neq 0$, so the lemma ensures the existence of a function $F_z$ as described.

We will use $f\mydef F_{7/10}$ for our convex constraint and $g\mydef F_{9/10}$ for our concave constraint.
It is then straightforward to check that $f$ is convex and $g$ concave at $(0,0)$ by evaluating the sign of the curvature of the curves $\gamma_{7/10}$ and $\gamma_{9/10}$ at $\varphi = 0$.

\subsubsection*{Actual swing}

\Cref{fig:swing-actual} shows a more accurate illustration of the convex and the concave swing.
The blue and yellow variable pieces represent variables $x$ and $y$, respectively.
The blue piece enters the gadget and is left-oriented, while the yellow piece exits the gadget and is right-oriented.
The top corner $o$ of the orange piece is fingerprinted at a wedge formed by the boundary of the container $\cont$.
We need to ensure that the angle at the corner $o$ can be changed continuously.
By moving the orange piece up or down, we can also change the lengths of the edges meeting at the corner $a$, and then the angle at $o$ also changes.
We need to move the contact corners $c_x$ and $c_y$ of the blue and yellow pieces correspondingly, so that the $z$-value (described above) remains invariant.
The yellow piece is fingerprinted at a wedge formed by the blue piece and the lower boundary of the gadget.

The figure shows the case where the blue and the yellow piece encode the value $0$ of both $x$ and $y$.
In the convex swing, the vertical distance from $o$ to the contact corners $c_x,c_y$ is $7\vert oa\vert/10$.
In the concave swing, the distance is $9\vert oa\vert/10$.

\begin{figure}
\centering
\includegraphics[page = 3]{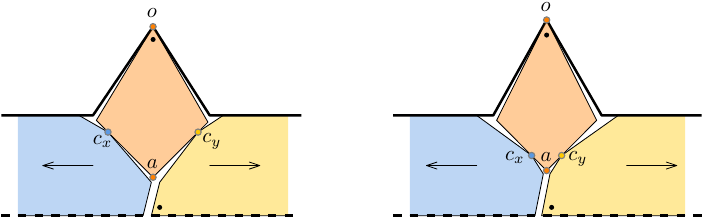}
\caption{The convex swing (left) and the concave swing (right).
}
\label{fig:swing-actual}
\end{figure}

\subsubsection*{Incorporation of the swing}

\Cref{fig:swing-incorportation} shows a schematic illustration of how the swing is incorportated in the construction.
We use a right-split and a swap before the swing itself, and a left-split and a swap after.

\begin{figure}
\centering
\includegraphics[page = 4]{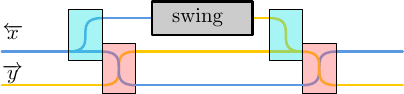}
\caption{Incorporation of the swing.
}
\label{fig:swing-incorportation}
\end{figure}

\subsubsection*{Canonical placements and solution preservation}
A placement of the three pieces of the gadget is canonical if the orange piece separates the yellow and blue pieces in the sense that the segment $c_xc_y$ (connecting the contact corners of the blue and yellow pieces) intersects both edges of the orange piece incident at the corner $a$.

\begin{proof}[Proof of Lemma \ref{lem:preservation} for the swing]
Suppose that for a given solution to the \fgetr formula $\Phi$, there exists a canonical placement of the previously introduced pieces~$\p_{i-1}$ that encodes that solution.
We place the yellow piece so that it encodes the value of $y$ as specified by $\Phi$.
We need to verify that there is enough empty room such that the orange piece can be placed in the gadget.
Consider without loss of generality the case that we introduced the gadget because of a convex constraint of the form $f(x,y)\geq 0$ in $\Phi$.
If the blue and yellow pieces touch the orange piece, it corresponds to the case $f(x,y)= 0$.
Since we consider a solution to $\Phi$, we have $f(x,y)\geq 0$, so there is enough room for the orange piece.
\end{proof}

\subsubsection*{Fingerprinting and almost-canonical placement}
The proof of \Cref{lem:AlmostCanonicalPlacement} for the swing follows exactly as for the swap.
We first fingerprint the orange and then the yellow piece.

\subsubsection*{Aligned placement}
\begin{proof}[Proof of Lemma \ref{lem:alignedPl} for the swing]
Consider a valid placement where the pieces $\p_{i-1}$ have an aligned $(i-1)\slack$-placement and the pieces of this swing have an almost-canonical placement.
We need to verify that the exiting yellow piece has an aligned $i\slack$-placement.
It is clear that the yellow piece has an aligned placement because it is bounded from above and below by horizontal edges of the container $\cont$ and a piece below.
It is however not immediately clear that it has an $i\slack$-placement, since we have no entering piece representing $y$, so we cannot use the assumption that the pieces $\p_{i-1}$ have an aligned $(i-1)\slack$-placement directly.
Instead of considering the swing as an isolated gadget, we therefore take the two extra pieces of the left-split, which is to the right of the swing, into account; see \Cref{fig:swing-incorportation}.
We then get that if these pieces do not have an aligned $i\slack$-placement, then the slack would exceed $\slack$, as we argued for the swap.
\end{proof}

\subsubsection*{Edge inequalities}
The pieces of the swing induce no edges in the dependency graphs, so there are no edge inequalities to verify.

\subsubsection*{The swing works}
In this paragraph we prove the part of \Cref{lem:inversion} that the swing is responsible for, namely that in a given aligned $\gadgets\slack$-placement, the convex swing used for the variables $x$ and $y$ implies that $f(\enc{K_x}, \enc{K_y})\geq 0$ and the concave swing implies $g(\enc{K_x}, \enc{K_y})\geq 0$.

\begin{proof}[Proof of Lemma \ref{lem:inversion} for the swing]
We prove that $f(\enc{K_x}, \enc{K_y})\geq 0$ when we use a convex swing; the statement for the concave swing follows in an analogous way.
Let $p_x$ and $p_y$ be the blue and yellow pieces, respectively.
We note that there are paths $P_x$ and $P_y$ in $G_x$ and $G_y$ attached to and directed towards the cycles $K_x$ and $K_y$, and the vertices farthest away from the cycles are the yellow and blue pieces $p_x$ and $p_y$, so that we have $\enc{K_x}\geq \enc{p_x}$ and $\enc{K_y}\geq\enc{p_y}$.
Since we consider an aligned $\gadgets\slack$-placement, the orange piece separates the yellow and blue pieces, as defined for a canonical placement.
Without loss of generality, we can consider a situation where the corners $c_x$ and $c_y$ of the blue and yellow pieces are in contact with the orange piece (otherwise, we can slide them towards the orange piece until they make such a contact).
We then get from the principle behind the gadget (described in the beginning of this section) that $f(\enc{p_x}, \enc{p_y})= 0$.
It is straightforward to verify that we have derivatives $f_x(0,0)=1$ and $f_y(0,0)=1$.
Since $f_x$ and $f_y$ are also continuous, we have $f_x>0$ and $f_y>0$ in a neighbourhood of constant size around the origin $(0,0)$.
Hence, we have $f(\enc{K_x}, \enc{K_y})\geq f(\enc{p_x}, \enc{p_y})=0$.
\end{proof}

%%%%%%%%%%%%%%%%%%%%%%%%%%%%%%
\subsection{Gramophone}
\label{sec:gramophone}
%%%%%%%%%%%%%%%%%%%%%%%%%%%%%%

The gramophone works when only translations of the pieces are allowed.
It works in fact also when rotations are allowed, but we will use the swing in that case since the gramophone requires non-polygonal features of the container or the pieces.
The idea behind the gramophone is to have a special pink piece which is translated horizontally by motions of one lane and translated vertically by motions of another lane; see \Cref{fig:gramophone-idea} for an illustration.
Therefore, the placement of the piece in a sense encodes two variables at once.
The pink piece has a corner~$c$ which is bounded from above by a curve, and that induces a dependency between the two lanes, translating to an inequality of the variables represented by the lanes.
The name is chosen since motion of one lane makes the pink piece trace and ``read'' the curve and thus act as the stylus of a gramophone.

The inequality can be changed by choosing another curve bounding the corner $c$.
We can therefore make gramophones for inequalities $(x+1)\cdot (y+1)-1\geq 0$ and $-(x+1)\cdot (y+1) + 1\geq 0$, which correspond to the regions on one side of a hyperbola.
We could also have chosen a circular arc or another well-behaved curve.
The curved arc is part of the boundary of the container $\cont$, so the reduction leads to an instance of the problem \pack\convexpolygon\curved\translation, but as we will see, we can also obtain an instance of the problem \pack\curved\polygon\translation\ by introducing a curved piece instead of using a curved container boundary.

In the following, we will first explain the principle in more detail.
Then we describe how to turn the principle into a gadget that can be installed in our framework.

\begin{figure}
\centering
\includegraphics[page=1]{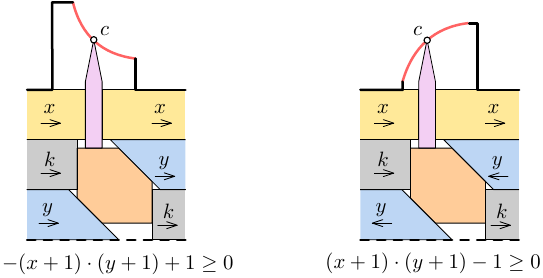}
\caption{The concave and convex gramophone.
Here, the gray pieces marked with $k$ cannot move, so these can be thought of as representing a constant $k$.
The purpose of these is to fix the horizontal placement of the orange piece such that horizontal motion of the blue pieces translates directly to vertical motion of the pink piece.
In all cases the gramophone relies on translations only.
For the resulting dependency between $x$ and $y$ to be non-linear, the curve restricting the corner $c$ must be non-linear, requiring the boundary of the container to have curved parts.
For color codes, see \Cref{fig:gramophone-precise}.}
\label{fig:gramophone-idea}
\end{figure}

\begin{figure}[b]
\centering
\includegraphics[page=2]{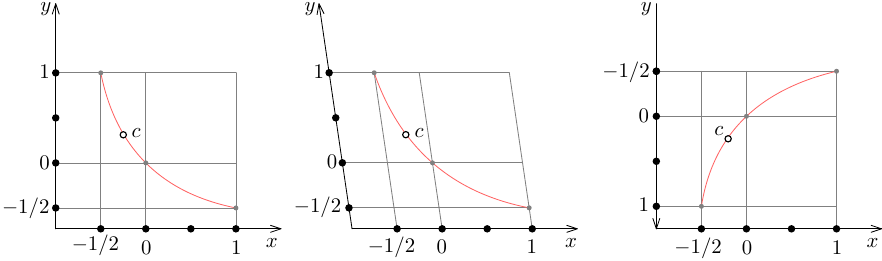}
\caption{All three graphs display the curve for the equation $(x+1)\cdot (y+1)-1 =0$.
The drawings differ because of different $y$-axes.
The left and middle diagrams correspond to the concave gramophone, while the right corresponds to the convex gramophone.
}
\label{fig:gramophone-principle}
\end{figure}

\subsubsection*{Principle}
To understand the principle behind the gramophone, it helps to draw a coordinate system showing the correspondence between the position of the corner $c$ and the variables $x$ and $y$ represented by the two lanes; see \Cref{fig:gramophone-principle}.
The curve restricting the corner $c$ from above is the curve described by the equation $(x+1)\cdot (y+1)-1 = 0$ in the respective coordinate system.
Note that the region of possible positions of the corner $c$ is below the curve.
That region is concave in the left and middle diagram but convex in the right diagram.
In order to get from concave to convex, we have flipped the $y$-axis, but that simply corresponds to changing the orientation of the blue pieces; see \Cref{fig:gramophone-idea}.
It then follows that gramophones can be made to encode both the concave constraint $-(x+1)\cdot (y+1)+1 \geq 0$ and the convex constraint $(x+1)\cdot (y+1)-1 \geq 0$.

\begin{figure}
\centering
\includegraphics[page=3]{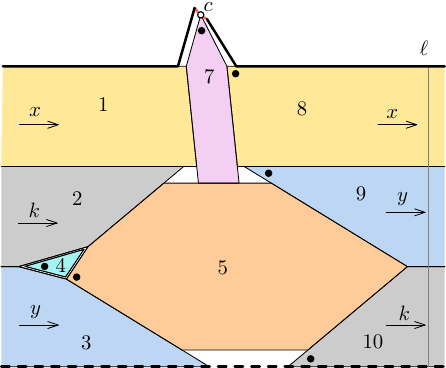}
\hspace{0.2cm}
\includegraphics[page=4]{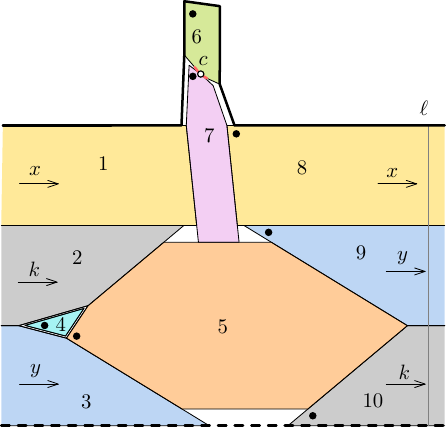}
\caption{The actual concave gramophone (left) and a variant with polygonal container (right).
The solid boundary bounding the gadgets from above are part of the boundary of the container, whereas the dashed segment bounding the gadgets from below is on the boundary of a piece that has been added to the construction earlier.
The gray pieces marked with $k$ cannot move, so these can be thought of as representing a constant $k$.
For a detailed picture of the pink and green pieces, see \Cref{fig:gramophone-pink}.
Color codes: 1 yellow, 2 gray, 3 blue, 4 turquoise, 5 orange, 6 green, 7 pink, 8 yellow, 9 blue, 10 gray.}
\label{fig:gramophone-precise}
\end{figure}

\subsubsection*{Actual gadget}
The actual design of the gramophone can be seen in \Cref{fig:gramophone-precise} (left) and many details differ from the simplified drawing in \Cref{fig:gramophone-idea}.
We only go through the  concave gramophones.
The convex gramophone differs only in the curved part and the orientation of the blue pieces.
In order to incorporate the gadget into the complete construction, we need to add some swaps and a constant lane, as we will explain below.

Note first that since $x$ and $y$ are restricted to tiny ranges $I(x),I(y)\subseteq [-\delta,\delta]$, the curved part needed is likewise of length $O(\delta)$, i.e., extremely short.
The gray, orange, blue and turquoise pieces form a swap.
We refer to the details and correctness of the swap to \Cref{sec:swap}.
The gray pieces will have a fixed placement and can therefore be considered as pieces encoding a constant value.
A lane for the constant will be started to the left of the gramophone and terminated to the right, which will be explained later.

Above the swap is the yellow lane and a pink piece.
The pink piece has edge-edge contacts with the yellow and orange pieces.
The top vertex of the pink piece is our special vertex $c$ that is bounded from above by a curve.
In the simplified description, the pink piece had a pair of parallel segments that were fixed at a vertical orientation by the yellow pieces.
However, in order to fingerprint the right yellow piece, we rotate the pink piece a bit clockwise (and adjust the edges of the yellow pieces accordingly) so that the angle of the corner where the right yellow piece is fingerprinted can be chosen freely.
The gadget is bounded from above by the container boundary and this part contains a tiny curved part, which corresponds to the curve $(x+1)\cdot (y+1)-1 = 0$, marked red. 
However, because of the slanted orientation of the pink piece, the $y$-axis of the corner $c$ is now likewise slanted as in \Cref{fig:gramophone-principle} (middle).

In the simplified construction, the angles were chosen so that horizontal motion of the blue pieces was translated to vertical motion of the orange piece with no scaling involved.
In the actual gadget, since we want to fingerprint the right blue piece, we need to chose the slope of the edge-edge contact between the blue and orange pieces freely, and then a horizontal motion of the blue pieces is scaled when transformed to a vertical motion of the orange piece.

These two differences (slanted and scaled $y$-axis) result in a linear deformation of the red curve, as compared to the simplified situation.
An example of such a deformation can be seen in \Cref{fig:gramophone-principle} (middle).
The curve will however still be contained in a hyperbola.
Moving the red curve up and down (where up and down is defined by the transformed $y$-axis) makes it possible to freely choose the fingerprinted angle of the pink piece.

\begin{figure}
\centering
\includegraphics[page=5]{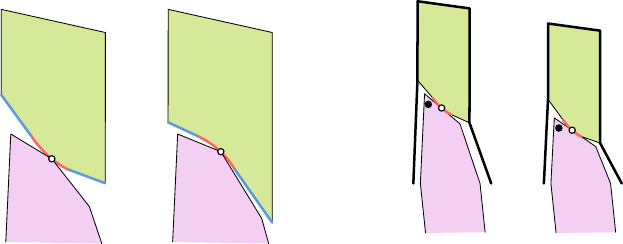}
\caption{Left: Detail of the pink and green pieces and their point of contact in the convex and the concave version of the polygonal gramophone.
Right: Fingerprinting the pink piece. By moving the green piece up and down,
the angle of the fingerprinted corner of the pink piece can be continuously altered.}
\label{fig:gramophone-pink}
\end{figure}

\subsubsection*{Polygonal gramophone}
We introduce special \emph{polygonal} gramophones that work with a polygonal container $\cont$, but has curved pieces in order to show hardness of the problem \pack\curved\polygon\translation.
The polygonal concave gramophone can be seen in \Cref{fig:gramophone-precise} (right).
In this version, we have an additional green piece which is the only piece that is curved.
If we encode the concave constraint $-(x+1) \cdot (y+1) +1 \geq 0$, then the green piece is convex. 
If we encode the constraint $(x+1) \cdot (y+1)-1 \geq 0$, the green piece will not be convex.
Recall from \Cref{sec:curvedPolygons} that for pieces with curved segments, it is needed that the curvature is bounded and that the curved segments have a straight line segment of length $1$ as prefix and suffix. 
In \Cref{fig:gramophone-pink} (left), the red curve and the blue segments together form a curved segment that satisfies the requirements. 
The pink piece is changed accordingly:
In this version of the gramophone, the pink piece has one corner which can touch the curved part of the green piece and an extra corner on each side to fill out the space below the green piece.
The left corner is fingerprinted.
It is shown in \Cref{fig:gramophone-pink} (right) how the angle of the fingerprinted corner can be changed freely by moving up and down the green piece.

\subsubsection*{Installation of the gramophone}
\Cref{fig:gramophone-entire} (left) shows how to install the gramophone in the complete construction.
We make a new lane representing a constant which is started just to the left of the gramophone and terminated just to the right; see \Cref{fig:gramophone-entire} (right) on how a constant lane can start and end.
In addition to this, three swaps are needed to organize the lanes in the right way.

\begin{figure}
\centering
\includegraphics[page=6]{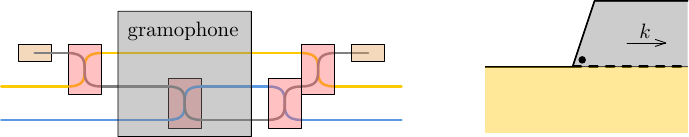}
\caption{Left: Installation manual for the gramophone.
The gray lane consists of pieces with a fixed placement and can be interpreted as encoding a constant.
This lane start and ends just outside the gramophone.
Right: The start of a constant lane consists of a single gray piece, which is fingerprinted at the left corner.
The top of the gadget is the container boundary and we assume that the gadget is bounded from below by a yellow piece representing the variable $x$, which has been introduced to the construction earlier.
The yellow piece is already known to be horizontally aligned, and it then follows that the gray is as well.
The lane of the gray pieces is terminated in an analogous way.}
\label{fig:gramophone-entire}
\end{figure}

\subsubsection*{Canonical placements and solution preservation}
The variable pieces are the yellow and blue pieces.
We define the canonical placements to be placements where the pieces have the edge-edge contacts as in \Cref{fig:gramophone-precise} and the turquoise piece is enclosed by the left gray and blue pieces and the orange piece.
If the green piece is present, it should have edge-edge contacts with the container boundary as shown.

\begin{proof}[Proof of Lemma \ref{lem:preservation} for the gramophone]
Suppose that for a given solution to the \etrinv formula $\Phi$, there exists a canonical placement of the previously introduced pieces $\p_{i-1}$ that encodes that solution.
The placement of the left yellow, gray, and blue piece is then fixed.
It is now clear from the construction that the remaining pieces can be placed, since the inequality of the gramophone is satisfied by the solution to $\Phi$.
\end{proof}

\subsubsection*{Fingerprinting and almost-canonical placement}
The proof of \Cref{lem:AlmostCanonicalPlacement} for the gramophone follows exactly as for the swap.
We fingerprint the pieces in the order shown in \Cref{fig:gramophone-precise} used for color codes.

\subsubsection*{Aligned placement}
From the alignment line $\ell$ (shown in \Cref{fig:gramophone-precise}), we get that the right yellow, right blue, and right gray pieces are correctly aligned.
As for the swap, to get \Cref{lem:alignedPl} for the gramophone, we can argue that the exiting blue and yellow pieces encode a value which is at most $\slack$ away from that encoded by the entering pieces.
It therefore follows that the placement must be an aligned $i\slack$-placement.

\subsubsection*{Edge inequalities}
We have an edge from the left to the right yellow piece and from the left to the right blue piece.
The edge inequalities (\Cref{lem:graphIneq}) follow similarly as for the split (yellow) and swap (blue).

\subsubsection*{Constant lane of gray pieces}
As seen from the installation manual in \Cref{fig:gramophone-entire} (left), there are four swaps in the gray lane (including the one inside the gramophone itself), so the lane consists of $5$ pieces $p_1,\ldots,p_5$ in order from left to right.
Here, $p_2$ and $p_3$ are the left and right gray piece in the gramophone itself.
The pieces $p_2,\ldots,p_5$ are introduced in swaps and in the gramophone, but $p_1$ is introduced in its own minuscule gadget shown in \Cref{fig:gramophone-entire} (right).
The canonical placement of $p_1$ is when it has edge contacts with the boundary and the yellow piece as shown, and $p_1$ is fingerprinted as indicated by the dot in the figure.
We now argue that these five pieces must have a fixed placement independent of the rest of the construction, so that they can indeed be considered as variable pieces encoding a constant.
Let us define the placement of the first piece $p_1$ shown in \Cref{fig:gramophone-entire} to encode the value $\enc{p_1}=0$.
We then know that in every placement, $0\leq \enc{p_1}$, since the piece may slide to the right, but not to the left due to the container boundary.
Similarly as for the proof of the edge inequality in \Cref{lem:graphIneq}, we get that $0\leq\enc{p_1}\leq\cdots\leq\enc{p_5}\leq 0$, where the last inequality follows since $p_5$ is bounded from the right by the container boundary.
We then have $\enc{p_1}=\cdots=\enc{p_5}=0$.

\subsubsection*{The gramophone works}
In this paragraph we prove the part of \Cref{lem:inversion} that the gramophone is responsible for, namely that in a given aligned $\gadgets\slack$-placement, a concave gramophone used for the variables $x$ and $y$ implies that $-(\enc{K_x}+1)\cdot (\enc{K_y}+1) + 1\geq 0$ while a convex gramophone implies that $(\enc{K_x}+1)\cdot (\enc{K_y}+1)-1\geq 0$.

\begin{proof}[Proof of Lemma \ref{lem:inversion} for the gramophone]
We show the statement for the concave gramophone; the proof for the convex gramophone is analogous.
The yellow pieces are part of the cycle $K_x$.
Since they encode values consistently by \Cref{lem:consistentCycle}, the rotation of the pink piece is fixed and it has edge-edge contacts with the yellow pieces.
Likewise, the blue pieces are part of $K_y$, so they also encode the same value of $y$.
Since the gray pieces encode the same constant, we now get that the orange piece has edge-edge contacts to all blue and gray pieces.
We can slide the pink piece down until it gets edge-edge contact to the orange piece.
We now consider sliding the yellow pieces to the right until the pink piece gets in contact with the curved part of $\cont$.
In this situation, we have $-(\enc{K_x}+1)\cdot (\enc{K_y}+1)+1=0$ by construction.
Since we have slid the yellow pieces to the right and thus increased the value of $\enc{K_x}$, we started with a placement where $-(\enc{K_x}+1)\cdot (\enc{K_y}+1)+1\geq 0$.
\end{proof}

%%%%%%%%%%%%%%%%%%%%%%%%%%%%%%

\section{Square container}
\label{sec:SquareContainer}
%%%%%%%%%%%%%%%%%%%%%%%%%%%%%%
Recall from \Cref{lem:4monotone} that when we reduce to problems where the container is a polygon, the resulting container is $4$-monotone (for the definition of $4$-monotone, see \Cref{def:4monotone}).
Let $\I_1$ be an instance of a packing problem where the container $\cont\mydef\cont(\I_1)$ is a 4-monotone polygon.
We show how to make a reduction from \pack\piecetype\polygon\motiontype\ to \pack\piecetype\square\motiontype.
We will use the reduction for the problem \pack\convexpolygon\polygon\rotation\ and obtain that \pack\convexpolygon\square\rotation\ is $\ER$-hard.
Likewise, we will use the reduction for the problem \pack\curved\polygon\translation\ (where we use polygonal gramophones as curvers) and get that \pack\curved\square\translation\ is $\ER$-hard.
In conclusion, we aim at proving the following theorem.

\begin{theorem}\label{thm:squarereduction}
Given an instance $\I_1$ of \pack\piecetype\polygon\motiontype\ where $\piecetype\in\{\convexpolygon,\curved\}$ and $\motiontype\in\{\translation,\rotation\}$, we can in polynomial time compute an instance $\I_2$ of \pack\piecetype\square\motiontype\ such that $\I_1$ is feasible if and only if $\I_2$ is feasible.
\end{theorem}

Most of the remaining part of this section is devoted to the proof of this theorem.
We show that the instance $\I_1$ can be reduced to an instance $\I_2$ where the container is a square $\sqcont$ with corners $b_1b_2b_3b_4$.
To this end, we introduce some auxiliary pieces that can essentially be placed in only one way in $\sqcont$.
We call these new pieces the \emph{exterior} pieces, whereas we call the pieces of~$\I_1$ the \emph{inner} pieces.
The empty space left by the exterior pieces is the $4$-monotone polygon $\cont$ which act as the container $\cont$ of $\I_1$.
We scale down the container $\cont$ and the inner pieces so that $\cont$ fits in an $\eps\times\eps$ square in the middle of $\sqcont$, for a small value $\eps=\Theta(1)$ to be defined shortly.
Since the original container $\cont$ has size $O(n^4)\times O(n)$, we scale it down by a factor of $O(n^4)$.

An example of the construction can be seen in \Cref{fig:squareRed}, where only the exterior pieces are shown.
Our construction is parameterized by a number $\eps>0$, and the container $\sqcont$ is the square $[0,1+\eps]\times [0,1+\eps]$.
The blue pieces $B_1,B_2,B_3,B_4$ are independent of $\eps$, but the green, turquoise, and orange pieces depend on $\eps$.
Only the exterior pieces are shown, but the full instance $\I_2$ does also include the inner pieces, which are intended to be packed in $\cont$.
Each green piece $G_i$ has a pair of corners $c_i,d_i$ of interior angle $\alpha_i$ such that $\alpha_i\neq\alpha_j$ for $i\neq j$.
In reality, these angles $\alpha_i$ are only slightly less than $\pi/2$, so that all the eight segments from a point $c_i$ or $d_i$ to the closest corner $b_j$ are almost equally large.
The angles $\alpha_i$ are independent of $\eps$ and the instance $\I_1$.

Each orange piece is a trapezoid, and edges of the orange pieces form the boundary of $\cont$.
At the opposite end of the pieces, they meet the green or turquoise pieces.
To make sure that they are placed in the right order, they have unique angles so that fingerprinting can be used to argue about their placement.
This will be explained in more detail in the end of this section.

We denote the placement of the exterior pieces shown in \Cref{fig:squareRed} to be \emph{canonical}.
We also define the placements we get from the figure by sliding the orange and turquoise pieces towards $\cont$ to be \emph{canonical}.
We aim at proving the following lemma.

\begin{lemma}
\label{lem:squarePack}
For a sufficiently small constant $\eps>0$, the following holds.
For all valid placements of the pieces in $\I_2$, the exterior pieces have a canonical placement.
\end{lemma}

It follows from the lemma that $\I_1$ has a solution if and only if $\I_2$ has one, so that the problems are equivalent under polynomial time reductions.

The first step in proving \Cref{lem:squarePack} is to prove that the blue and green pieces, up to a rotation, can only be placed in the canonical way, i.e., even when we disregard the turquoise and orange pieces and the inner pieces.
To prove this, we consider the situation when $\eps$ gets very small, as shown in \Cref{fig:squareRed2} (left), so that the blue and green pieces cover almost all of~$\cont$.
Then all the turquoise and orange pieces are skinny, and the $4$-monotone polygon $\cont$ is very small.
When changing $\eps$, we keep all angles constant and the blue pieces constant, and $\cont$ is scaled appropriately so that it fits in the central square of size $\eps\times\eps$ (this square is drawn with dashed segments in \Cref{fig:squareRed}).
We will eventually scale the pieces from $\I_1$ so that they fit in a square of size $\eps\times\eps$.
Since they originally fit in a container of size $O(n^4)\times O(n)$ and $\eps$ is a constant, this scaling is polynomial and thus allowed in our reduction.

\begin{figure}
\centering
\includegraphics[page = 3,width=0.9\textwidth]{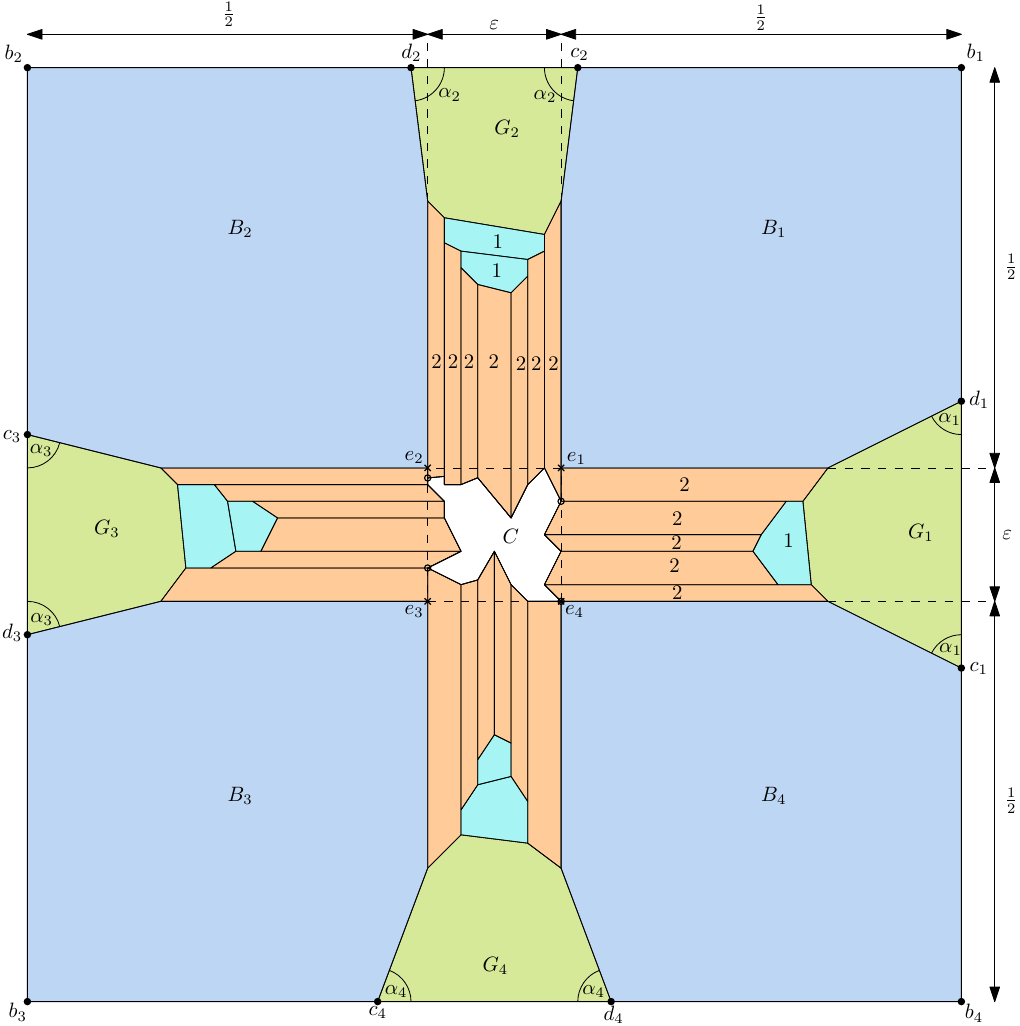}
\caption{An example of the instance we get from the reduction using a square container for a given $4$-monotone polygon $\cont$.
In reality, the angles $\alpha_i$ are just slightly less than $\pi/2$, and all the eight segments from a point $c_i$ or $d_i$ to the closest corner $b_j$ have length almost $1/2$.
Color codes:
$B_1,\ldots,B_4$ are blue, $G_1,\ldots,G_4$ are green, 1 is turquoise, 2 is orange.
}
\label{fig:squareRed}
\end{figure}

\begin{figure}
\centering
\includegraphics[page = 4]{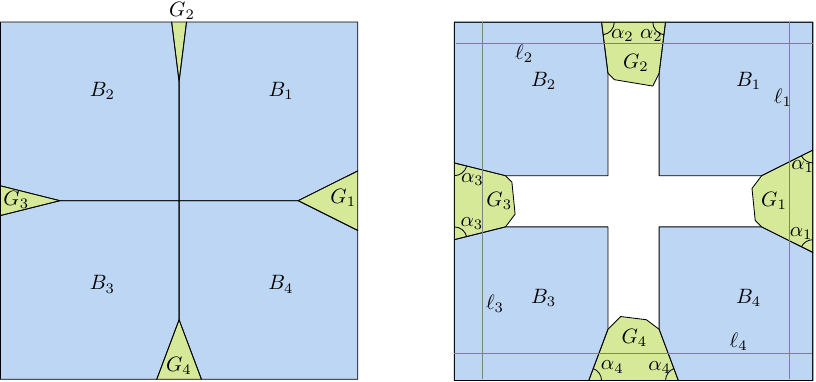}
\caption{Left: The situation from \Cref{fig:squareRed} in the limit $\eps=0$.
Right:
The canonical placement of the green and blue pieces.
For some sufficiently small $\eps$, $Q_m$ must be the shown placement.
We have four alignment segments $\ell_1,\ell_2,\ell_3,\ell_4$.}
\label{fig:squareRed2}
\end{figure}

\begin{lemma}
If $\eps>0$ is a sufficiently small constant, the canonical placement is the only way to place the blue and green pieces into $\sqcont$, i.e., even when the turquoise and orange pieces and the inner pieces do not have to be placed.
(When rotations are allowed, the three rotations of this packing by angles $\pi/2,\pi,3\pi/2$ are also possible.)
\end{lemma}

\begin{proof}
Suppose for {the purpose of} contradiction that for arbitrarily small $\eps>0$, there are other ways in which the eight pieces can be placed.
For every $m\in\N$, let $Q_m$ be such a placement for some $\eps\in (0,1/m)$.
Recall that the blue pieces are independent of $\eps$, and note that as $m\longrightarrow\infty$, the shape of the green pieces converge to the pieces shown in \Cref{fig:squareRed2} (left).
Note that the pieces are compact sets in the plane, and recall that the Hausdorff distance turns the set of non-empty compact sets into a compact metric space in its own right.
By passing to a subsequence, we may therefore assume that for each piece $p$, the placement of $p$ according to $Q_m$ is likewise converging with respect to the Hausdorff distance.

We are going to apply the Single Fingerprint \Cref{lemma:boundDiff0} to the corner $b_1$ of the square container $S$, and we want to prove that 
in the limit placement, a right corner of a blue piece coincides with $b_1$.
For the following consider the notation of \Cref{lemma:boundDiff0}.
We define the triangle~$T$ so that $x\in b_1b_4$, $y=b_1$, and $z\in b_1b_2$, 
and we set $u\mydef 0$.
Choose $\sigma$ so small that the eight blue and green pieces have the unique angle property.
We get that a blue piece must be placed such that 
one of its right corners $v$ is within distance $O\left(\sqrt{\slack/\sigma}\right)$ from $b_1$.
Now, as $m\longrightarrow\infty$, $\slack$ gets arbitrarily small, and therefore we get that $v$ coincides with $b_1$ in the limit.

We likewise get that corners with right angles of the other blue segments are placed at the other corners $b_2,b_3,b_4$ of $\sqcont$.

Conceivably, the blue pieces can be placed incorrectly in two ways:
(i) their cyclic order around the boundary of $\sqcont$ can be different from $B_1,B_2,B_3,B_4$ and (ii) one of the right corners $e_i$ which is supposed to be placed in the interior of $\sqcont$ can be placed at a corner of $\sqcont$.
Due to the difference in the angles $\alpha_1,\ldots,\alpha_4$, it is clear that in each of these cases, the green pieces cannot be placed.
Hence, the pieces must be placed as shown in \Cref{fig:squareRed2} (left) in the limit.

We conclude that by choosing $m$ sufficiently large (and thus $\eps$ sufficiently small), the difference between each piece of $Q_m$ and the packing of \Cref{fig:squareRed2} (left) can be made arbitrarily small.
We now prove that if the difference is sufficiently small, the placement $Q_m$ must be the canonical placement.
Intuitively, this means that the canonical placement is ``locked'' in the sense that it is not possible to move the pieces just a little bit and obtain another valid placement.

\begin{figure}
\centering
\includegraphics[page = 5]{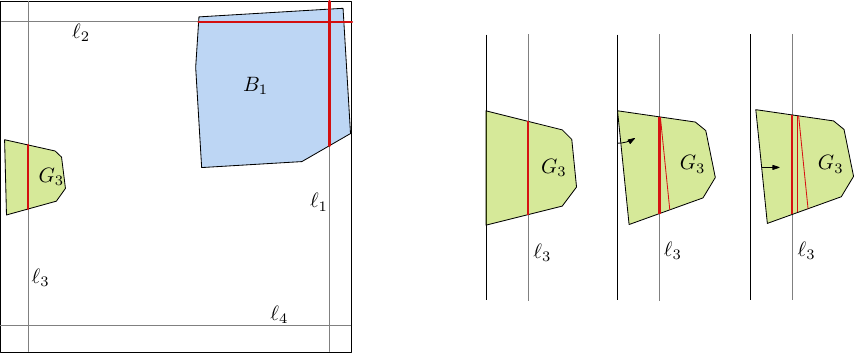}

\caption{Left: The red segments show the parts of the alignment segments that is occupied by the pieces $B_1$ and $G_3$. 
Right: The fat red segments show the parts of the alignment segment $\ell_3$ that is occupied by the piece $G_3$ in three different situations.
The thin red segments are the parts occupied in the preceding situations.
First is shown the canonical situation from \Cref{fig:squareRed2} (right).
The middle situation shows that as $G_3$ is rotated, it occupies more.
The last situation shows that when $G_3$ is translated, it occupies even more.}
\label{fig:squareRed3}
\end{figure}

In order to argue about the precise placement,
we make an alignment argument using four alignment segments $\ell_1,\ell_2,\ell_3,\ell_4$ simultaneously.
For each edge of $\sqcont$, we have an alignment segment $\ell_i$ parallel to and close to the edge, as shown in \Cref{fig:squareRed2} (right).
For each of the eight blue and green pieces $p$, we measure how much $p$ \emph{occupies} of each alignment segment, and we take the sum of all these measures and show that it is strictly minimum in the canonical placement.
In that placement, the segments are fully covered by the pieces, which means that there cannot be any other placement since they would occupy more of the segments than what is available.

Note that as we consider a placement that is close to the canonical placement, we get that the blue piece $B_1$ only intersects the segments $\ell_1$ and $\ell_2$, and the other blue pieces likewise only intersect two segments each.
Each green piece~$G_i$ only intersects $\ell_i$.
We now define the occupied parts of the alignment segments as follows, see \Cref{fig:squareRed3} for an illustration.
Each green piece~$G_i$ occupies the part $G_i\cap \ell_i$ of the segment $\ell_i$ and nothing of the other segments.
The piece blue piece $B_1$ occupies the part of $\ell_1$ from the upper endpoint of $\ell_1$ to the lowest point in $B_1\cap \ell_1$.
Similarly, $B_1$ occupies the part of $\ell_2$ from the right endpoint to the leftmost point in $B_1\cap \ell_2$.
Each of the other pieces $B_2,B_3,B_4$ occupies parts of the two segments it intersects, defined in the analogous way.
It follows that the parts of a segment $\ell_i$ occupied by two different pieces are interior-disjoint:
This is trivial for the parts covered by the green pieces, but a blue piece $B_j$ can also occupy parts that it does not cover close to the boundary of $\sqcont$.
However, these parts of the alignment segments are close to the corner $b_j$, and since the constructed placement is close to canonical, the other pieces are far away, so they cannot cover anything close to the corner $b_j$.

\begin{figure}
\centering
\includegraphics[page = 1]{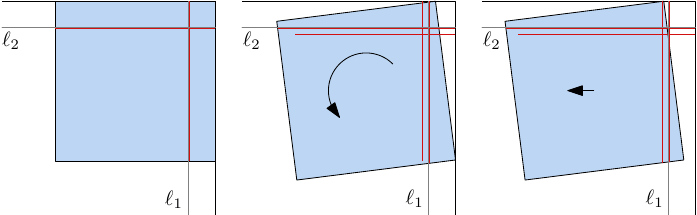}
\caption{The fat red segments show the parts of the alignment segments $\ell_1$ and $\ell_2$ that are occupied by a piece in three different situations.
The thin red segments are the parts occupied in the preceding situation.
}
\label{fig:squareRed5}
\end{figure}

Let $L$ be the sum of lengths of the occupied parts.
Since the occupied parts are interior-disjoint, we get that $L$ is at most $4\cdot (1+\eps)$, i.e., the sum of lengths of $\ell_1,\ldots,\ell_4$.
We now show that~$L$ has a strict local minimum in the canonical placement.
To this end, we observe that starting with the canonical placement, any small movement of one of the pieces makes $L$ increase.
This is seen for a green piece $G_3$ in \Cref{fig:squareRed3} (right).
For a blue piece, 
consider the blue piece in \Cref{fig:squareRed5}. 
First is shown the canonical placement.
The middle figure shows a situation where the piece is rotated slightly while keeping 
it is far up and to the right as possible.
We see that it occupies more of both alignment segments.
The last situation shows that when the piece is translated to the left, the part occupied of $\ell_2$ increases 
as much as the translation, while the part occupied of $\ell_1$ slightly decreases (\emph{slightly} 
since the piece has been only slightly rotated and the angles $\alpha_i$ are close to $\pi/2$).
In total, the piece occupies more of $\ell_1$ and $\ell_2$.
The occupied amount likewise increases when the piece is translated down.
It then follows that moving the piece from the canonical placement by any small movement (rotation and translation combined) 
will strictly increase the occupied amount of the alignment segments.

To sum up, choosing $m$ sufficiently big, the packing $Q_m$ can be made arbitrarily close to the 
canonical packing, but this implies that 
it must be identical to the packing, as it would otherwise occupy more of the alignment segments $\ell_1,\ldots,\ell_4$ that what is possible.
Hence, when $\eps$ is small enough, the eight blue and green pieces can only be placed in the canonical way.
\end{proof}

To finish the proof of \Cref{thm:squarereduction}, it remains to argue about the fixed placement of the orange and turquoise pieces.
The use of fingerprinting and aligning used here is similar to how we argued about the gadgets.
We consider the piece $G_3$ as an example, and refer to \Cref{fig:squareRed5NEW} for an illustration of the process.
We fix the pieces in layers, so that there are a constant number of pieces in each layer.
We first fingerprint the pieces touching $G_3$ in the canonical packing.
We then align them.
This may leave a bit of empty space to the right of $G_3$ and to the left of the orange and turquoise pieces, but that is allowed in the canonical placements, and the empty space will be too small for any piece to fit there.
We then fingerprint the pieces touching the first turquoise piece, and repeat.
This leads to a canonical placement of the exterior pieces, proving \Cref{lem:squarePack}.

\begin{figure}
\centering
\includegraphics[page = 2]{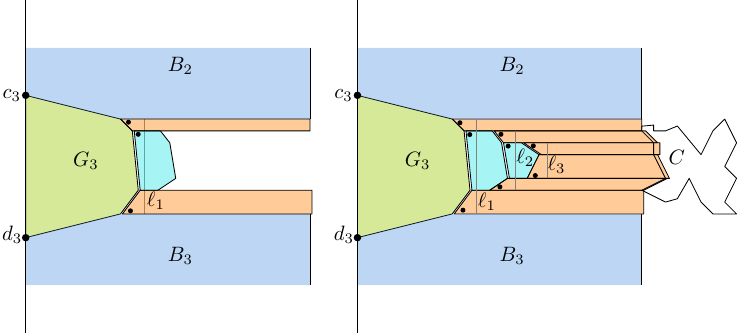}
\caption{
Left: Fixing the alignment of the first layer of pieces. The three pieces are fingerprinted and then aligned.
The empty space between $G_3$ and the three pieces is so small that no piece can fit in there.
Right: The three layers which together fix all the orange and turquoise pieces.
Note that the empty space between the exterior pieces leaves less room in $\cont$ for the inner pieces (refer to \Cref{fig:squareRed}).
}
\label{fig:squareRed5NEW}
\end{figure}

Recall that the (unscaled) instance $\I_1$ has $\slack=\Theta(n^{-296})$ and size $O(n^4)\times O(n)$.
We scale it down to be contained in a square of size $\eps\times\eps$, where $\eps=\Theta(1)$, so we get that the slack of our created instance $\I_2$ is $\slack=\Theta(n^{-304})$, which is polynomial.
The reason that we use the turquoise pieces is that if instead all orange pieces were adjacent to the green pieces, we would need to fingerprint a superconstant number of pieces at once, and then we would need $\slack$ to be smaller than polynomial, as mentioned in \Cref{sec:overview}.

As mentioned above, the empty space to the left of the orange and turquoise pieces sticking out from $G_3$ is too small that any pieces can fit there.
We can therefore without loss of generality slide the pieces to the left so that they create all the edge-edge contacts shown in \Cref{fig:squareRed}.
Similarly, those sticking our from $G_4$ can be pushed down, etc.
After pushing all the orange and turquoise pieces towards their green piece, the empty area left for the inner pieces is exactly the container $\cont$.
We conclude that there is a solution to the created instance $\I_2$ with a square container if and only if there is a solution to the given instance $\I_1$, and \Cref{thm:squarereduction} has been proven.
We then have the following corollary.

\thmSquarePack*
 
%\section*{References}
%\addcontentsline{toc}{section}{References}
%\printbibliography[heading=none]

\clearpage
 	\printbibliography

\end{document}